%% file: manuscript.tex
 \documentclass{mlcv}
 
 \usepackage[round]{natbib}
 \renewcommand{\cite}[1]{\citep{#1}}
 
 \usepackage{xcolor}         

 \usepackage{cases}
 
 \newcommand{\darkgreen}{green!60!black}
 \usetikzlibrary{calc}
 \usetikzlibrary{math}
 \usetikzlibrary{external}
 \usepackage{pgfplots}
 \usepackage{pgfplotstable}
 \pgfplotsset{compat=1.18}
 \usepgfplotslibrary{groupplots}
 \usetikzlibrary{positioning, arrows.meta, fit, graphs, graphs.standard, backgrounds,matrix}
 \usepgfplotslibrary{statistics}

 \makeatletter
 \pgfplotsset{
 	boxplot/hide outliers/.code={
 		\def\pgfplotsplothandlerboxplot@outlier{}%
 	}
 }
 \makeatother
 \pgfplotsset{
 	every axis/.style={
 		error bars/y dir=both,
 		error bars/y explicit,
 		error bars/error bar style={line width=0.5pt},
 		error bars/error mark options={rotate=90, mark size=0.4ex, line width=0.5pt}
 	},
 	boxplot/.append style={
 		/pgfplots/boxplot/hide outliers
 	}
 }
 
 \usepackage{autonum}
 
 \usepackage{environ}
 \usepackage{etoolbox}

 \newcommand{\pop}[2]{\text{POP}_{#1, #2}}
 \newcommand{\popcc}[3]{\text{POP}_{#1, #2}[#3]}

 \newcommand{\pzo}[1]{P_{01}[#1]}
 \newcommand{\poz}[1]{P_{10}[#1]}
 \newcommand{\dom}[1]{\textnormal{dom}\;#1}
 \newcommand{\tclosure}[1]{\textnormal{tcl}\,#1}
 \newcommand{\closure}[2]{\textnormal{cl}_{#1}\,#2}
 \newcommand{\pathsG}[3]{\mathcal{P}_{#1}(#2, #3)}
 \newcommand{\pathsVE}[4]{\mathcal{P}_{(#1, #2)}(#3, #4)}
 \newcommand{\pfs}[1]{X_{#1}}
 \newcommand{\pfsc}[2]{X_{#1}[#2]}
 \newcommand{\ppc}[1]{\hat{X}_{#1}}
 
 \newcommand{\ice}{\textnormal{ive}}
 \newcommand{\ici}{\textnormal{ivi}}
 \newcommand{\lb}{\textnormal{lb}}
 \newcommand{\ub}{\textnormal{ub}}
 \newcommand{\pairs}[1]{P_{#1}}
 \newcommand{\triples}[1]{T_{#1}}
 
 \DeclareMathOperator*{\argmax}{argmax}

 \newif\ifproofsinappendix
 \proofsinappendixtrue
 
 \newcommand{\myappendixproofs}{}
 
 \NewEnviron{proofatend}[1][Proof]{%
 	\ifproofsinappendix
 	\global\edef\myappendixproofs{%
 		\unexpanded\expandafter{\myappendixproofs}%
 		\unexpanded{%
 			\subsection*{#1}%
 			\begin{proof}
 			}%
 			\unexpanded\expandafter{\BODY}%
 			\unexpanded{\end{proof}}%
 	}
 	\else
 	\begin{proof}
 		\BODY
 	\end{proof}
 	\fi
 }
 
 \newcommand{\makeproofs}{%
 	\ifproofsinappendix
 	\section{Proofs}\label{appendix:proofs}
 	For any set $V$ and any $E' \subseteq \pairs{V}$, abbreviate the transitive closure of $E'$ by 
 	\begin{align}
 		\tclosure{E'} \coloneqq \{pq\in \pairs{V} \mid \pathsVE{V}{E'}{p}{q} \neq \emptyset\}\enspace.
 	\end{align}
 	\myappendixproofs
 	\fi
 }
 
 \makeatletter
 \pgfplotsset{
 	boxplot prepared from table/.code={
 		\def\tikz@plot@handler{\pgfplotsplothandlerboxplotprepared}%
 		\pgfplotsset{
 			/pgfplots/boxplot prepared from table/.cd,
 			#1,
 		}
 	},
 	/pgfplots/boxplot prepared from table/.cd,
 	table/.code={
 		\pgfplotstableread[col sep=comma]{#1}\boxplot@datatable
 	},
 	row/.initial=0,
 	make style readable from table/.style={
 		#1/.code={
 			\pgfplotstablegetelem{\pgfkeysvalueof{/pgfplots/boxplot prepared from table/row}}{##1}\of\boxplot@datatable
 			\pgfplotsset{boxplot/#1/.expand once={\pgfplotsretval}}
 		}
 	},
 	make style readable from table=lower whisker,
 	make style readable from table=upper whisker,
 	make style readable from table=lower quartile,
 	make style readable from table=upper quartile,
 	make style readable from table=median,
 	make style readable from table=lower notch,
 	make style readable from table=upper notch
 }
 \makeatother
 
 \newlength{\hspacing}
 \newlength{\vspacing}

\Crefname{section}{Sec.}{Secs.}
\Crefname{appendix}{App.}{Apps.}
\Crefname{definition}{Def.}{Defs.}
\Crefname{theorem}{Thm.}{Thms.}
\Crefname{proposition}{Prop.}{Props.}
\Crefname{corollary}{Cor.}{Cors.}
\Crefname{figure}{Fig.}{Figs.}

\begin{document}

\author{David Stein$^1$ $\cdot$ Jannik Irmai$^1$ $\cdot$ Bjoern Andres$^1$}
\title{Partial Optimality in the Preordering Problem}
\date{Machine Learning for Computer Vision\\TU Dresden}
\maketitle

\begin{abstract}
	Preordering is a generalization of clustering and partial ordering with applications in  bioinformatics and social network analysis. 
	Given a finite set $V$ and a value $c_{ab} \in \mathbb{R}$ for every ordered pair $ab$ of elements of $V$, the preordering problem asks for a preorder $\lesssim$ on $V$ that maximizes the sum of the values of those pairs $ab$ for which $a \lesssim b$.
	Building on the state of the art in solving this NP-hard problem partially, we contribute new partial optimality conditions and efficient algorithms for deciding these conditions.
	In experiments with real and synthetic data, these new conditions increase, in particular, the fraction of pairs $ab$ for which it is decided efficiently that $a \not\lesssim b$ in an optimal preorder.
\end{abstract}

\tableofcontents
\footnotetext[1]{\href{mailto:david.stein1@tu-dresden.de,jannik.irmai@tu-dresden.de,bjoern.andres@tu-dresden.de}{david.stein1@tu-dresden.de, jannik.irmai@tu-dresden.de, bjoern.andres@tu-dresden.de}}

\input{section-introduction}

\input{section-related-work}
\input{section-partial-preorders}
\input{section-maps}
\input{section-conditions}
\input{section-cut-conditions}
\input{section-join-conditions}
\input{section-fixation-conditions}
\input{section-algorithms}

\input{section-experiments}
\input{section-conclusion}

\bibliographystyle{plainnat}
\bibliography{manuscript}

\appendix
\makeproofs
\input{appendix-equivalence-classes}
\input{appendix-energy-minimization}
\input{appendix-alpha-beta-swap-moves}
\input{appendix-google+}
\input{appendix-boecker-condition}

\end{document}

%% file: section-introduction.tex
\section{Introduction}

A \emph{preorder} is a binary relation on a set that is reflexive and transitive.
Given a finite set $V$ and a value $c_{ab} \in \mathbb{R}$ for every ordered pair $ab \in V^2$ such that $a \neq b$, the preordering problem asks for a preorder $\lesssim$ on $V$ that maximizes the total value $\sum_{a \lesssim b, a \neq b} c_{ab}$.
See \Cref{figure:introduction} for an example.
The preordering problem has applications in bioinformatics \cite{jacob2008,boecker2009} and social network analysis \cite{irmai2025}.
It is closely related to clustering and ordering:
In clustering, the feasible solutions can be identified with equivalence relaxations, i.e.~symmetric preorders.
In partial ordering, the feasible solutions are anti-symmetric preorders.
A preorder need neither be symmetric nor anti-symmetric.
Yet, the symmetric subset of a preorder defines a clustering, and the anti-symmetric subset well-defines a partial order on the clusters (\Cref{figure:introduction}).
Thus, preordering is both a hybrid and a joint relaxation of clustering and partial ordering \citep{irmai2025}.
At the same time, preordering remains NP-hard, even for values in $\{-1,1\}$ \citep{weller2012}.

We attempt to solve the preordering problem partially by deciding efficiently for some pairs $ab \in V^2$ with $a \neq b$ whether $a \lesssim b$ in an optimal preorder. 
Using a technique introduced by \citet{shekhovtsov-2013,shekhovtsov-2014} and applied, e.g., by \citet{shekhovtsov-2015}, we prove new partial optimality conditions for the preordering problem that can be verified by establishing that certain maps from the feasible set to itself are improving (as made rigorous in \Cref{definition:improving-map} below).
For these decision problems, we propose efficient algorithms.
To examine the effectiveness and efficiency of these algorithms and thus the effectiveness of the underlying partial optimality conditions, we conduct experiments with real and synthetic data.
Compared to the state of the art in solving the preordering problem partially, the new conditions increase, in particular, the fraction of pairs $ab$ for which it is decided efficiently that $a \not\lesssim b$ in an optimal preorder.
\ifproofsinappendix
All proofs are deferred to~\Cref{appendix:proofs}.
\fi

\input{figure-illustration-preorder-maximal-specificy-combined}

\subsection{Problem statement}

For any set $V$, let $\pairs{V} = \{pq \in V^2 \mid p \neq q\}$ and $\triples{V} = \{pqr \in V^3 \mid p \neq q \wedge p \neq r \wedge q \neq r\}$ be the sets of pairs and triples of distinct elements of $V$.

\begin{definition}[\citet{wakabayashi1998}]
\label{definition:problem}
For any finite set $V$ and any $c \in \mathbb{R}^{\pairs{V}}$, let 
$X_V \coloneqq \{x\in \{0, 1\}^{P_V} \mid \forall pqr\in T_V \colon 	x_{pq} + x_{qr} - x_{pr} \leq 1 \}$.
For any $\varphi_c \colon X_V \to \mathbb{R} \colon x \mapsto \sum_{pq \in {P_V}} c_{pq} \, x_{pq}$,
call
$\max \ \{\varphi_c(x) \mid x \in X_V \}$
the instance of the \emph{(maximum value) preordering problem} with respect to $V$ and $c$ ($\pop{V}{c}$).
\end{definition}

For any $x \in X_V$, $x^{-1}(1)$ is a transitive relation on $V$. 
This relation is also irreflexive, thus a transitive digraph over the node set $V$.
For this reason, $\pop{V}{c}$ is also called the \emph{transitive digraph problem} \citep[cf.][]{weller2012}.
Here, we extend $x$ and $c$ such that $\forall i \in V$: $x_{ii} := 1  \land c_{ii} := 0$.
This extension is a value-preserving bijection from the characteristic functions of transitive digraphs to the characteristic functions of preorders.
It justifies the term \emph{preordering problem} for $\pop{V}{c}$.

%% file: figure-illustration-preorder-maximal-specificy-combined.tex
\begin{figure}[t]
	\centering
	\begin{minipage}[b]{0.525\linewidth}
		\centering 
		\small
		\begin{minipage}{0.55\linewidth}
			\centering
			\begin{tabular}{l|r@{\;\;}r@{\;\;}r@{\;\;}r@{\;\;}r@{\;\;}r}
				$c$ & $i$ & $j$ & $k$ & $l$ & $m$   \\
				\midrule
				$i$ &       & $2$  & $-1$ & $-1$ & $-1$ \\
				$j$ & $2$   &      & $-1$ & $2$  & $-1$ \\
				$k$ & $-4$  & $-4$ &      & $3$  &  $2$ \\
				$l$ & $1$   & $1$  & $-1$ &    & $-1$ \\
				$m$ & $-1$  & $-1$ & $1$  & $-2$ & \\
			\end{tabular}
		\end{minipage}
		\hfill
		\begin{minipage}{0.4\linewidth}
			\centering
			\begin{tikzpicture}[
				yscale=-1,
				vertex/.style={circle, draw, inner sep=0.4ex, fill=white},
				]
				\draw[draw=none, fill=black!15] (0, 0) ellipse (1 and 0.4);
				\draw[draw=none, fill=black!15] (1.5, 0) ellipse (0.35 and 0.35);
				\draw[draw=none, fill=black!15] (1.5, -1.5) ellipse (0.35 and 0.35);
				\draw[draw=none, fill=black!15] (0, -1.5) ellipse (0.35 and 0.35);
				\node[vertex, label=left:$i$] (0) at (-0.5, 0) {};
				\node[vertex, label=right:$j$] (1) at (0.5, 0) {};
				\node[vertex, label=left:$l$] (3) at (0, -1.5) {};
				\node[vertex, label=right:$k$] (2) at(1.5, 0) {};
				\node[vertex, label=right:$m$] (4) at (1.5, -1.5) {};
				\draw[-latex]  (2) edge (3);
				\draw[-latex] (0) edge[bend left] (1);
				\draw[-latex] (1) edge[bend left] (0);
				\draw[-latex] (0) edge (3);
				\draw[-latex] (1) edge (3);
				\draw[-latex] (2) edge (4);
			\end{tikzpicture}
		\end{minipage}
		\caption{An instance of the preordering problem with the set $V = \{i,j,k,l,m\}$ and the values $c$ is defined by the table on the left.
			An optimal solution, i.e.~a maximum value preorder $\lesssim$, is shown on the right with arrows from $a$ to $b$ indicating $a \lesssim b$.
			The symmetric subset of the preorder is an equivalence relation and thus defines a partition, or clustering, of $V$ (gray).
			The anti-symmmetric subset well-defines a partial order over the clusters.}
		\label{figure:introduction}
	\end{minipage}%
	\hfill
	\begin{minipage}[b]{0.425\linewidth}
		\centering
		\begin{tikzpicture}[yscale=-1, xscale=1]
			\node[draw=black, circle,inner sep=0.4ex] (1) at (0, 0) {};
			\node[draw=black, circle,inner sep=0.4ex] (2) at (1, 0) {};
			\node[draw=black, circle,inner sep=0.4ex] (3) at (2, 0) {};
			
			\node[draw=black, circle,inner sep=0.4ex] (4) at (0, 1) {};
			\node[draw=black, circle,inner sep=0.4ex] (5) at (1, 1) {};
			\node[draw=black, circle,inner sep=0.4ex] (6) at (2, 1) {};
			
			\draw[-latex] (1) -- (2);
			\draw[-latex] (2) -- (3);
			\draw[-latex, gray] (1) edge[bend right] (3);
			
			\draw[-latex] (4) -- (5);
			\draw[-latex] (5) -- (6);
			\draw[-latex, gray] (4) edge[bend left] (6);

			\draw[-latex, dashed] (5) -- (2);
			\draw[-latex, gray, dashed] (6) -- (1);
			\draw[-latex, gray, dashed] (6) -- (2);
			\draw[-latex, gray, dashed] (5) -- (1);
			
		\end{tikzpicture}
		\caption{In the example depicted above, must-join constraints (solid black arrows) and must-cut constraints (dashed black arrow) defined by a partial characteristic function of preorders, $\tilde{x}\in \tilde{X}_V$, imply additional must-join constraints (solid gray arrows) as well as additional must-cut constraints (dashed gray arrows).
			In this case, the closure $\closure{V}{\tilde{x}}$ is strictly more specific than $\tilde{x}$.}
		\label{figure:illustration-maximal-specificity}
	\end{minipage}
\end{figure}

%% file: section-related-work.tex
\section{Related Work}

The preordering problem in the form of \Cref{definition:problem} is stated by \citet{wakabayashi1998} who establishes NP-hardness and discusses specializations with additional constraints.
The state of the art in solving this problem partially is by the conditions, bounds and algorithms of \citet{boecker2009}.
Here, we prove their partial optimality conditions via improving maps~\citep{shekhovtsov-2013,shekhovtsov-2014} and establish a more general condition that subsumes these.

More specific is the transitivity editing problem that asks for the minimum number of edges to insert into or remove from a given digraph in order to make it transitive.
This problem is isomorphic to the preordering problem with values in $\{-1,1\}$ and remains NP-hard \citep{weller2012}.
Partial optimality conditions for the transitivity editing problem are established and analyzed also by \citet{weller2012}, including conditions transferred from \citet{boecker2009} and additional conditions specific to this problem.

More restrictive than the preordering problem are the clique partition problem \cite{groetschel-1984} and the correlation clustering problem \cite{bansal-2004} whose feasible solutions are equivalence relations.
Partial optimality conditions for these problems established and studied by \citet{alush2012,lange2018,lange2019,stein2023}.
Some partial optimality conditions we define relate to theirs, and we point out these relations below.
To describe clusterings constrained by partial optimality, \citet{hornakova2017} introduce the notion of maximally specific partial characterizations. 
We adapt their definition to preorders.

Also more restrictive than the preordering problem is the partial ordering problem \citep{muller1996partial} for which partial optimality has not been discussed specifically, to the best of our knowledge.
More constrained still is the linear ordering problem \cite{marti-2011} for which partial optimality conditions are established by \citet{stein2024}.
The conditions we define do not relate directly to theirs because the totality constraint changes the problem fundamentally.

A branch-and-cut algorithm for the preordering problem in the framework of \citet{gurobi} is defined by \citet{irmai2025}.
Local search algorithms for the preordering problem are defined by \citet{boecker2009,irmai2025}.
The geometry of the preorder polytope is studied by \citet{gurgel1992,gurgel1997adjacency}. 

%% file: section-partial-preorders.tex
\section{Constraint Propagation}

Given an instance of the preordering problem $\pop{V}{c}$, our goal is to decide efficiently for some $ab \in P_V$ that $x_{ab} = 0$, and for some $ab \in P_V$ that $x_{ab} = 1$, in an optimal solution $x$.
We collect such constraints in a \emph{partial function}, i.e., a function from some $P\subseteq P_V$ to $\{0, 1\}$. 
With minor abuse of notation, let $\{0, 1, *\}^{P_V} \coloneqq \bigcup_{P\subseteq P_V}\{0, 1\}^{P}$ denote the set of all such functions. 
For any $\tilde{x} \in \{0, 1, *\}^{P_V}$, let $\dom{\tilde{x}}:= \tilde{x}^{-1}(\{0, 1\})$ denote the domain of $\tilde{x}$.
For any $a \in V$: $aa \notin \dom \tilde x$.
Yet, we define $\tilde{x}_{aa} := 1$, by convention.
Once we have established partial optimality $\tilde{x} \in \{0, 1, *\}^{P_V}$ and look for more, we are no longer concerned with the feasible set $X_V$ but instead with the set of those $x \in X_V$ that agree with $\tilde x$ on the domain of $\tilde x$:

\begin{definition}
	Let $V \neq \emptyset$ finite and $\tilde{x}\in \{0, 1, *\}^{\pairs{V}}$. 
	The elements of 
		$X_{V}[\tilde{x}] \coloneqq \{x\in X_V \mid \forall pq\in \dom{\tilde x}\colon x_{pq} = \tilde{x}_{pq}\}$
	are called the \emph{completions} of $\tilde{x}$ in $X_V$. 
	Moreover, $\tilde{x}$ is called a \emph{partial characteristic function of preorders on $V$} if and only if
		$X_{V}[\tilde{x}] \neq \emptyset$.
	For the set of all partial characteristic functions of preorders on $V$, we write $\tilde{X}_V \coloneq \{\tilde{x}\in \{0, 1, *\}^{\pairs{V}} \mid X_V[\tilde{x}] \neq \emptyset\}$.
\end{definition}%

\begin{proposition}
	\label{proposition:consistency-in-p}
	For any finite $V \neq \emptyset$ and any $\tilde{x}\in \{0, 1, *\}^{\pairs{V}}$, we can decide efficiently if $X_V[\tilde{x}] \neq \emptyset$.
\end{proposition}%
\begin{proofatend}[Proof of~\Cref{proposition:consistency-in-p}]
	Let $\tilde{x}\in \{0, 1, *\}^{\pairs{V}}$. 
	We show that $\tilde{x}$ is consistent if and only if $\tilde{x}^{-1}(0) \cap \tclosure{\tilde{x}^{-1}(1)} = \emptyset$. Since $\tclosure{\tilde{x}^{-1}(1)} $ can be computed in polynomial time, the claim follows. 
	
	(1) Let $\tilde{x}^{-1}(0) \cap \tclosure{\tilde{x}^{-1}(1)}= \emptyset$. We define $x\in \{0, 1\}^{\pairs{V}}$ such that $x^{-1}(1) = \tclosure{\tilde{x}^{-1}(1)}$. 
	Obviously $x\in X_V$. 
	Moreover, $\tilde{x}^{-1}(1) \subseteq \tclosure{\tilde{x}^{-1}(1)} = x^{-1}(1)$. 
	Furthermore, $\tilde{x}^{-1}(0) \setminus x^{-1}(0) = \tilde{x}^{-1}(0) \cap x^{-1}(1) = \tilde{x}^{-1}(0) \cap \tclosure{\tilde{x}^{-1}(1)}= \emptyset$ by assumption and thus $\tilde{x}^{-1}(0) \subseteq x^{-1}(0)$. Therefore, $x\in X_V[\tilde{x}]$. 
	 
	(2) Let $\tilde{x}^{-1}(0) \cap \tclosure{\tilde{x}^{-1}(1)} \neq \emptyset$. 
	Consider some fixed $pq\in \tilde{x}^{-1}(0) \cap \tclosure{\tilde{x}^{-1}(1)}$. 
	We assume that there exists $x\in X_V[\tilde{x}]$. 
	Then it follows by transitivity of $x$ that $\tclosure{\tilde{x}^{-1}(1)} \subseteq x^{-1}(1)$. 
	In particular, we have that $1 = x_{pq} \neq \tilde{x}_{pq} = 0$, contradicting the assumption. 
	Therefore:
	$X_V[\tilde{x}] = \emptyset$. 
\end{proofatend}%

For any partial characteristic function of preorders on $V$, $\tilde{x}$, pairs $ab \in P_V$ may exist for which $ab \notin \dom \tilde x$ and yet, there exists $\beta \in \{0,1\}$ such that every completion $x$ of $\tilde x$ is such that $x_{ab} = \beta$.
See \Cref{figure:illustration-maximal-specificity} for an example.
We seek to add all such pairs to $\dom \tilde x$, i.e., to propagate constraints.
Hence, we characterize these pairs and show how they can be found efficiently:

\begin{definition}
	Let $V\neq \emptyset$ finite. For any $\tilde{x}\in \tilde{X}_V$, the pairs in
		$\pairs{V}[\tilde{x}] \coloneq \{pq\in \pairs{V} \mid \forall x, x'\in X_V[\tilde{x}]\colon x_{pq} = x'_{pq}\}$
	are called \emph{decided}.  
	The pairs $\pairs{V} \setminus \pairs{V}[\tilde{x}]$ are called \emph{undecided}. Moreover, $\tilde{x}$ is called \emph{maximally specific} if and only if 
		$\pairs{V}[\tilde{x}] = \dom{\tilde{x}}\enspace$.
\end{definition}
\begin{proposition}
	\label{proposition:maximal-specificity-in-p}
	Maximal specificity is efficiently decidable.
	Moreover, for any $\tilde x \in \tilde X_V$, the maximally specific partial characteristic function of preorders on $V$ that has the same completions as $\tilde{x}$ is unique, called the \emph{closure} $\closure{V}{\tilde{x}}$ of $\tilde x$, and is such that
	$
	(\closure{V}{\tilde{x}})^{-1}(1) = \{pq\in \pairs{V} \mid \pathsVE{V}{\tilde{x}^{-1}(1)}{p}{q} \neq \emptyset\} 
	$
	and
	$(\closure{V}{\tilde{x}})^{-1}(0) = \{pq\in \pairs{V} \mid \exists p'q'\in \tilde{x}^{-1}(0) \colon \pathsVE{V}{\tilde{x}^{-1}(1)}{p'}{p} \neq \emptyset 
	\land 
	\pathsVE{V}{\tilde{x}^{-1}(1)}{q}{q'} \neq \emptyset
	\}
	$.
\end{proposition}%
\begin{proofatend}[Proof of~\Cref{proposition:maximal-specificity-in-p}]
	Let $\tilde{x}\in \{0, 1, *\}^{\pairs{V}}$. 
	We show that $\tilde{x}$ is maximally specific if and only if the following two conditions hold:
	\begin{align}
		\label{eq:closure-one-edges}
		\tilde{x}^{-1}(1) = &\{pq\in \pairs{V} \mid \pathsVE{V}{\tilde{x}^{-1}(1)}{p}{q} \neq \emptyset\}\\
		\label{eq:closure-zero-edges}
		\tilde{x}^{-1}(0) 
		=&\left\{pq\in \pairs{V} \mid 
		\exists p'q'\in \tilde{x}^{-1}(0) \colon \pathsVE{V}{\tilde{x}^{-1}(1)}{p'}{p} \neq \emptyset \land 
		\pathsVE{V}{\tilde{x}^{-1}(1)}{q}{q'} \neq \emptyset
		\right\}
	\end{align}
	Firstly, we assume that \eqref{eq:closure-one-edges} and \eqref{eq:closure-zero-edges} are true and show that $\tilde{x}$ is maximally specific. $\dom{\tilde{x}} \subseteq \pairs{V}[\tilde{x}]$ is trivially true. 
	It suffices to show that for every $pq\in \pairs{V}\setminus \dom{\tilde{x}}$ there exist  completions $x', x'' \in X_V[\tilde{x}]$ such that $x'_{pq} = 0$ and $x''_{pq} = 1$. From this it follows that every $pq\in \pairs{V}\setminus \dom{\tilde{x}}$ is undecided and thus $\pairs{V}[\tilde{x}] \subseteq \dom{\tilde{x}}$. 
	Let $pq\in P_V\setminus \dom{\tilde{x}}$ be fixed.
	Let $x'\in \{0, 1\}^{\pairs{V}}$ such that $\left(x'\right)^{-1}(1) = \tilde{x}^{-1}(1)$. 
	We have $x'\in X_V$ since $\left(x'\right)^{-1}(1) = \tclosure{\left(x'\right)^{-1}(1)}$ by \eqref{eq:closure-one-edges}. 
	Moreover, $x' \in X_V[\tilde{x}]$ and $x'_{pq} = 1$ by definition of $x'$.
	Let $x''\in \{0, 1\}^{\pairs{V}}$ such that 
	\begin{align}
		&(x'')^{-1}(1) \coloneqq \tilde{x}^{-1}(1) \cup \left\{p'q'\in \pairs{V} \mid \pathsVE{V}{\tilde{x}^{-1}(1)}{p'}{p} \neq \emptyset \land \pathsVE{V}{\tilde{x}^{-1}(1)}{q}{q'} \neq \emptyset\right\}\enspace.
	\end{align}%
	We have $x''\in X_V$ since $(x'')^{-1}(1) = \tclosure{(x'')^{-1}(1)}$.
	Moreover, we have $\tilde{x}^{-1}(1) \subseteq (x'')^{-1}(1)$ by definition of $x''$. Furthermore, we show that $\tilde{x}^{-1}(0) \subseteq (x'')^{-1}(0)$. 
	In order to prove this, it suffices to show that $(x'')^{-1}(1) \cap \tilde{x}^{-1}(0) = \emptyset$. 
	For the sake of contradiction, let us assume that $p'q'\in (x'')^{-1}(1) \cap \tilde{x}^{-1}(0)$. 
	By definition of $x''$ and the fact that $\tilde{x}^{-1}(1) \cap \tilde{x}^{-1}(0) = \emptyset$, we have $\pathsVE{V}{\tilde{x}^{-1}(1)}{p'}{p} \neq \emptyset$ and $\pathsVE{V}{\tilde{x}^{-1}(1)}{q}{q'} \neq \emptyset$. 
	By \eqref{eq:closure-zero-edges} and $p'q'\in \tilde{x}^{-1}(0)$, we have $pq\in \tilde{x}^{-1}(0)$, in contradiction to $pq\in P_V\setminus \dom{\tilde{x}}$.
	Therefore, $(x'')^{-1}(1) \cap \tilde{x}^{-1}(0) = \emptyset$. 
	Thus, $x'' \in X_V[\tilde{x}]$, i.e., $x''$ is a proper completion of $\tilde{x}\in \tilde{X}_V$. Moreover, $x''_{pq} = 1$.
	
	Secondly, we assume that \eqref{eq:closure-one-edges} and \eqref{eq:closure-zero-edges} are not both true and show that $\tilde{x}$ is not maximally specific. In particular, we show that there exists a decided pair that is not included in the domain of $\tilde{x}$, i.e.,  $pq\in \pairs{V}[\tilde{x}]\setminus \dom{\tilde{x}}$. 
	If \eqref{eq:closure-one-edges} is false, then there exists some pair $pq\in \pairs{V} \setminus \dom{\tilde{x}}$ such that $\pathsVE{V}{\tilde{x}^{-1}(1)}{p}{q}\neq \emptyset$. 
	For every solution $x\in X_V[\tilde{x}]$ this implies $x_{pq} = 1$ by transitivity and therefore $pq\in \pairs{V}[\tilde{x}] \setminus \dom{\tilde{x}}$. 
	If \eqref{eq:closure-zero-edges} is false, then there exist a pair $pq\in \pairs{V}\setminus \dom{\tilde{x}}$ and a pair $p'q'\in \tilde{x}^{-1}(0)$ with $\pathsVE{V}{\tilde{x}^{-1}(1)}{p'}{p} \neq \emptyset$ and $\pathsVE{V}{\tilde{x}^{-1}(1)}{q}{q'} \neq \emptyset$. 
	For every solution $x\in X_V[\tilde{x}]$ this implies $x_{pq} = 0$. 
	In order to show this, let us assume that $x_{pq} = 1$. 
	Then there is a $p'q'$-path in $(V, x^{-1}(1))$ and thus $x_{p'q'} = 1$, in contradiction to $p'q'\in \tilde{x}^{-1}(0)$. 
	Therefore, the pair $pq$ is decided to be cut, i.e., $pq\in \pairs{V}[\tilde{x}] \setminus \dom{\tilde{x}}$. 
	As conditions~\eqref{eq:closure-one-edges} and~\eqref{eq:closure-zero-edges} can be decided efficiently, we conclude that maximum specificity can be decided efficiently.
	
	Next, let $\tilde{x}\in \tilde{X}_V$ be fixed. We show that $\closure{V}{\tilde{x}}$ is maximally specific. Assume that $\closure{V}{\tilde{x}}$ is not maximally specific. 
	We consider two cases: 
	
	Firstly, assume $\closure{V}{\tilde{x}}$ does not fulfill~\eqref{eq:closure-one-edges}. Thus, there is $pq\in P_V\setminus \left(\closure{V}{\tilde{x}}\right)^{-1}(1)$ such that $\pathsVE{V}{\left(\closure{V}{\tilde{x}}\right)^{-1}(1)}{p}{q} \neq \emptyset$, i.e., $\closure{V}{\tilde{x}}$ does not satisfy~\eqref{eq:closure-one-edges}. Therefore, it follows that $\pathsVE{V}{\tilde{x}^{-1}(1)}{p}{q} = \emptyset$ by definition of $\closure{V}{\tilde{x}}$. 
	We have that $\pathsVE{V}{\left(\closure{V}{\tilde{x}}\right)^{-1}(1)}{p}{q} \neq \emptyset$ implies $\pathsVE{V}{\tilde{x}^{-1}(1)}{p}{q} \neq \emptyset$, which is thus a contradiction. To see the latter argument let $(V', E') \in \pathsVE{V}{\left(\closure{V}{\tilde{x}}\right)^{-1}(1)}{p}{q} \neq \emptyset$. 
	Notice that any $p'q'\in E'$ is such that $\pathsVE{V}{\tilde{x}^{-1}(1)}{p'}{q'} \neq \emptyset$ by definition of $\closure{V}{\tilde{x}}$. 
	Therefore, we can replace any edge $p'q'\in E'$ in the path $(V', E')$ by a $p'q'$-path in $(V, \tilde{x}^{-1}(1))$. 
	Thus $\pathsVE{V}{\tilde{x}^{-1}(1)}{p}{q} \neq \emptyset$. 
	Overall, we have derived a contradiction.
	Therefore, $\closure{V}{\tilde{x}}$ fulfills~\eqref{eq:closure-one-edges} for every $\tilde{x}\in \tilde{X}_V$.
	
	Secondly, assume $\closure{V}{\tilde{x}}$ does not fulfill~\eqref{eq:closure-zero-edges}. 
	Thus, there is $pq\in P_V\setminus \left(\closure{V}{\tilde{x}}\right)^{-1}(0)$ such that there exists $p'q'\in \left(\closure{V}{\tilde{x}}\right)^{-1}(0)$ with $\pathsVE{V}{\left(\closure{V}{\tilde{x}}\right)^{-1}(1)}{p'}{p} \neq \emptyset$ and $\pathsVE{V}{\left(\closure{V}{\tilde{x}}\right)^{-1}(1)}{q}{q'} \neq \emptyset$. 
	By definition of $\closure{V}{\tilde{x}}$ it follows that there exists some $p''q''\in \tilde{x}^{-1}(0)$ such that $\pathsVE{V}{\tilde{x}^{-1}(1)}{p''}{p'}\neq \emptyset$ and $\pathsVE{V}{\tilde{x}^{-1}(1)}{q'}{q''}\neq \emptyset$. 
	Moreover, it follows that 
	$\pathsVE{V}{\tilde{x}^{-1}(1)}{p'}{p} \neq \emptyset$ and $\pathsVE{V}{\tilde{x}^{-1}(1)}{q}{q'} \neq \emptyset$. 
	Overall, we have $p''q''\in \tilde{x}^{-1}(0)$ and $\pathsVE{V}{\tilde{x}^{-1}(1)}{p''}{p} \neq \emptyset$ and $\pathsVE{V}{\tilde{x}^{-1}(1)}{q}{q''} \neq \emptyset$. 
	Thus, $pq\in \left(\closure{V}{\tilde{x}}\right)^{-1}(0)$, which contradicts the assumption.  
	
	We have now proven that $\closure{V}{\tilde{x}}$ is maximally specific. 
	
	Lastly, we show that the maximally specific partial characterization that has the same completions as $\tilde{x}'$ is unique. 
	For that purpose, let us define $\tilde{X}_V[\tilde{x}] = \{\tilde{x}'\in \tilde{X}_V\mid X_V[\tilde{x}] = X_V[\tilde{x}']\}$. 
	Now, assume that there are two maximally specific $\hat{x}', \hat{x}''\in \tilde{X}_V[\tilde{x}]$ such that $\tilde{x}' \neq \tilde{x}''$. 
	By definition of $\tilde{X}_V[\tilde{x}]$ we have $\pfsc{V}{\tilde{x}'} = \pfsc{V}{\tilde{x}''}$. 
	We first show that $\dom{\tilde{x}'} = \dom{\tilde{x}''}$. 
	For the sake of contradiction, assume $pq\in \dom{\tilde{x}'}\setminus \dom{\tilde{x}''}$. 
	As both $\tilde{x}'$ and $\tilde{x}''$ are maximally specific by assumption, we conclude $pq\in P_V[\tilde{x}']\setminus P_V[\tilde{x}'']$, i.e., $pq$ is decided by $\tilde{x}'$ but not by $\tilde{x}''$. 
	This contradicts $\pfsc{V}{\tilde{x}'} = \pfsc{V}{\tilde{x}''}$. Thus, $\dom{\tilde{x}'}\subseteq \dom{\tilde{x}''}$. 
	Similarly, we can show $\dom{\tilde{x}''}\subseteq \dom{\tilde{x}'}$.  
	Therefore, we conclude $\dom{\tilde{x}'} = \dom{\tilde{x}''}$. 
	Now, we show that $\tilde{x}' = \tilde{x}''$. 
	From the assumption that $\tilde{x}'_e \neq \tilde{x}''_e$ for some $e\in \dom{\tilde{x}'} = \dom{\tilde{x}''}$, it directly follows that $\pfsc{V}{\tilde{x}'} \neq \pfsc{V}{\tilde{x}''}$, which is a contradiction. 
	Hence $\tilde{x}' = \tilde{x}''$. 
	Thus, the maximally specific partial characterization of preorders on $V$ that has the same completions as $\tilde{x}$ is unique.
	By the previous statements, it is given by $\closure{V}{\tilde{x}}$.
\end{proofatend}%

With this established, we can consider the set of all maximally specific partial characteristic functions of preorders on $V$, $\hat{X}_V \coloneqq \{\tilde{x}\in \tilde{X}_{V}\mid \tilde{x} = \closure{V}{\tilde{x}}\}$.
For any $\hat{x}\in \hat{X}_V$, we can consider the preordering problem constrained to the feasible set $X_V[\hat{x}]$, for which we write $\popcc{V}{c}{\hat{x}}$. 

%% file: section-maps.tex
\section{Improving Maps}

We fix some notation and recall the notion of improving maps:
For any set $V$ and any $U \subseteq V$, let $\delta(U, V\setminus U) := U \times \left(V\setminus U\right) $ and $\delta(U) := \delta(U, V\setminus U) \cup \delta(V\setminus U, U)$. 
For any digraph $G = (V, E)$ and any $p,q \in V$, let $\pathsG{G}{p}{q}$ denote the set of all directed paths in $G$ from $p$ to $q$. 
\ifproofsinappendix
\else
For any set $V$ and any $E' \subseteq \pairs{V}$, abbreviate the transitive closure of $E'$ by 
\begin{align}
	\tclosure{E'} \coloneqq \{pq\in \pairs{V} \mid \pathsVE{V}{E'}{p}{q} \neq \emptyset\}\enspace.
\end{align}
\fi

\begin{definition}
\label{definition:improving-map}
	Let $X\neq \emptyset$ finite, $\varphi\colon X\to \mathbb{R}$ and $\sigma\colon X\to X$. 
	If $\varphi(\sigma(x))\geq \varphi(x)$ for all $x\in X$, then $\sigma$ is called \emph{improving} for the problem $\max_{x\in X}\varphi(x)$. 
\end{definition}
\begin{proposition}
	\label{proposition:partial-optimality}
	Let $X\neq \emptyset$ finite, $\varphi\colon X \to \mathbb{R}$,
	$\sigma\colon X\to X$ improving for $\max_{x\in X}\varphi(x)$, and
	$X'\subseteq X$. 
	If $\sigma(X) \subseteq X'$ there is an optimal solution $x^*$ to $\max_{x\in X}\varphi(x)$ such that $x^* \in X'$.
\end{proposition}
\begin{proofatend}[Proof of~\Cref{proposition:partial-optimality}]
	Let $x^*\in X$ be an optimal solution to $\max_{x\in X}\varphi(x)$ such that $x^*\notin X'$. Then, $\sigma(x^*)$ is also optimal since $\varphi(\sigma(x^*)) \geq \varphi(x^*)$. Thus, $\varphi(\sigma(x^*)) = \varphi(x^*)$, by improvingness of $\sigma$. 
	Moreover, $\sigma(x^*) \in X'$. 
	This concludes the proof.
\end{proofatend}

\begin{corollary}
	\label{corollary:partial-optimality-single-variable}
	Let $S\neq \emptyset$ finite, $X\subseteq \{0, 1\}^S$ and $\varphi\colon X\to \mathbb{R}$. 
	Let $\sigma \colon X\to X$ be improving for $\max_{x\in X}\varphi(x)$. 
	Let $s\in S$ and $\beta\in \{0, 1\}$. 
	If $\sigma(x)_s = \beta$ for all $x\in X$, there is an optimal solution $x^*$ to $\max_{x\in X}\varphi(x)$ such that $x^*_s = \beta$.
\end{corollary}

In \Cref{section:conditions}, we will construct maps $\sigma \colon X_V \to X_V$ from any feasible solution of $\pop{V}{c}$ to another and establish conditions under which such maps are improving.
These maps will be composed of two maps we introduce here:

\begin{definition}
	\label{definition:elementary-cut-map}
	Let $V\neq \emptyset$ finite and $U\subseteq V$. 
	We define the \emph{elementary dicut map} $\sigma_{\delta(U, V\setminus U)}\colon X_V \to X_V$ such that for all $x\in X_V$ and all $pq\in \pairs{V}$:
	\begin{equation}
		\sigma_{\delta(U, V\setminus U)}(x)_{pq} = \begin{cases}
			0 & \textnormal{if $pq\in \delta(U, V\setminus U)$} \\
			x_{pq} & \textnormal{otherwise}
		\end{cases}\enspace.
	\end{equation}
\end{definition}
\begin{lemma}
	\label{lemma:elementary-cut-map-well-definedness}
	The elementary dicut map is well-defined, i.e., $\forall V\neq \emptyset\forall U \subseteq V \colon \sigma_{\delta(U, V\setminus U)}(X_V) \subseteq X_V$.
\end{lemma}
\begin{proofatend}[Proof of~\Cref{lemma:elementary-cut-map-well-definedness}]
	Let $x\in X_V$ and $x' = \sigma_{\delta(U, V\setminus U)}(x)$. Assume we have $pqr\in \triples{V}$ such that
	\begin{equation}
		x'_{pq} + x'_{qr} - x'_{pr} > 1\enspace.
	\end{equation}
	Then, $x'_{pq} = x'_{qr} = 1$ and $x'_{pr} = 0$. 
	It follows that $x_{pq} = x_{qr} = 1$ and $pq, qr\notin\delta(U, V\setminus U)$ from \Cref{definition:elementary-cut-map}. 
	If $pr\notin \delta(U, V\setminus U)$ it follows that $x_{pr} = 0$ from \Cref{definition:elementary-cut-map}. 
	In this case, $x\notin X_V$, which is a contradiction. 
	If $pr\in \delta(U, V\setminus U)$, then $p\in U$ and $r\in V\setminus U$. From $pq\notin \delta(U, V\setminus U)$ it follows that $q\in U$, and from $qr\notin \delta(U, V\setminus U)$ it follows that $q\in V\setminus U$, which is a contradiction.
\end{proofatend}

\begin{definition}
	\label{definition:elementary-join-map}
	Let $V\neq \emptyset$ finite and $ij\in \pairs{V}$. 
	We define the \emph{elementary join map} $\sigma_{ij}\colon X_V \to X_V$ such that for all $x\in X_V$ and all $pq\in \pairs{V}$:
	 \begin{numcases}{\sigma_{ij}(x)_{pq} = }
		\label{eq:elementary-join-map-path}
		1 & \textnormal{if $x_{pi} = x_{jq} = 1$}\\
		x_{pq} & \textnormal{otherwise}	\enspace .
	\end{numcases}
\end{definition}

\begin{lemma}
	\label{lemma:elementary-join-map-well-definedness}
	The elementary join map is well-defined, i.e., $\forall V\neq \emptyset \forall ij\in \pairs{V}\colon \sigma_{ij}(X_V) \subseteq X_V$.
\end{lemma}
\begin{proofatend}[Proof of~\Cref{lemma:elementary-join-map-well-definedness}]
	Let $x\in X_V$ and $x' = \sigma_{ij}(x)$. Assume that $x'\notin X_V$, i.e., we have $pqr\in \triples{V}$ such that
	\begin{equation}
		x'_{pq} + x'_{qr} - x'_{pr} > 1\enspace.
	\end{equation}
	Then, $x'_{pq} = x'_{qr} = 1$ and $x'_{pr} = 0$. From~\Cref{definition:elementary-join-map} it follows that $x_{pr} = 0$, and 
	$x_{pi} = 0$ or $x_{jr} = 0$.
	For $x_{pq}$ and $x_{qr}$, we consider the four remaining cases. 
	
	(1) Let $x_{pq} = x_{qr} = 1$ and $x_{pr} = 0$.
	In this case, $x\notin X_V$. 
	
	(2) Let $x_{pq} = 1$, $x_{qr} = 0$ and $x_{pr} = 0$. 
	Because $x'_{qr} = 1$, we have $x_{qi} = x_{jr} = 1$. 
	From this and $x_{pq} = 1$ follows $x_{pi} = x_{jr} = 1$, which is a contradiction.
	
	(3) Let $x_{pq} = 0$, $x_{qr} = 1$ and $x_{pr} = 0$. 
	Because $x'_{pq} = 1$, we have $x_{pi} = x_{jq} = 1$. 
	From this and $x_{qr} = 1$ follows $x_{pi} = x_{jr} = 1$, which is a contradiction.
	
	(4) Let $x_{pq} = x_{qr} = 0$ and $x_{pr} = 0$. 
	Because $x'_{pq} = x'_{qr} = 1$, we have $x_{pi} = x_{jq} = 1$ and $x_{qi} = x_{jr} = 1$. 
	From this follows $x_{pi} = x_{jr} = 1$, which is a contradiction.
	
	This concludes the proof.
\end{proofatend}

Once we have established some partial optimality in terms of a maximally specific partial characteristic function $\hat x \in \hat X_V$, we need to focus on maps $\sigma$ that maintain this partial optimality:
\begin{definition}
For any $V \neq\emptyset$ finite, any $\hat{x}\in \ppc{V}$ and any $\sigma\colon \pfs{V} \to \pfs{V}$, $\sigma$ is called \emph{true to $\hat{x}$} if $\sigma(\pfsc{V}{\hat{x}})\subseteq \pfsc{V}{\hat{x}}$. 
\end{definition}%

In order to establish that a map $\sigma$ is true to a partial characteristic function $\hat x \in \hat X_V$, we consider the pairs whose label can change from $i \in \{0, 1\}$ to $j\in \{0, 1\}$ under $\sigma$ (\Cref{def:pzo-poz}). 
More specifically, we will apply~\Cref{lemma:pzo-poz-supersets}.
\begin{definition}
\label{def:pzo-poz}
For any $V\neq \emptyset$ finite, 
any $\hat{x}\in \hat{X}_V$,
any $i, j \in \{0, 1\}$ and 
any $\sigma\colon X_V[\hat x] \to X_V$, let
$P_{ij}[\sigma] := \{e\in P_V\mid \exists x\in X_V[\hat x] \colon x_{e} = i \land \sigma(x)_{e} = j\}$.
\end{definition}
\begin{lemma}
	\label{lemma:pzo-poz-supersets}
Let $V\neq \emptyset$ finite, $\hat{x}\in \hat{X}_V$ and $\sigma\colon \pfsc{V}{\hat{x}}\to \pfs{V}$.
If there exist $P_{01}'\supseteq \pzo{\sigma}$ and $P_{10}' \supseteq \poz{\sigma}$ such that $P_{10}'\cap \hat{x}^{-1}(1) = P_{01}'\cap \hat{x}^{-1}(0) = \emptyset$, then $\sigma$ is true to $\hat{x}$.
\end{lemma}
\begin{proofatend}[Proof of~\Cref{lemma:pzo-poz-supersets}]
	From $P_{10}'\cap \hat{x}^{-1}(1) = P_{01}'\cap \hat{x}^{-1}(0) = \emptyset$ follows $\poz{\sigma}\cap \hat{x}^{-1}(1) = \pzo{\sigma}\cap \hat{x}^{-1}(0) = \emptyset$. 
	Thus, for every $e\in \hat{x}^{-1}(1)$ we have $e\notin \poz{\sigma}$, and therefore 
	for every $x\in \pfsc{V}{\hat{x}}$ we have $\sigma(x)_e = 0$ by definition of~$\poz{\sigma}$ and $x_e = \hat{x}_e = 1$.
	Moreover, for every $e\in \hat{x}^{-1}(0)$ we have $e\notin \pzo{\sigma}$, and therefore for every $x\in \pfsc{V}{\hat{x}}$ we have $\sigma(x)_e = 1$ by definition of~$\pzo{\sigma}$ and $x_e = \hat{x}_e = 0$.
\end{proofatend}

Another idea we will use in \Cref{section:conditions} is to apply a map $\sigma$ only if $x_{ij} = b$ for some $ij \in P_V, b \in \{0,1\}$:
\begin{definition}
For any $V\neq \emptyset$ finite, 
any $\hat{x}\in \hat{X}_V$, 
any $ij\in P_V\setminus \dom{\hat{x}}$,
any $b\in \{0, 1\}$ and 
any $\sigma\colon X_V[\hat x] \to X_V$, 
let $\sigma^{ij|b}\colon \pfsc{V}{\hat{x}}\to \pfs{V}$ such that for all $x \in \pfsc{V}{\hat{x}}$:
$\sigma^{ij | b}(x) = x$ if $x_{ij} = b$, and $\sigma^{ij | b}(x) = \sigma(x)$ otherwise.
\end{definition}

%% file: section-conditions.tex
\section{Partial Optimality Conditions}
\label{section:conditions}

%% file: section-cut-conditions.tex
\subsection{Cut Conditions}
To begin with, consider $ij \in \pairs{V}$ and $U\subseteq V$ such that $ij\in \delta(U, V\setminus U)$.
We establish a sufficient condition for the map \smash{$\sigma_{\delta(U, V\setminus U)}^{ij | 0}$} to be improving.
In this case, there exists an optimal solution $x$ such that $x_{ij} = 0$ (\Cref{proposition:edge-cut-condition}).
In order to maintain partial optimality established previously in the form of some \smash{$\hat x \in \hat X_V$}, we discuss trueness (\Cref{lemma:cut-map-trueness}).
\begin{lemma}
	\label{lemma:cut-map-trueness}
	Let $V\neq \emptyset$ finite, $\hat{x}\in \ppc{V}$ and $U \subseteq V$.
	$\sigma_{\delta(U, V\setminus U)}$ is true to $\hat{x}$ if $\delta(U, V\setminus U) \cap \hat{x}^{-1}(1) = \emptyset$.
\end{lemma}
\begin{proofatend}[Proof of~\Cref{lemma:cut-map-trueness}]
	Let $x\in \pfsc{V}{\hat{x}}$ and $x' = \sigma_{\delta(U, V\setminus U)}(x)$.
	Since $x' \leq x$, we have $\hat{x}^{-1}(0) \subseteq x^{-1}(0) \subseteq (x')^{-1}(0)$. 
	For every $pq\in \hat{x}^{-1}(1)$ we have $pq\notin \delta(U, V\setminus U)$ by assumption, and thus, $x'_{pq} = x_{pq} = 1$ by~\Cref{definition:elementary-cut-map}. Therefore, $\hat{x}^{-1}(1) \subseteq (x')^{-1}(1)$.
\end{proofatend}

\begin{theorem}
	\label{proposition:edge-cut-condition}
	Let $V\neq \emptyset$ finite, $c\in \mathbb{R}^{\pairs{V}}$ and $\hat{x}\in \hat{X}_V$. Moreover, let $U \subseteq V$ and $ij\in \delta(U, V\setminus U)\setminus \dom{\hat{x}}$. 
	If $\delta(U, V\setminus U) \cap \hat{x}^{-1}(1) = \emptyset$ and \eqref{eq:improvingness-conditional-cut-map-constrained} holds,
	there is a solution $x^*$ to $\popcc{V}{c}{\hat{x}}$ with $x^*_{ij} = 0$.
		\begin{equation}
		\label{eq:improvingness-conditional-cut-map-constrained}
		c_{ij}^- \geq \sum_{pq\in \delta(U, V\setminus U)\setminus \hat{x}^{-1}(0)}c_{pq}^+
	\end{equation}
\end{theorem}
\begin{proofatend}[Proof of~\Cref{proposition:edge-cut-condition}]
	Let $x\in X_V[\hat{x}]$ and $x' = \sigma_{\delta(U, V\setminus U)}^{ij | 0}(x)$. 
	Now, $x'\in X_V[\hat{x}]$ by~\Cref{lemma:cut-map-trueness} and $\delta(U, V\setminus U) \cap \hat{x}^{-1}(1) = \emptyset$.
	
	(1) If $x_{ij} = 0$, we have $\varphi_c(x') = \varphi_c(x)$. 
	
	(2) If $x_{ij} = 1$, then
	\begin{align}
		&\varphi_c(x') - \varphi_c(x)  = \sum_{pq\in \delta(U, V\setminus U)\setminus \hat{x}^{-1}(0)}c_{pq}\left(0 - x_{pq} \right) 
		= - c_{ij} - \sum_{\substack{pq\in \delta(U, V\setminus U) \setminus \hat{x}^{-1}(0) \\ pq \neq ij}}c_{pq}x_{pq}\\
		\geq &-c_{ij} - \sum_{\substack{pq\in \delta(U, V\setminus U)\setminus \hat{x}^{-1}(0) \\ pq \neq ij}}c^+_{pq} 
		= c^-_{ij} - \sum_{pq\in \delta(U, V\setminus U)\setminus \hat{x}^{-1}(0)} c^+_{pq} \overset{\eqref{eq:improvingness-conditional-cut-map-constrained}}{\geq} 0
	\end{align}
	
	Therefore, the map $\sigma^{ij | 0}_{\delta(U, V\setminus U)}$ is improving. Moreover, $x'_{ij} = 0$. By~\Cref{corollary:partial-optimality-single-variable} it follows that there exists an optimal solution $x^*$ to $\popcc{V}{c}{\hat{x}}$ such that $x^*_{ij} = 0$.
\end{proofatend}

\Cref{proposition:edge-cut-condition} is related to Thm.~1 of~\citet{lange2019} but considers directed edges and dicuts instead of undirected edges and cuts.
Moreover, it incorporates already-established partial optimality.
More efficiently decidable than the condition of \Cref{proposition:edge-cut-condition} is the condition of the following specialization:

\begin{corollary}
	\label{corollary:directed-cut-condition}
	Let $V\neq \emptyset$ finite, $c\in \mathbb{R}^{\pairs{V}}$, $\hat{x}\in \ppc{V}$, $U \subseteq V$ and $\hat{x}\in \hat{X}_V$ such that $\hat{x}^{-1}(1)\cap \delta(U, V\setminus U) = \emptyset$. 
	If for all $ij\in \delta(U, V\setminus U) \setminus \hat{x}^{-1}(0)$ we have $0 \geq c_{ij} $ then
	there is a solution $x^*$ to $\popcc{V}{c}{\hat{x}}$ such that $x^*_{ij} = 0$ for all $ij\in \delta(U, V\setminus U)\setminus \hat{x}^{-1}(0)$.
\end{corollary}
\begin{proofatend}[Proof of~\Cref{corollary:directed-cut-condition}]
	For every $ij\in \delta(U, V\setminus U)\setminus \hat{x}^{-1}(0)$, we have $c_{ij}^+ = 0$. 
	Thus, \Cref{proposition:edge-cut-condition} is fulfilled for every $ij\in \delta(U, V\setminus U)$ as 
	\begin{equation}
		c_{ij}^- \geq \sum_{pq\in \delta(U, V\setminus U) \setminus \hat{x}^{-1}(0)}c_{pq}^+  = 0\enspace.
	\end{equation}
	By \Cref{proposition:edge-cut-condition}, there exists an optimal solution $x^*\in X_V[\hat{x}]$ such that $x^*_{ij} = 0$. 
	By iteration over all $ij\in \delta(U, V\setminus U) \setminus \hat{x}^{-1}(0)$, the claim follows.
\end{proofatend}

\Cref{corollary:directed-cut-condition} relates to Thm.~1 of~\citet{alush2012}. 
But unlike in the case of clustering where setting all edges in a cut to zero leads to a decomposition of the problem, here, in the case of preordering, setting all directed edges in a dicut to zero does not necessarily result in a decomposition.

%% file: section-join-conditions.tex
\subsection{Join Condition}
Next, consider $U, U' \subseteq V$ disjoint, $i\in U$ and $j \in U'$. 
We establish a sufficient condition for the map $\gamma^{ij | 1}$ with $\gamma := \sigma_{ij} \circ \sigma_{\delta(V\setminus U, U)} \circ \sigma_{\delta(U', V\setminus U')}$ to be improving.
In this case, there exists an optimal solution with $x_{ij} = 1$ for some $ij\in P_V$ (\Cref{proposition:edge-join-condition}).
The map $\gamma$ is illustrated in \Cref{figure:illustration-join-map}.
In order to maintain partial optimality established previously in the form of some $\hat x \in \hat X_V$, we discuss trueness of $\gamma^{ij | 1}$ to $\hat x$ (\Cref{lemma:join-map-pzo-poz,lemma:join-map-trueness}).

\input{figure-visualization-join-map-example-combined}

\begin{lemma}
\label{lemma:join-map-pzo-poz}
For any $V\neq \emptyset$ finite, 
any $U, U' \subseteq V$ disjoint,
$U'' := V \setminus \left(U\cup U'\right)$,
any $i\in U$, $j\in U'$, 
any $\hat{x}\in \hat{X}_V$ and
$P_{01}' \coloneqq \{pq\in U \times U'\mid \hat{x}_{pi} \neq 0 \neq \hat{x}_{jq}\} \setminus \hat{x}^{-1}(1)$ and
$P_{10}' \coloneqq \left(\left(U'' \times U\right) \cup \left(U'\times U \right) \cup \left(U'\times U''\right)\right) \setminus \hat{x}^{-1}(0)$,
we have $\pzo{\gamma^{ij | 1}}\subseteq P_{01}'$ and $\poz{\gamma^{ij | 1}}\subseteq P_{10}'$.
\end{lemma}
\begin{proofatend}[Proof of~\Cref{lemma:join-map-pzo-poz}]
Let $x\in X_V[\hat{x}]$, and $\bar{x} = \gamma^{ij | 1}(x)$. 
If $x_{ij} = 1$, we have $x' = x$. 
Now, we consider the case $x_{ij} = 0$. 
Let $x'' = \sigma_{\delta(V\setminus U, U)} \circ \sigma_{\delta(U', V\setminus U')}(x)$. 

(1) Let $pq\in P_U \cup P_{U'} \cup P_{U''} \cup \left(U\times U''\right) \cup \left(U''\times U'\right)$. 
We show that $\bar{x}_{pq} = x_{pq}$ and thus $pq\notin \pzo{\gamma^{ij | 1}}$ and $pq\notin \poz{\gamma^{ij | 1}}$.
We have $x''_{pq} = x_{pq}$ and $p\in V\setminus U$ or $q\in V\setminus U'$.
Thus, $pi \in \delta(V\setminus U, U)$ or $jq \in \delta(U', V\setminus U')$. 
Firstly, we consider the case $1 = x_{pq} = x''_{pq}$. Then $\bar{x}_{pq} = 1$, as $\bar{x} \geq x''$ by~\Cref{definition:elementary-join-map}.
Secondly, we consider the case $0 = x_{pq} = x''_{pq}$. 
Since $pi \in \delta(V\setminus U, U)$ or $jq \in \delta(U', V\setminus U')$ it follows that $x''_{pi} = 0$ or $x''_{jq} = 0$, respectively. Therefore, it follows that $\bar{x}_{pq} = 0$ by~\Cref{definition:elementary-join-map}. 
Thus $pq\notin \pzo{\gamma^{ij|1}} \cup \poz{\gamma^{ij|1}}$.

(2) Let $pq\in U \times U'$. 

First, we show $pq\notin \poz{\gamma^{ij | 1}}$. 
We have $x''_{pq} = x_{pq}$. 
If $1 = x_{pq} = x''_{pq}$, then $\bar{x}_{pq} = 1$, as $\bar{x} \geq x''$. 
Thus, $pq\notin \poz{\gamma^{ij | 1}}$. 

Second, we show: If $\hat{x}_{pi} = 0$ or $\hat{x}_{jq} = 0$, then $pq\notin \pzo{\gamma^{ij | 1}}$. 
It follows that $x_{pi} = 0$ or $x_{jq} = 0$. 
We conclude that $x''_{pi} = 0$ or $x''_{jq} = 0$ by~\Cref{definition:elementary-cut-map}. Therefore, $x_{pq} = 0$ implies $\bar{x}_{pq} = x''_{pq} = 0$ by~\Cref{definition:elementary-join-map}. 
Thus $pq\notin \pzo{\gamma^{ij | 1}}$ if $\hat{x}_{pi} = 0$ or $\hat{x}_{jq} = 0$. 

(3) Let $pq\in \left(U''\times U\right) \cup \left(U'\times U\right) \cup \left(U'\times U''\right)$. 
This implies $pi\in \delta(V\setminus U, U)$ and $jq\in \delta(U', V\setminus U')$ and therefore $x''_{pi} = 0$ and $x''_{jq} = 0$. 
We have $x''_{pq} = 0$ since $pq\in \delta(V\setminus U, U)\cup \delta(U', V\setminus U')$. Moreover, $\bar{x}_{pq} = x''_{pq} = 0$ by~\Cref{definition:elementary-join-map} and since $x''_{pi} = x''_{jq} = 0$.
Therefore, $pq\notin \pzo{\gamma^{ij | 1}}$. 

(4) Obviously, if $pq\in \hat{x}^{-1}(0)$, then $pq\notin \poz{\gamma^{ij|1}}$, and if $pq\in \hat{x}^{-1}(1)$, then $pq\notin \pzo{\gamma^{ij | 1}}$. 

This concludes the proof.
\end{proofatend}

\begin{corollary}
\label{lemma:join-map-trueness}
For any $V\neq \emptyset$ finite, 
any $U, U' \subseteq V$ disjoint,
$U'' := V\setminus \left(U\cup U'\right)$,
any $i\in U$, 
any $j\in U'$, 
any $\hat{x}\in \hat{X}_V$ 
and $P_{01}'$ and $P_{10}'$ as in~\Cref{lemma:join-map-pzo-poz},
$\gamma^{ij | 1}$ is true to $\hat{x}$ if 
$\hat{x}^{-1}(0)\cap P_{01}' = \emptyset$ and 
$\hat{x}^{-1}(1) \cap P_{10}' = \emptyset$.
\end{corollary}

\begin{theorem}
	\label{proposition:edge-join-condition}
	Let $V\neq \emptyset$ finite and $c \in \mathbb{R}^{\pairs{V}}$.
	Let $U, U'\subseteq V$ disjoint and $U'' := V\setminus \left(U \cup U'\right)$. 
	Moreover, let $i \in U$, $j\in U'$, $\hat{x}\in \hat{X}_V$ and $P_{01}'$ and $P_{10}'$ as in~\Cref{lemma:join-map-pzo-poz}. 
	If $\gamma^{ij | 1}$ is true to $\hat{x}$ and~\eqref{eq:improvingness-join-generalization-map-1} holds then there is a solution $x^*$ to $\popcc{V}{c}{\hat{x}}$ such that $x^*_{ij} = 1$. 
	\begin{align}
		\label{eq:improvingness-join-generalization-map-1}
		c_{ij}^+ \geq 
		&\sum_{pq\in P_{01}'} c_{pq}^- 
		+ 
		\sum_{pq\in P_{10}'} c_{pq}^+
	\end{align}
\end{theorem}
\begin{proofatend}[Proof of~\Cref{proposition:edge-join-condition}]
	Let $x\in X_V[\hat{x}]$ and $x' = \gamma^{ij | 1}(x)$.  
	
	\noindent
	(1) Let $x_{ij} = 1$. Then $x' = x$ and $\varphi_c(x') = \varphi_c(x)$.
	
	\noindent
	(2) Let $x_{ij} = 0$. We have $\poz{\gamma^{ij | 1}} \cap \pzo{\gamma^{ij | 1}} \subseteq P_{01}'\cap P_{10}' = \emptyset$ by~\Cref{lemma:join-map-pzo-poz}, $x'_{pq} = x_{pq}$ for every $pq\notin P_{01}'\cup P_{10}'$ and $ij\in P_{01}'$. Now:
	\begin{align}
		&\varphi_c(x') - \varphi_c(x) 
		= \sum_{pq\in P_{01}'\setminus \{ij\}}c_{pq}\left(x'_{pq} - x_{pq}\right) + c_{ij} + \sum_{pq\in P_{10}'}c_{pq}\left(x'_{pq} - x_{pq}\right) \\
		\geq & -\sum_{pq\in P_{01}'\setminus \{ij\}}c_{pq}^- + c_{ij} - \sum_{pq\in P_{10}'}c_{pq}^+
		= -\sum_{pq\in P_{01}'}c_{pq}^- +c_{ij}^+ - \sum_{pq\in P_{10}'}c_{pq}^+
		\overset{\eqref{eq:improvingness-join-generalization-map-1}}{\geq}\;0 
	\end{align}
	Thus, the map $\gamma^{ij | 1}$ is improving for $\popcc{V}{c}{\hat{x}}$. 
	Since $\gamma^{ij | 1}(x)_{ij} = 1$ for all $x\in X_V[\hat{x}]$, there is an optimal solution $x^*$ to $\popcc{V}{c}{\hat{x}}$ such that $x^*_{ij} = 1$. 
\end{proofatend}

\Cref{proposition:edge-cut-condition,proposition:edge-join-condition} are complementary, witnessed by~\Cref{example:reinforcement}.

\begin{example}
	\label{example:reinforcement}
	
	Consider the instance of the preordering problem with $V = \{p, q, r\}$ and such that $c_{pq} = +1$, $c_{qp} = -2$, $c_{pr} = -2$, $c_{rp} = +2$, $c_{qr} = +2$ and $c_{rq} = +1$. 
	In addition, consider the application of partial optimality conditions according to~\Cref{table:examples}. 
	The table shows that \Cref{proposition:edge-cut-condition} can reinforce~\Cref{proposition:edge-join-condition}.
	
	\begin{table}[h]
		\centering
		\caption{Below we report the partial optimality obtained by applying~ \Cref{proposition:edge-cut-condition} and  \Cref{proposition:edge-join-condition} either in isolation (Rows 1 and 2) or jointly (Row 3).}
		\label{table:examples}
		\begin{tabular}{l r r}
			\toprule
			Partial Optimality Conditions & $\hat{x}^{-1}(0)$ & $\hat{x}^{-1}(1)$ \\
			\midrule
			\Cref{proposition:edge-cut-condition} & $\{pr, qp\}$ & $\emptyset$ \\
			 \Cref{proposition:edge-join-condition} & $\emptyset$ & $\{rp\}$ \\
			\Cref{proposition:edge-cut-condition} + \Cref{proposition:edge-join-condition} & $\{pr, qp, qr\}$ & $\{rp, rq\}$ \\
			\bottomrule
		\end{tabular}
	\end{table}
\end{example}

%% file: figure-visualization-join-map-example-combined.tex
\begin{figure}
	\begin{minipage}[t]{0.55\linewidth}\small
		\centering
		\begin{tikzpicture}[scale=1.5,scale=0.8]
			\draw[gray] (0, 0) -- (30:1.0);
			\draw[gray] (0, 0) -- (150:1.0);
			\draw[gray] (0, 0) -- (270:1.0);
			
			\node at (90:0.52) {$V \!\setminus\! (U\cup U')$};	
			\node at ($(210:0.5) +  (-0.5, 0)$) {$U$};
			\node at ($(330:0.5) + (0.5, 0)$) {$U'$};
			
			\node[label=left:$i$,circle,draw=black, inner sep=0.4ex] (i) at ($(210:0.5) +  (-0.15, -0.25)$) {};
			\node[label=right:$j$,circle,draw=black, inner sep=0.4ex] (j) at ($(330:0.5) +  (0.15, -0.25)$) {};
			
			\draw[-latex, dashed] (i) -- (j);
			
			\node[circle, draw=black, inner sep=0.4ex] (p) at ($(i) + (0, 0.25)$) {};
			\node[circle, draw=black, inner sep=0.4ex] (q) at ($(j) + (0, 0.25)$) {};
			
			\node[circle, draw=black, inner sep=0.4ex] (p2) at ($(i) + (0, -0.35)$) {};
			\node[circle, draw=black, inner sep=0.4ex] (q2) at ($(j) + (0, -0.35)$) {};
			
			\node[circle, draw=black, inner sep=0.4ex] (r) at (0, 0.25) {};
			
			\draw[-latex] (q) edge[bend right] (p);
			\draw[-latex, dashed] (p2) edge (q2);
			\draw[-latex] (p2) edge (i);
			\draw[-latex] (j) edge (q2);
			\draw[-latex] (q) edge[bend right] (r);
			\draw[-latex] (r) edge[bend right] (p);
			
			\node at (1.5, -0.25) {$\mapsto$};
			
			\begin{scope}[xshift=3.0cm]
				\draw[gray] (0, 0) -- (30:1.0);
				\draw[gray] (0, 0) -- (150:1.0);
				\draw[gray] (0, 0) -- (270:1.0);
				
				\node at (90:0.52) {$V \!\setminus\! (U\cup U')$};	
				\node at ($(210:0.5) +  (-0.5, 0)$) {$U$};
				\node at ($(330:0.5) + (0.5, 0)$) {$U'$};
				
				\node[label=left:$i$,circle,draw=black, inner sep=0.4ex] (i) at ($(210:0.5) +  (-0.15, -0.25)$) {};
				\node[label=right:$j$,circle,draw=black, inner sep=0.4ex] (j) at ($(330:0.5) +  (0.15, -0.25)$) {};
				
				\draw[-latex] (i) -- (j);
				
				\node[circle, draw=black, inner sep=0.4ex] (p) at ($(i) + (0, 0.25)$) {};
				\node[circle, draw=black, inner sep=0.4ex] (q) at ($(j) + (0, 0.25)$) {};
				
				\node[circle, draw=black, inner sep=0.4ex] (p2) at ($(i) + (0, -0.35)$) {};
				\node[circle, draw=black, inner sep=0.4ex] (q2) at ($(j) + (0, -0.35)$) {};
				
				\node[circle, draw=black, inner sep=0.4ex] (r) at (0, 0.25) {};
				
				\draw[-latex, dashed] (q) edge[bend right] (p);
				\draw[-latex] (p2) edge (q2);
				\draw[-latex] (p2) edge (i);
				\draw[-latex] (j) edge (q2);
				\draw[-latex, dashed] (q) edge[bend right] (r);
				\draw[-latex, dashed] (r) edge[bend right] (p);
				
			\end{scope}

		\end{tikzpicture}
		\\[-1ex]
		\caption{The map $\sigma_{ij} \circ \sigma_{\delta(V\setminus U, U)} \circ \sigma_{\delta(U', V\setminus U')}$ transforms the preorder depicted on the left to the preorder depicted on the right.}
		\label{figure:illustration-join-map}
	\end{minipage}%
	\hfill
	\begin{minipage}[t]{0.4\linewidth}\small
		\centering
		\begin{tikzpicture}[
			xscale=1.5,
			yscale=-1.5,
			vertex/.style={circle, draw, inner sep=0.4ex, fill=white}
			]
			\node[vertex,label=left:$p$] (1) at (0, 0) {};
			\node[vertex,label=right:$q$] (2) at (1, 0) {};
			\node[vertex,label=left:$r$] (3) at (0, 1) {};
			\node[vertex,label=right:$s$] (4) at (1, 1) {};
			\draw[-latex] (1) -- node[below] {+5} (2);
			\draw[-latex] (1) -- node[left] {+1} (3);
			\draw[-latex] (1) -- node[left] {+1} (4);
			\draw[-latex] (2) -- node[right] {+1} (3);
			\draw[-latex] (2) -- node[right] {+1} (4);
			\draw[-latex] (3) -- node[above] {+2} (4);
		\end{tikzpicture}
		\\[-1ex]
		\caption{Depicted above is an instance of the preordering problem witnessing
			that \Cref{proposition:subset-U-maximizer} is not subsumed by~\Cref{corollary:boecker-condition-weak}.}
		\label{figure:example-preorder-join-subset-not-subsumed}
	\end{minipage}
\end{figure}

%% file: section-fixation-conditions.tex
\subsection{Fixation Conditions}
We now introduce a more complex condition that can fix variables to 0 or 1.

\begin{theorem}
	\label{proposition:subset-U-maximizer}
	Let $V\neq \emptyset$ finite, $c \in \mathbb{R}^{\pairs{V}}$, $\hat{x}\in \ppc{V}$,
	$U \subseteq V$, $\hat{x}' := \hat{x}\rvert_{P_U}$, $c' := c\rvert_{P_U}$,
	$ij\in P_V\setminus \dom{\hat{x}}$, $b \in \{0, 1\}$ and 
	$y\in \argmax_{x'\in \pfsc{U}{\hat{x}'} \colon x'_{ij} = b}\varphi_{c'}(x')$.
	Let $\tau \colon \pfsc{V}{\hat{x}} \to \pfsc{V}{\hat{x}}$ such that $\tau(x)\rvert_{P_{V\setminus U}} = x\rvert_{P_{V\setminus U}}$ and $\tau(x)\rvert_{P_U}  = y$ for every $x\in \pfsc{V}{\hat{x}}$.
	If~\eqref{eq:lower-bound-subset-U-maximizer}--\eqref{eq:subset-U-maximizer-bounds} hold, there is a solution $x^*$ to $\popcc{V}{c}{\hat{x}}$ with $x^*_{ij} = b$. 
	\\
	\begin{minipage}{0.55\columnwidth}
	\begin{align}
		\label{eq:lower-bound-subset-U-maximizer}
		\lb &\leq \max_{\substack{x\in \pfsc{U}{\hat{x}'} \\ x_{ij} = b}}\varphi_{c'}(x)\\
	\end{align}
	\end{minipage}\hfill
	\begin{minipage}{0.35\columnwidth}
	\begin{align}
		\label{eq:upper-bound-subset-U-maximizer}
		\ub &\geq \max_{\substack{x\in \pfsc{U}{\hat{x}'} \\ x_{ij} = 1-b}}\varphi_{c'}(x)\\
	\end{align}
	\end{minipage}
	\\[-4.5ex]
	\begin{minipage}{0.55\columnwidth}
	\begin{align}
		\label{eq:upper-bound-boundary-subset-U-maximizer}
		\ub' &\geq \max_{\substack{x\in \pfsc{V}{\hat{x}} \\ x_{ij} = 1-b}} \sum_{pq\in \delta(U)}c_{pq}\left(x_{pq} - \tau(x)_{pq}\right)
	\end{align}
	\end{minipage}\hfill
	\begin{minipage}{0.35\columnwidth}
	\begin{align}
		\label{eq:subset-U-maximizer-bounds}
		\lb - \ub \geq \ub'
	\end{align}
	\end{minipage}
\end{theorem}
\begin{proofatend}[Proof of~\Cref{proposition:subset-U-maximizer}]
	We consider $\tau^{ij | b}$ and show that $\tau^{ij | b}$ is improving. 
	If $x_{ij} = b$, we have $\tau^{ij | b}(x) = x$ and $\varphi_{c}(\tau^{ij | b}(x)) = \varphi_c(x)$. 
	If $x_{ij} = 1-b$, we have $\tau^{ij | b}(x) = \tau(x)$ and 
	\begin{align}
		&\varphi_{c}(\tau(x)) - \varphi_c(x) 
		= \sum_{pq\in P_U}c_{pq}(\tau(x)_{pq} - x_{pq}) + \sum_{pq\in \delta(U)}c_{pq}(\tau(x)_{pq} - x_{pq})\\
		\geq & \max_{\substack{x\in \pfsc{U}{\hat{x}'} \\ x_{ij} = b}}\varphi_{c'}(x) 
		- \max_{\substack{x\in \pfsc{U}{\hat{x}'} \\ x_{ij} = 1-b}}\varphi_{c'}(x)
		+ \min_{\substack{x\in \pfsc{V}{\hat{x}} \\ x_{ij} = 1-b}} \sum_{pq\in \delta(U)}c_{pq}\left(\tau(x)_{pq} - x_{pq}\right)\\
		\geq & \; \lb - \ub - \ub' 
		\overset{\eqref{eq:subset-U-maximizer-bounds}}{\geq} \;0\enspace.
	\end{align}
	Thus, $\tau^{ij | b}$ is improving. 
	Moreover, $\tau^{ij | b}(x)_{ij} = b$ for all $x\in \pfsc{V}{\hat{x}}$. 
	This concludes the proof.
\end{proofatend}

\Cref{proposition:subset-U-maximizer} remains non-trivial even if $U = V$:

\begin{corollary}
	\label{corollary:boecker-condition-weak}
	Let $V\neq \emptyset$ finite, $c\in \mathbb{R}^{\pairs{V}}$ and $\hat x\in \ppc{V}$.  Let $ij\in \pairs{V} \setminus \dom{\hat{x}}$ and $b \in \{0, 1\}$.
	Moreover, assume~\eqref{eq:bounds-boecker-condition-weak}.
	 If $\lb \geq \ub$,
	there is a solution $x^*$ to $\popcc{V}{c}{\hat{x}}$ with $x^*_{ij} = b$.
	\begin{align}
		\label{eq:bounds-boecker-condition-weak}
		\lb \leq \max_{\substack{x\in \pfsc{V}{\hat{x}} \\ x_{ij} = b}}\varphi_{c}(x) \quad \text{and} \quad
		\ub \geq \max_{\substack{x\in \pfsc{V}{\hat{x}} \\ x_{ij} = 1-b}}\varphi_{c}(x)
	\end{align}
\end{corollary}
More efficient to decide than the condition of \Cref{corollary:boecker-condition-weak} is the condition of the following specialization:

\begin{corollary}
	\label{corollary:boecker-condition-strong}
	Let $V\neq \emptyset$ finite, $c\in \mathbb{R}^{\pairs{V}}$ and $\hat x\in \ppc{V}$. Let $ij\in \pairs{V}\setminus \dom{\hat{x}}$ and $b \in \{0, 1\}$. 
	Moreover, assume~\eqref{eq:bounds-boecker-condition-strong}.
	If $\lb > \ub$,
	there is a solution $x^*$ to $\popcc{V}{c}{\hat{x}}$ with $x^*_{ij} = b$.
	\begin{align}
		\label{eq:bounds-boecker-condition-strong}
		\lb \leq \max_{x\in \pfsc{V}{\hat{x}}}\varphi_{c}(x) 
		\quad \text{and} \quad
		\ub \geq \max_{\substack{x\in \pfsc{V}{\hat{x}} \\ x_{ij} = 1-b}}\varphi_{c}(x)
	\end{align}
\end{corollary}
\begin{proofatend}[Proof of~\Cref{corollary:boecker-condition-strong}]
	We have
	\begin{align}
		\max_{x\in \pfsc{V}{\hat{x}}}\varphi_{c}(x) \geq \lb > \ub \geq \max_{\substack{x\in \pfsc{V}{\hat{x}} \\ x_{ij} = 1-b}}\varphi_{c}(x)
	\end{align}
	and therefore
	\begin{equation}
		\max_{x\in \pfsc{V}{\hat{x}}}\varphi_{c}(x) = \max_{\substack{x\in \pfsc{V}{\hat{x}} \\ x_{ij} = b}}\varphi_{c}(x) > \max_{\substack{x\in \pfsc{V}{\hat{x}} \\ x_{ij} = 1-b}}\varphi_{c}(x)\enspace.
	\end{equation}
	By~\Cref{corollary:boecker-condition-weak}, there is a solution $x^*$ to $\popcc{V}{c}{\hat{x}}$ such that $x^*_{ij} = b$.
\end{proofatend}

\Cref{corollary:boecker-condition-strong} with bounds as discussed in~\Cref{section:algorithms-preorder-join-bounds} is due to~\citet{boecker2009}.
\Cref{example:subset-U-maximizer} below shows that \Cref{proposition:subset-U-maximizer} is a strict generalization of \Cref{corollary:boecker-condition-weak} and \Cref{corollary:boecker-condition-strong}.

\begin{example}
	\label{example:subset-U-maximizer}
	Consider the instance of the preordering problem depicted in \Cref{figure:example-preorder-join-subset-not-subsumed}.
	Let $\hat{x} = \emptyset$.
	We analyze~\Cref{proposition:subset-U-maximizer} for $i = p$ and $j = q$. 
	For arbitrary $U$, consider the bounds $\lb \coloneqq c_{ij}$, $\ub \coloneqq \sum_{uv\in P_U\setminus \{ij\}}c_{uv}^+$ and $\ub' \coloneqq \sum_{uv\in \delta(U)} \lvert c_{uv}\rvert$.
	For $U = \{p, q\}$, we obtain $\lb = 5, \ub = 0, \ub' = 4$.
	Thus, the condition of \Cref{proposition:subset-U-maximizer} is fulfilled. 
	For $U = \{p, q, r, s\}$, we obtain $\lb = 5, \ub = 6, \ub' = 0$.
	Thus, the condition of \Cref{corollary:boecker-condition-weak} is not fulfilled. 
	This shows that partial optimality due to~\Cref{proposition:subset-U-maximizer} for $U \neq \emptyset$ is not subsumed by partial optimality due to~\Cref{corollary:boecker-condition-weak}.
\end{example}

%% file: section-algorithms.tex
\section{Algorithms}
\label{section:algorithms}

\subsection{Applying \Cref{proposition:edge-cut-condition}}%
We search for subsets $U\subseteq V$ that satisfy the condition of \Cref{proposition:edge-cut-condition} for a fixed pair $ij\in P_V$ by minimizing the rhs.~of \eqref{eq:improvingness-conditional-cut-map-constrained}, i.e., we compute a $U$ that satisfies \eqref{eq:improvingness-conditional-cut-map-constrained} with a \emph{maximum margin}. 
This is a minimum $ij$-cut problem, specifically:
\begin{equation}
	\min_{\substack{U\subseteq V \\ i \in U, j\notin U \\ \hat{x}^{-1}(1) \cap \delta(U, V\setminus U) = \emptyset}}\sum_{pq\in \delta(U, V\setminus U)\setminus \hat{x}^{-1}(0)} c_{pq}^+\enspace.
\end{equation}
Thus, the condition of \Cref{proposition:edge-cut-condition} can be decided for all pairs $ij\in P_V$ by solving $\mathcal{O}(\vert V\vert^2)$ many max-flow problems. 
We use the push-relabel algorithm \cite{goldberg-1988}.
We re-use candidates $U$: Given a candidate $U$ for a specific $ij\in \delta(U, V\setminus U)$ we decide \eqref{eq:improvingness-conditional-cut-map-constrained} for all $pq\in \delta(U, V\setminus U)$.

\subsection{Applying \Cref{corollary:directed-cut-condition}}
In order to find candidates $U$ for \Cref{corollary:directed-cut-condition}, we introduce the auxiliary digraph $G' = (V, E')$ with
$E' = (\{pq\in P_V\mid c_{pq} > 0\} \cup \hat{x}^{-1}(1)) \setminus \hat{x}^{-1}(0)$. 
A $U \subseteq V$ fulfills \Cref{corollary:directed-cut-condition} if $E' \cap \delta(U, V\setminus U) = \emptyset$. 
For each $u \in V$, we compute the set
	$W_u = \{q\in V\mid \exists uq\text{-path in } G'\}$.
For all $u\in V$, we have $E' \cap \delta(W_u, V\setminus W_u) = \emptyset$.
Moreover, $W_u$ is the smallest possible set among the sets $W'_u$ that contain $u$ and fulfill  $E' \cap \delta(W'_u, V\setminus W'_u) = \emptyset$. 
We compute $W_u$ for all $u\in V$ by computing an all-pairs reachability matrix by breadth first search, in $\mathcal{O}(\vert V\vert^2 + \vert V\vert \vert E'\vert)$ time, and then compute $W_u$ in time $\mathcal{O}(\vert V\vert)$, for every $u\in V$.

\subsection{Applying \Cref{proposition:edge-join-condition}}
We find candidates $U, U'\subseteq V$ for \Cref{proposition:edge-join-condition} by minimizing the rhs.~of \eqref{eq:improvingness-join-generalization-map-1}. We formulate this task as a three-label energy minimization problem (see~\Cref{appendix:energy-minimization})
and solve it heuristically using $\alpha\beta$-swap moves~\cite{boykov2001}. In~\Cref{appendix:alpha-beta-swaps}, we show that an optimal $\alpha\beta$-swap can be computed efficiently.

\subsection{Applying \Cref{proposition:subset-U-maximizer,corollary:boecker-condition-weak,corollary:boecker-condition-strong}}
\label{section:algorithms-preorder-join-bounds}

\textbf{Approximate bounds \eqref{eq:lower-bound-subset-U-maximizer} and~\eqref{eq:upper-bound-subset-U-maximizer}.}
Deciding the conditions of \Cref{proposition:subset-U-maximizer,corollary:boecker-condition-weak,corollary:boecker-condition-strong} requires the computation of lower and upper bounds on the objective value of the preordering problem. 
In order to lower-bound the rhs.~of \eqref{eq:lower-bound-subset-U-maximizer}, we compute a feasible solution by means of the local search algorithms of \citet{irmai2025}, in particular, greedy arc fixation and greedy arc insertion.
In order to upper-bound the rhs.~of \eqref{eq:upper-bound-subset-U-maximizer}, we separate values of edges incident to $\{i, j\}$ from values of edges in $V' = V\setminus \{i, j\}$:
\begin{align}
	\label{eq:upper-bound-decomposition}
	&\max_{\substack{x\in \pfsc{V}{\hat{x}} \\ x_{ij} = b}}\sum_{pq\in P_V}c_{pq}x_{pq}
	\leq \max_{\substack{x\in \pfsc{V}{\hat{x}} \\ x_{ij} = b}}\sum_{\substack{pq\in \pairs{V} \\ \{p, q\}\cap \{i, j\} \neq \emptyset}}c_{pq} x_{pq} + \max_{x\in \pfsc{V'}{\hat{x}'}}\sum_{pq\in P_{V'}}c_{pq} x_{pq}
\end{align}%
We follow~\citet{boecker2009} and define the induced value of exclusion and the induced value of inclusion as upper bounds of the first term of the rhs.~of \eqref{eq:upper-bound-decomposition}:
\begin{lemma}
	\label{lemma:ici-ice-bounds}
	Let $V\neq \emptyset$, $c\in \mathbb{R}^{\pairs{V}}$, $\hat{x}\in \ppc{V}$, $ij\in \pairs{V}\setminus \dom{\hat{x}}$ and $V' = V\setminus \{i, j\}$. 
	Furthermore, let
	\begin{align}
		&\ice(ij) 
		\coloneqq
		\max_{x\in \pfsc{V}{\hat{x}}}c_{ji} x_{ji} 
		 + \sum_{w\in V'}\max_{\substack{x\in \pfsc{V}{\hat{x}} \\ x_{ij} = 0}}(c_{iw}x_{iw} + c_{wj}x_{wj}) 
		+ 
		\sum_{w\in V'}(\max_{x\in \pfsc{V}{\hat{x}}}c_{wi}x_{wi} + \max_{x\in \pfsc{V}{\hat{x}}}c_{jw}x_{jw})\\
		&\ici(ij) 
		\coloneqq 
		c_{ij} + \max_{x\in \pfsc{V}{\hat{x}}}c_{ji}x_{ji} 
		+ 
		\sum_{w\in V'}\max_{\substack{x\in \pfsc{V}{\hat{x}}\\ x_{ij} = 1}}\left(c_{jw} x_{jw} + c_{iw}x_{iw}\right) 
		 +
		\sum_{w\in V'}\max_{\substack{x\in \pfsc{V}{\hat{x}}\\ x_{ij} = 1}}\left(c_{wi}x_{wi} + c_{wj}x_{wj}
		\right)
	\end{align}
	Then 
	\begin{align}
		\ice(ij) \geq\max_{\substack{x\in \pfsc{V}{\hat{x}} \\ x_{ij} = 0}}\sum_{\substack{pq\in \pairs{V} \\ \{p, q\}\cap \{i, j\} \neq \emptyset}}c_{pq} x_{pq} 
		\qquad and \qquad
		\ici(ij)  \geq \max_{\substack{x\in \pfsc{V}{\hat{x}} \\ x_{ij} = 1}} \sum_{\substack{pq\in \pairs{V} \\ \{p, q\}\cap \{i, j\} \neq \emptyset}}c_{pq} x_{pq} \enspace.
	\end{align}%
\end{lemma}%
\begin{proofatend}[Proof of~\Cref{lemma:ici-ice-bounds}]
	On the one hand, we have:
	\begin{align}
		&\max_{\substack{x\in \pfsc{V}{\hat{x}} \\ x_{ij} = 0}}\sum_{\substack{pq\in \pairs{V} \\ \{p, q\}\cap \{i, j\} \neq \emptyset}}c_{pq} x_{pq} \\
		= &\max_{\substack{x\in \pfsc{V}{\hat{x}} \\ x_{ij} = 0}} 
		\{
		c_{ji}x_{ji} 
		+ 
		\sum_{w\in V'}\left(c_{iw}x_{iw} + c_{wj}x_{wj}\right) 
		+ 
		\sum_{w\in V'}\left(c_{wi}x_{wi} + c_{jw} x_{jw}\right)\}\\
		\leq 
		&\max_{x\in \pfsc{V}{\hat{x}}}c_{ji}x_{ji} + 
		\sum_{w\in V'} \max_{\substack{x\in \pfsc{V}{\hat{x}} \\ x_{ij} = 0}}\left(c_{iw}x_{iw} + c_{wj}x_{wj}\right)  
		+ \sum_{w\in V'} (\max_{x\in \pfsc{V}{\hat{x}}}c_{wi}x_{wi} 
		+  \max_{x\in \pfsc{V}{\hat{x}}}c_{jw}x_{jw}) \\
		= &\;\ice(ij)
	\end{align}
	
	On the other hand, we have:
	\begin{align}
		&\max_{\substack{x\in \pfsc{V}{\hat{x}} \\ x_{ij} = 1}}\sum_{\substack{pq\in \pairs{V} \\ \{p, q\}\cap \{i, j\} \neq \emptyset}}c_{pq} x_{pq} \\
		= 
		&\max_{\substack{x\in \pfsc{V}{\hat{x}} \\ x_{ij} = 1}}
		\{c_{ij} + c_{ji}x_{ji} 
		+ 
		\sum_{w\in V'}\left(c_{jw}x_{jw} + c_{iw}x_{iw}\right) 
		+ 
		\sum_{w\in V'}\left(c_{wi}x_{wi} + c_{wj} x_{wj}\right)\}\\
		\leq 
		& \; c_{ij} 
		+ 
		\max_{x\in \pfsc{V}{\hat{x}}} c_{ji}x_{ji} 
		+  \sum_{w\in V'} \max_{\substack{x\in \pfsc{V}{\hat{x}} \\ x_{ij} = 1}}\left(c_{jw}x_{jw} + c_{iw}x_{iw}\right)
		+ \sum_{w\in V'} \max_{\substack{x\in \pfsc{V}{\hat{x}} \\ x_{ij} = 1}}\left(c_{wi}x_{wi} +c_{wj}x_{wj}\right)  \\
		= &\;\ici(ij)
	\end{align}
	This concludes the proof.
\end{proofatend}%

For any $pqr\in \triples{V}$, let $E(pqr) := \{pq, qr, pr\}$. 
We upper-bound the second term of the rhs.~of \eqref{eq:upper-bound-decomposition} using an edge-disjoint triple packing $T \subseteq \triples{V}$, where $E(pqr) \cap E(ijk) = \emptyset$ for all distinct $pqr, ijk\in \triples{V}$. 
This packing yields the following upper bound for any $x \in \pfsc{V}{\hat{x}}$: 
\begin{align}
	\sum_{e\in P_V}c_{e} x_{e}	= \smashoperator[r]{\sum_{e\in P_V\setminus \bigcup_{t\in T}E(t)}} c_{e} x_{e} + \sum_{t\in T}\sum_{e\in E(t)}c_{e}x_{e} 
	\leq \hspace{-2ex}\sum_{e\in P_V\setminus \bigcup_{t\in T}E(t)}\max_{x\in \pfsc{V}{\hat{x}}}c_{e} x_{e} + \sum_{t\in T}\max_{x\in \pfsc{V}{\hat{x}}}\sum_{e\in E(t)}c_{e}x_{e} \enspace.
\end{align}%

\textbf{Tight bounds \eqref{eq:lower-bound-subset-U-maximizer} and~\eqref{eq:upper-bound-subset-U-maximizer}.}
Tight bounds \eqref{eq:lower-bound-subset-U-maximizer} and \eqref{eq:upper-bound-subset-U-maximizer} can be computed efficiently in the special case where $b = 1$ and the $x^+ \in \{0, 1\}^{P_V}$ such that $\forall e \in P_V \colon x^+_e = 1 \Leftrightarrow c_e \geq 0 \vee \hat{x}_e = 1$ is feasible for~$\popcc{V}{c}{\hat{x}}$, i.e., $x^+\in \pfsc{V}{\hat{x}}$, and thus a solution to $\max_{x\in \pfsc{V}{\hat{x}}}\varphi_{c}(x)$.
In this case, for any $ij \in P_V\setminus \dom{\hat{x}}$ such that $c_{ij}\geq 0$, solving $\max_{x\in \pfsc{V}{\hat{x}} \colon x_{ij} = b'}\varphi_{c}(x)$ is tractable,
either by optimality of $x^+$ (for $b' = 1$), or by solving a minimum $ij$-cut problem (for $b' = 0$): 
\begin{proposition}
	\label{proposition:tractable-U}
	Let $V\neq \emptyset$, $c\in \mathbb{R}^{\pairs{V}}$ and $\hat x\in \ppc{V}$. Let $x^+ \in \{0, 1\}^{P_V}$ such that 
	\begin{align}
		\forall e\in P_V\colon \quad x^+_{e} = \begin{cases}
			1 & \textnormal{if $e\notin \dom{\hat{x}} \land c_e \geq 0$} \\
			0 & \textnormal{if $e\notin \dom{\hat{x}} \land c_e < 0$} \\
			\hat{x}_e & \textnormal{if $e\in \dom{\hat{x}}$}
		\end{cases}\enspace.
	\end{align}
	Moreover, let $ij\in \pairs{V}\setminus \dom{\hat{x}}$ with $c_{ij} \geq 0$ and thus $x^+_{ij} = 1$.
	If $x^+\in \pfsc{V}{\hat{x}}$ then 
	\begin{align}
		\label{eq:tractable-U-xij-1}
		&\max_{\substack{x\in \pfsc{V}{\hat{x}} \\ x_{ij} = 1}}\varphi_{c}(x) = \varphi_{c}(x^+)\enspace, \\
		\label{eq:tractable-U-xij-0}
		&
		\max_{\substack{x\in \pfsc{V}{\hat{x}} \\ x_{ij} = 0}}\varphi_{c}(x) =
		\smashoperator[l]{\max_{\substack{U\subseteq V \\ ij\in \delta(U, V\setminus U) \\ \delta(U, V\setminus U) \cap \hat{x}^{-1}(1) = \emptyset}}}
		\hspace{-2ex}\varphi_{c}(\sigma_{\delta(U, V\setminus U)}(x^+)) 
		 =\varphi_c(x^+) - \smashoperator[l]{\min_{\substack{U\subseteq V \\ ij\in \delta(U, V\setminus U) \\ \delta(U, V\setminus U) \cap \hat{x}^{-1}(1) = \emptyset}}}
		\sum_{pq\in \delta(U, V\setminus U)\setminus \hat{x}^{-1}(0)} \hspace{-2ex}c_{pq}^+\enspace.
	\end{align}
\end{proposition}%
\begin{proofatend}[Proof of~\Cref{proposition:tractable-U}]
	Since $x^+ \in \pfsc{V}{\hat{x}}$ and $x^+_{ij} = 1$, we have 
	\begin{align}
		\max_{\substack{x\in \pfsc{V}{\hat{x}} \\ x_{ij} = 1}}\varphi_{c}(x) \geq \varphi_c(x^+)\enspace.
	\end{align}
	Moreover, for every $x\in \pfsc{V}{\hat{x}}$ we have 
	\begin{align}
		\varphi_c(x) = 
		\smashoperator[r]{\sum_{e\in P_V}}c_e x_e \leq 
		\smashoperator[r]{\sum_{e\in \hat{x}^{-1}(1)}}c_e + 
		\smashoperator[r]{\sum_{e\in P_V \setminus \dom{\hat{x}}}}c_e^+ = \varphi_c(x^+)\enspace. 
	\end{align}
	Thus, \eqref{eq:tractable-U-xij-1} is true.
	
	For every $U\subseteq V$ such that $\delta(U, V\setminus U)\cap \hat{x}^{-1}(1) = \emptyset$ and $ij\in \delta(U, V\setminus U)$ the $x'' \coloneq \sigma_{\delta(U, V\setminus U)}(x^+)$ is such that $x''\in \pfsc{V}{\hat{x}}$ and $x''_{ij} = 0$ by~\Cref{definition:elementary-cut-map} and~\Cref{lemma:cut-map-trueness}. 
	
	Now, let $x'\in \pfsc{V}{\hat{x}}$ be the solution to the lhs.~of \eqref{eq:tractable-U-xij-0}. 
	We define $U = \{p\in V\mid x'_{ip} = 1\}$ and $x'' = \sigma_{\delta(U, V\setminus U)}(x^+)$. 
	We have $i\in U$ and $j\in V\setminus U$ and thus $ij\in \delta(U, V\setminus U)$ since $x'_{ij} = 0$. 
	Moreover, we have $\delta(U, V\setminus U)\cap \hat{x}^{-1}(1) = \emptyset$ since $x'_{pq} = 0$ for every $pq\in \delta(U, V\setminus U)$ by definition of $U$.
	Therefore, we have $x'' \in \pfsc{V}{\hat{x}}$ and $x''_{ij} = 0$. 
	Furthermore, we show that $\varphi_c(x'') \geq \varphi_c(x')$:
	\begin{align}
		&\varphi_c(x'') - \varphi_c(x') = \sum_{pq\in P_V\setminus \dom{\hat{x}}}c_{pq}(x''_{pq} - x'_{pq})
		= \sum_{\substack{pq\in P_V\setminus \dom{\hat{x}} \\ pq\notin \delta(U, V\setminus U)}}c_{pq}(x^+_{pq} - x'_{pq})\\
		=& \sum_{\substack{pq\in P_V\setminus \dom{\hat{x}} \\ pq\notin \delta(U, V\setminus U) \\ c_{pq} \geq 0}}c_{pq}(1 - x'_{pq}) + \sum_{\substack{pq\in P_V\setminus \dom{\hat{x}} \\ pq\notin \delta(U, V\setminus U) \\ c_{pq} < 0}}c_{pq}(0 - x'_{pq})
		\geq \;0\enspace.
	\end{align} 
	Therefore, the first equality of \eqref{eq:tractable-U-xij-0} is true.
	
	For every $e\in \delta(U, V\setminus U)$, we have $x^+_e = 1$ if and only if $c_e \geq 0$ and $\hat{x}_e \neq 0$ since $\delta(U, V\setminus U)\cap \hat{x}^{-1}(1) = \emptyset$. 
	From this observation follows the second equality in \eqref{eq:tractable-U-xij-0}.
\end{proofatend}%

\textbf{Bound \eqref{eq:upper-bound-boundary-subset-U-maximizer}.}
We describe three maps $\tau$ for \Cref{proposition:subset-U-maximizer}:
For a fixed set $U\subseteq V$, these maps all establish the optimal solution within $U$ and cut on the boundary of $U$ either all pairs $\delta(U, V\setminus U)$, all pairs $\delta(V \setminus U, U)$, or all pairs $\delta(U)$.
Let $y\in \argmax_{x'\in \pfsc{U}{\hat{x}'} \colon x'_{ij} = b}\varphi_{c'}(x')$. We define:
{\small
\begin{align}
	\tau^y_{\delta(U, V\setminus U)}(x)_{pq} &= \begin{cases}
		y_{pq} & \textnormal{if $pq\in P_U$}\\
		0 & \textnormal{if $pq\in \delta(U, V\setminus U)$}\\
		1 & \textnormal{if $pq\in \delta(V\setminus U, U) \land \exists r\in U\colon x_{pr} = 1 \land y_{rq} = 1 $}
		\\
		0 &\textnormal{if $pq\in \delta(V\setminus U, U) \land \forall r\in U\colon x_{pr} = 0\lor y_{rq} = 0 $}
		\\
		x_{pq} & \textnormal{if $pq\in P_{V\setminus U}$}
	\end{cases} \\
	\tau^y_{\delta(V\setminus U, U)}(x)_{pq} &= \begin{cases}
		y_{pq} & \textnormal{if $pq\in P_U$}\\
		0 & \textnormal{if $pq\in \delta(V\setminus U, U)$}\\
		1 & \textnormal{if $pq\in \delta(U, V\setminus U) \lor \exists r\in U\colon y_{pr} = 1 \land x_{rq} = 1 $}\\
		0 & \textnormal{if $pq\in \delta(U, V\setminus U)\land \forall r\in U\colon y_{pr} = 0 \lor x_{rq} = 0 $}\\
		x_{pq} & \textnormal{if $pq\in P_{V\setminus U}$}
	\end{cases} \\
	\tau^y_{\delta(U)}(x)_{pq} &= \begin{cases}
		y_{pq} & \textnormal{if $pq\in P_U$}\\
		0 & \textnormal{if $pq\in \delta(U)$}\\
		x_{pq} & \textnormal{if $pq\in P_{V\setminus U}$}
	\end{cases}
\end{align}}%
Below, \Cref{corollary:preorder-join-boundary-bounds-general} describes the bounds \eqref{eq:upper-bound-boundary-subset-U-maximizer} induced by these three choices of $\tau$ in general.
\Cref{corollary:preorder-join-boundary-bounds-exact} describes these bounds in the special case employed in~\Cref{proposition:tractable-U}.
\Cref{lemma:preorder-join-variables-boundary-tech-1} is a technical statement that we apply in~\Cref{corollary:trueness-preorder-join-maps} to establish trueness of any of the above-mentioned maps to some $\hat{x}\in \ppc{V}$, and that we apply in~\Cref{corollary:preorder-join-boundary-bounds-general,corollary:preorder-join-boundary-bounds-exact} in order to establish efficiently computable bounds \eqref{eq:upper-bound-boundary-subset-U-maximizer}.
\begin{lemma}
	\label{lemma:preorder-join-variables-boundary-tech-1}
	Let $V\neq \emptyset$ finite, $U \subseteq V$, $\hat{x}\in \ppc{V}$, $x\in \pfsc{V}{\hat{x}}$, $\hat{x}' := \hat{x}\rvert_{P_U}$ and $y\in \pfsc{U}{\hat{x}'}$.
	
	(1) For $\tau = \tau^y_{\delta(U, V\setminus U)}$ let 
	\begin{align}
		P_{01}' &\coloneqq 
		\{ pq \in \delta(V\setminus U, U) \mid
		\exists r \in U \colon \hat{x}_{pr} \neq 0 \land y_{rq} = 1 \} \setminus \hat{x}^{-1}(1) \\
		P_{10}' &= P_{10}''\coloneqq \delta(U, V\setminus U)\setminus \hat{x}^{-1}(0)\enspace.
	\end{align}%
	Moreover, let $P_{01}'' = \delta(V\setminus U, U)\setminus \hat{x}^{-1}(1)$. 
	Then $\pzo{\tau} \cap \delta(U) \subseteq P_{01}' \subseteq P_{01}''$ and $\poz{\tau} \cap \delta(U) \subseteq P_{10}' = P_{10}''$.
	
	(2) For $\tau = \tau^y_{\delta(V\setminus U, U)}$ let 
	\begin{align}
		P_{01}' &\coloneqq  
		\{ pq \in \delta(U, V\setminus U) \mid
		\exists r \in U \colon y_{pr} = 1\land \hat{x}_{rq} \neq 0 \} \setminus \hat{x}^{-1}(1)\\
		P_{10}' &= P_{10}''\coloneqq \delta(V\setminus U, U)\setminus \hat{x}^{-1}(0)\enspace.
	\end{align}%
	Moreover, let $P_{01}'' = \delta(U, V\setminus U)\setminus \hat{x}^{-1}(1)$.
	Then $\pzo{\tau} \cap \delta(U) \subseteq P_{01}' \subseteq P_{01}''$ and $\poz{\tau} \cap \delta(U) \subseteq P_{10}' = P_{10}''$.
	
	(3) For $\tau = \tau^y_{\delta(U)}$ let 
	\begin{align}
		P_{01}' &= P_{01}'' \coloneqq \emptyset \\ 
		P_{10}' &= P_{10}'' \coloneqq \delta(U)\setminus \hat{x}^{-1}(0)\enspace.
	\end{align}%
	Then $\pzo{\tau} \cap \delta(U) \subseteq P_{01}' = P_{01}''$ and $\poz{\tau} \cap \delta(U) \subseteq P_{10}' = P_{10}''$.
\end{lemma}
\begin{proofatend}[Proof of~\Cref{lemma:preorder-join-variables-boundary-tech-1}]
	In all three cases: $P_{01}' \subseteq P_{01}''$ and $P_{10}' \subseteq P_{10}''$.
	
	(1) Let $\tau = \tau^y_{\delta(U, V\setminus U)}$ and $x' = \tau(x)$.
	
	Firstly, we show that $\pzo{\tau}\cap \delta(U) \subseteq P_{01}'$. 
	Obviously, $\pzo{\tau}\cap \hat{x}^{-1}(1) = \emptyset$.
	For every $pq\in \delta(U, V\setminus U)$ we have $x'_{pq} = 0$ and thus $pq\notin \pzo{\tau}$.
	For every $pq\in \delta(V\setminus U, U)$ such that for every $r\in U$ it holds  $\hat{x}_{pr} = 0$ or $y_{rq} = 0$, we have that $x'_{pq} = 0$ by definition of $\tau$, and therefore $pq\notin \pzo{\tau}$. 
	Thus, $\pzo{\tau}\cap \delta(U) \subseteq P_{01}'$. 
	
	Secondly, we show that $\poz{\tau}\cap \delta(U) \subseteq P_{10}'$. 
	Obviously, $\hat{x}^{-1}(0) \cap \poz{\tau} = \emptyset$.
	For every $pq\in \delta(V\setminus U, U)$ we have that $x_{pq} = 1$ implies $x'_{pq} = 1$, by definition of $\tau$. 
	Therefore, $\delta(V\setminus U, U)\cap \poz{\tau} = \emptyset$. 
	Thus, $\poz{\tau} \cap\delta(U)\subseteq P_{10}'$.
	
	(2) Let $\tau = \tau^y_{\delta(V\setminus U, U)}$ and $x' = \tau(x)$. 
	
	Firstly, we show that $\pzo{\tau}\cap \delta(U) \subseteq P_{01}'$. 
	Obviously, $\pzo{\tau}\cap \hat{x}^{-1}(1) = \emptyset$.
	For every $pq\in \delta(V\setminus U, U)$ we have $x'_{pq} = 0$ and thus, $pq\notin \pzo{\tau}$.
	For every $pq\in \delta(U, V\setminus U)$ such that for every $r\in U$ it holds  $y_{pr} = 0$ or $\hat{x}_{rq} = 0$, we have that $x'_{pq} = 0$ by definition of $\tau$, and therefore $pq\notin \pzo{\tau}$. 
	Thus, $\pzo{\tau}\cap \delta(U) \subseteq P_{01}'$. 
	
	Secondly, we show that $\poz{\tau}\cap \delta(U) \subseteq P_{10}'$. 
	Obviously, $\hat{x}^{-1}(0) \cap \poz{\tau} = \emptyset$.
	For every $pq\in \delta(U, V\setminus U)$ we have that $x_{pq} = 1$ implies $x'_{pq} = 1$ by definition of $\tau$. 
	Therefore, $\delta(U, V\setminus U)\cap \poz{\tau} = \emptyset$. 
	Thus, $\poz{\tau} \cap\delta(U)\subseteq P_{10}'$.
	
	(3) Let $\tau = \tau^y_{\delta(U)}$ and $x' = \tau(x)$. 
	
	Firstly, we show that $\pzo{\tau}\cap \delta(U) \subseteq P_{01}'$. 
	For every $pq\in \delta(U)$ we have $x'_{pq} = 0$ and thus $pq\notin \pzo{\tau}$. Therefore, $\pzo{\tau}\cap \delta(U) = \emptyset \subseteq P_{01}' = \emptyset$.
	
	Secondly, we show that $\poz{\tau}\cap \delta(U) \subseteq P_{10}'$. 
	Obviously, $\hat{x}^{-1}(0) \cap \poz{\tau} = \emptyset$.
	Thus, $\poz{\tau} \cap\delta(U)\subseteq P_{10}'$.
\end{proofatend}
\begin{corollary}
	\label{corollary:trueness-preorder-join-maps}
	Let $V\neq \emptyset$ finite, $U \subseteq V$, $\hat{x}\in \ppc{V}$, $\hat{x}' := \hat{x}\rvert_{P_U}$, $y\in \pfsc{U}{\hat{x}'}$ and $\tau \in \{\tau^y_{\delta(U, V\setminus U)}, \tau^y_{\delta(V\setminus U, U)}, \tau^y_{\delta(U)}\}$.
	Let $P_{01}'$ ($P_{01}''$) and $P_{10}'$ ($P_{10}''$) be defined as in~\Cref{lemma:preorder-join-variables-boundary-tech-1}. 
	Then $\tau$ is true to $\hat{x}$ if $P_{10}'\cap \hat{x}^{-1}(1) = P_{01}'\cap \hat{x}^{-1}(0) =\emptyset$ ($P_{10}''\cap \hat{x}^{-1}(1) = P_{01}''\cap \hat{x}^{-1}(0) =\emptyset$).
\end{corollary}
\begin{corollary}
	\label{corollary:preorder-join-boundary-bounds-general}
	Let $V\neq \emptyset$, $c \in \mathbb{R}^{\pairs{V}}$, $\hat{x}\in \ppc{V}$, $U \subseteq V$, $\hat{x}' :=\hat{x}\rvert_{P_U}$ and $c' := c\rvert_{P_U}$. 
	Moreover, let $ij\in P_V\setminus \dom{\hat{x}}$, $b \in \{0, 1\}$ and 
	$y\in \argmax_{x'\in \pfsc{U}{\hat{x}'} \colon x'_{ij} = b}\varphi_{c'}(x')$.
	Let $\tau \in \{\tau^{y}_{\delta(U, V\setminus U)}, \tau^{y}_{\delta(V\setminus U, U)}, \tau^{y}_{\delta(U)}\}$
	and let $P_{01}''$ and $P_{10}''$ be defined as in~\Cref{lemma:preorder-join-variables-boundary-tech-1}. 
	Then:
	\begin{align}
		\forall x\in \pfsc{V}{\hat{x}}\colon \sum_{pq\in \delta(U)}c_{pq}\left(x_{pq} - \tau(x)_{pq}\right)  \leq \sum_{pq\in P_{01}''}c_{pq}^- + \sum_{pq\in P_{10}''}c_{pq}^+
	\end{align}%
\end{corollary}
\begin{proofatend}[Proof of~\Cref{corollary:preorder-join-boundary-bounds-general}]
	Let $\tau \in \{\tau^y_{\delta(U, V\setminus U)}, \tau^y_{\delta(V\setminus U, U)}, \tau^y_{\delta(U)}\}$, $x\in \pfsc{V}{\hat{x}}$ and $x' = \tau(x)$. 
	By~\Cref{lemma:preorder-join-variables-boundary-tech-1}, we have that $P_{01}''\cap P_{10}'' = \emptyset$.
	Moreover, for every $pq\in \delta(U) \setminus\left(P_{01}''\cup P_{10}''\right)$ it follows that $pq\notin \poz{\tau} \cup \pzo{\tau}$, and therefore $x'_{pq} = x_{pq}$. 
	Thus, for every $pq\in P_{01}''$ we have that $x_{pq} - x'_{pq}\in \{0, -1\}$, and 
	for every $pq\in P_{10}''$ we have that $x_{pq} - x'_{pq} \in \{0, 1\}$. 
	Therefore:
	\begin{align}
		\sum_{pq\in \delta(U)}c_{pq}\left(x_{pq} - x'_{pq}\right) 
		= \sum_{\mathclap{pq\in P_{01}''}}c_{pq}(x_{pq} - x'_{pq}) + \sum_{pq\in P_{10}''}c_{pq}(x_{pq} - x'_{pq})
		\leq \sum_{pq\in P_{01}''}c_{pq}^- + \sum_{pq\in P_{10}''}c_{pq}^+
	\end{align}
	This concludes the proof.
\end{proofatend}
\begin{corollary}
	\label{corollary:preorder-join-boundary-bounds-exact}
	Let $V\neq \emptyset$, $c \in \mathbb{R}^{\pairs{V}}$ and $\hat{x}\in \ppc{V}$. 
	Additionally, let $U \subseteq V$, $\hat{x}' :=\hat{x}\rvert_{P_U}$ and $c' := c\rvert_{P_U}$. 
	Moreover, let $ij\in P_V\setminus \dom{\hat{x}}$ and $y^+ \in \{0, 1\}^{P_U}$ such that
	\begin{align}
		\forall e\in P_U\colon \quad y^+_{e} = \begin{cases}
			1 & \textnormal{if $e\notin \dom{\hat{x}'} \land c_e \geq 0$} \\
			0 & \textnormal{if $e\notin \dom{\hat{x}'} \land c_e < 0$} \\
			\hat{x}'_e & \textnormal{if $e\in \dom{\hat{x}'}$}
		\end{cases}\enspace.
	\end{align}
	Let $\tau \in \{\tau^{y^+}_{\delta(U, V\setminus U)}, \tau^{y^+}_{\delta(V\setminus U, U)}, \tau^{y^+}_{\delta(U)}\}$
	and $P_{01}'$ and $P_{10}'$ as in~\Cref{lemma:preorder-join-variables-boundary-tech-1}. 
	If $y^+ \in \pfsc{U}{\hat{x}'}$, then:
	\begin{align}
		\forall x\in \pfsc{V}{\hat{x}}\colon \sum_{pq\in \delta(U)}c_{pq}\left(x_{pq} - \tau(x)_{pq}\right)  \leq \sum_{pq\in P_{01}'}c_{pq}^- + \sum_{pq\in P_{10}'}c_{pq}^+
	\end{align}
\end{corollary}
\begin{proofatend}[Proof of~\Cref{corollary:preorder-join-boundary-bounds-exact}]
	Let $\tau \in \{\tau^{y^+}_{\delta(U, V\setminus U)}, \tau^{y^+}_{\delta(V\setminus U, U)}, \tau^{y^+}_{\delta(U)}\}$, $x\in \pfsc{V}{\hat{x}}$ and $x' = \tau(x)$. 
	By~\Cref{lemma:preorder-join-variables-boundary-tech-1}, we have that $P_{01}'\cap P_{10}' = \emptyset$.
	Moreover, for every $pq\in \delta(U) \setminus\left(P_{01}'\cup P_{10}'\right)$ it follows that $pq\notin \poz{\tau} \cup \pzo{\tau}$, and therefore $x'_{pq} = x_{pq}$. 
	Thus, for every $pq\in P_{01}'$ we have that $x_{pq} - x'_{pq}\in \{0, -1\}$, and
	for every $pq\in P_{10}'$ we have that $x_{pq} - x'_{pq} \in \{0, 1\}$. 
	Therefore:
	\begin{align}
		\sum_{pq\in \delta(U)}c_{pq}\left(x_{pq} - x'_{pq}\right) 
		= \sum_{\mathclap{pq\in P_{01}'}}c_{pq}(x_{pq} - x'_{pq}) + \sum_{pq\in P_{10}'}c_{pq}(x_{pq} - x'_{pq})
		\leq  \sum_{pq\in P_{01}'}c_{pq}^- + \sum_{pq\in P_{10}'}c_{pq}^+
	\end{align}
	This concludes the proof.
\end{proofatend}

\subsection{Contraction}
The contraction of equivalence classes is discussed in~\Cref{appendix:merging-of-equivalence-classes}.

%% file: section-experiments.tex
\section{Experiments}\label{section:experiments}
\begin{figure*}[p]
\begin{minipage}{\textwidth}
	\centering
	\input{figure-synthetic}
	\caption{
		Shown above are the percentage of fixed variables (Row~1) and runtimes (Row~2) 
		for applying \Cref{proposition:edge-cut-condition},~\Cref{corollary:directed-cut-condition},~\Cref{proposition:edge-join-condition}, ~\Cref{proposition:subset-U-maximizer} (with~\Cref{proposition:tractable-U}) and~\Cref{corollary:boecker-condition-strong} for $b\in \{0, 1\}$ individually 
		to instances of the synthetic dataset with respect to $\alpha\in \left[0, 1\right]$, $\lvert V\rvert = 40$ and $p_E \in \{0.25, 0.50, 0.75\}$.
		}
	\label{figure:synthetic}
\end{minipage}\\[3ex]
\begin{minipage}[t]{0.485\textwidth}
	\centering
	\input{figure-synthetic-alphas}
	\caption{
		Shown above are the percentage of fixed variables (Row~1) and runtimes (Row~2) 
		for applying \Cref{proposition:edge-cut-condition,corollary:directed-cut-condition,proposition:edge-join-condition,corollary:boecker-condition-strong,proposition:subset-U-maximizer} jointly
		to instances of the synthetic dataset with respect to $\lvert V\rvert = 40$ and $p_E \in \{0.25, 0.50, 0.75\}$.}
	\label{figure:synthetic-alphas}
\end{minipage}%
\hfill
\begin{minipage}[t]{0.485\textwidth}
	\centering
	\input{figure-synthetic-ns}
	\caption{
		Shown above are the percentage of fixed variables (Row~1) and runtimes (Row~2) 
		for applying \Cref{proposition:edge-cut-condition,corollary:directed-cut-condition,proposition:edge-join-condition,corollary:boecker-condition-strong,proposition:subset-U-maximizer} jointly
		to instances of the synthetic dataset with $\alpha\in \{0.25, 0.65, 0.70, 0.75\}$ and $p_E \in \{0.25, 0.50, 0.75\}$.}
	\label{figure:synthetic-ns}
\end{minipage}\\[3ex]
\begin{minipage}{\textwidth}
	\centering
	\input{figure-twitter}
	\caption{
		Shown above are the percentage of fixed variables (Row~1) as well as the runtimes (Row~2) 
		for applying \Cref{proposition:edge-cut-condition},~\Cref{corollary:directed-cut-condition},~\Cref{proposition:edge-join-condition}, ~\Cref{proposition:subset-U-maximizer} (together with~\Cref{proposition:tractable-U}) and~\Cref{corollary:boecker-condition-strong} for $b\in \{0, 1\}$ individually 
		as well as all conditions jointly (Column~7)
		to instances of the Twitter dataset. 
		}
	\label{figure:twitter}
\end{minipage}
\end{figure*}

We examine the algorithms defined in \Cref{section:algorithms} for deciding the partial optimality conditions established in \Cref{section:conditions} empirically, on three collections of instances of the preordering problem. 
We report percentages of variables fixed and runtimes on one core of an Intel Core i9-12900KF and 64 GB of RAM.
The code of the algorithms and for reproducing the experiments is provided in \citet{stein-code-2026}.

\subsection{Synthetic Instances}

For a systematic study, we synthesize instances wrt.~a preorder $x'$ and a design parameter $\alpha \in \left[0, 1\right]$.
Values $c_{pq}$ are drawn from two Gaussian distributions with means $+1 - \alpha$ and $-1 + \alpha$, the first in case $x'_{pq} = 1$, the second in case $x'_{pq} = 0$. 
The spread of both is $0.1+0.3\alpha$.
Any preorder $x'$ is constructed iteratively, wrt.~a design parameter $p_E \in \left[0, 1\right]$, starting in iteration $t = 0$ with the preorder $x'_0 := 0$. 
In every iteration $t$, a distinction is made:
If $\lvert (x'_t)^{-1}(1)\rvert / \lvert P_V\rvert < p_E$, an edge $e\in P_V \setminus (x'_t)^{-1}(1)$ is drawn uniformly at random, the preorder of the next iteration is defined as $x'_{t+1} := \sigma_{e}(x'_t)$, and the iteration counter is incremented.
Otherwise, $x'_t$ is output.
For any choice of $\alpha$ and $|V|$, an ensemble of $100$ such instances is constructed, from $5$ preorders $x'$, with $20$ value vectors for each.
When reporting percentages of variables fixed and runtimes, we consistently report the median, $\frac{1}{4}$- and $\frac{3}{4}$-quantile over the respective ensemble.

\Cref{figure:synthetic} shows the percentage of variables fixed and the runtime, both as a function of $\alpha$, for the cut conditions \Cref{proposition:edge-cut-condition},~\Cref{corollary:directed-cut-condition} and \Cref{corollary:boecker-condition-strong} ($b = 0$), 
and the join conditions \Cref{proposition:edge-join-condition},~\Cref{proposition:subset-U-maximizer} (together with~\Cref{proposition:tractable-U}) and~\Cref{corollary:boecker-condition-strong} ($b = 1$) separately.
All join conditions are applied on top of partial optimality from the cut conditions. 
It can be seen that \Cref{proposition:edge-cut-condition} fixes more variables to zero than \Cref{corollary:directed-cut-condition} consistently but requires more than two orders of magnitude longer to apply.
\Cref{corollary:boecker-condition-strong} ($b = 0$) and~\Cref{proposition:edge-cut-condition} are similarly effective in fixing variables to zero and take similar time to apply.
All join conditions are effective.
\Cref{corollary:boecker-condition-strong} ($b = 1$) fixes the most variables to one.
The effect of applying conditions jointly is shown as a function of $\alpha$ in \Cref{figure:synthetic-alphas}, and as a function of $|V|$ in \Cref{figure:synthetic-ns}.
It can be seen from \Cref{figure:synthetic-alphas} that all variables are fixed for instances with sufficiently small $\alpha$, and no variables are fixed for instances with sufficiently large $\alpha$.
It can be seen from \Cref{figure:synthetic-ns} that the time for applying the conditions jointly is polynomial in $|V|$.

\subsection{Social Networks}
Toward real applications, we examine partial optimality for instances of the preordering problem defined wrt.~published Twitter and Google+ ego networks \cite{leskovec2012}.
We define these instances such that $c_{ij} = 1$ if $i$ follows $j$, and $c_{ij} = -1$, otherwise. 
We analyze all Twitter instances and those Google+ instances with $\lvert V\rvert \leq 250$. The results for Google+ are similar to those for Twitter and are thus deferred to~\Cref{appendix:google+}, in particular~\Cref{figure:gplus-250}.

For Twitter, \Cref{figure:twitter} shows the percentage of variables fixed and the runtime, both as a function of $|V|$, for applications of cut conditions, \Cref{proposition:edge-cut-condition},~\Cref{corollary:directed-cut-condition} and \Cref{corollary:boecker-condition-strong} ($b = 0$), and join conditions, \Cref{proposition:edge-join-condition},~\Cref{proposition:subset-U-maximizer} (together with~\Cref{proposition:tractable-U}) and~\Cref{corollary:boecker-condition-strong} ($b = 1$), separately and jointly.
Also here, all join conditions are applied on top of partial optimality from the cut conditions. 
It can be seen that \Cref{proposition:edge-cut-condition} fixes more variables to zero (median: 29.6\%) than \Cref{corollary:directed-cut-condition} (median: 16.7\%) consistently, but takes two orders of magnitude more time to apply.
Both \Cref{proposition:edge-cut-condition,corollary:directed-cut-condition} are generally more effective in fixing variables to zero than~\Cref{corollary:boecker-condition-strong} ($b = 0$).
Fewer variables are fixed to one, but each join condition fixes some variables.

For a direct comparison of \Cref{corollary:boecker-condition-strong} to all conditions applied jointly, see \Cref{appendix:boecker}.

%% file: figure-synthetic.tex
\tikzmath{\plotheight=3.2;}
\tikzmath{\plotwidth=3.2;}

\setlength{\hspacing}{0.2cm}
\setlength{\vspacing}{0.2cm}

\providecommand{\addzeropersistencyplot}{}	
\renewcommand{\addzeropersistencyplot}[2]{\addplot+[
	only marks,
	mark=*,
	mark options={scale=0.25},
	color=#2
	] table[
	col sep=comma,
	x expr=\thisrow{alpha}, 
	y expr=\thisrow{medianZeroVariables}*100,
	y error minus expr=(\thisrow{medianZeroVariables} - \thisrow{q25ZeroVariables})*100,
	y error plus expr=(\thisrow{q75ZeroVariables} - \thisrow{medianZeroVariables})*100
	] {#1};}

\providecommand{\addonepersistencyplot}{}	
\renewcommand{\addonepersistencyplot}[2]{\addplot+[
	only marks,
	mark=*,
	mark options={scale=0.25},
	color=#2
	] table[
	col sep=comma,
	x expr=\thisrow{alpha}, 
	y expr=\thisrow{medianOneVariables}*100,
	y error minus expr=(\thisrow{medianOneVariables} - \thisrow{q25OneVariables})*100,
	y error plus expr=(\thisrow{q75OneVariables} - \thisrow{medianOneVariables})*100
	] {#1};}

\providecommand{\addruntimeplot}{}	
\renewcommand{\addruntimeplot}[2]{\addplot+[
	only marks,
	mark=*,
	mark options={scale=0.25},
	color=#2
	] table[
	col sep=comma,
	x expr=\thisrow{alpha}, 
	y expr=\thisrow{medianDuration} / 1e9,
	y error minus expr=(\thisrow{medianDuration} - \thisrow{q25Duration}) / 1e9,
	y error plus expr=(\thisrow{q75Duration} - \thisrow{medianDuration}) / 1e9
	] {#1};}

	\tikzsetnextfilename{fig-synthetic}
	\begin{tikzpicture}		
		\begin{groupplot}[group style={group size= 6 by 2, horizontal sep=\hspacing, vertical sep=\vspacing}]

			\nextgroupplot[
			title={Thm.~\ref{proposition:edge-cut-condition}},
			title style={align=center, font=\small, yshift=-1ex},
			xmin=0,
			xmax=1,
			xtick distance=0.5,
			ymin=0,
			ymax=100,
			width=\plotwidth cm,
			height=\plotheight cm,
			xticklabels=\empty,
			ylabel={$\frac{\lvert \dom{\hat{x}}\rvert}{\lvert P_V\rvert}\left[\%\right]$},
			]
			\addzeropersistencyplot{./data/RandomTransitivelyClosed/edgeCut/alphas/0.25/stats_40.csv}{blue};
			\addzeropersistencyplot{./data/RandomTransitivelyClosed/edgeCut/alphas/0.50/stats_40.csv}{\darkgreen};
			\addzeropersistencyplot{./data/RandomTransitivelyClosed/edgeCut/alphas/0.75/stats_40.csv}{red};
			
			\nextgroupplot[
			title={Cor.~\ref{corollary:directed-cut-condition}},
			title style={align=center, font=\small, yshift=-1ex},
			xmin=0,
			xmax=1,
			xtick distance=0.5,
			ymin=0,
			ymax=100,
			width=\plotwidth cm,
			height=\plotwidth cm,
			yticklabels=\empty,
			xticklabels=\empty
			]
			\addzeropersistencyplot{./data/RandomTransitivelyClosed/directedCut/alphas/0.25/stats_40.csv}{blue};
			\addzeropersistencyplot{./data/RandomTransitivelyClosed/directedCut/alphas/0.50/stats_40.csv}{\darkgreen};
			\addzeropersistencyplot{./data/RandomTransitivelyClosed/directedCut/alphas/0.75/stats_40.csv}{red};

			\nextgroupplot[
			title={Cuts+Thm.~\ref{proposition:edge-join-condition}},
			title style={align=center, font=\small, yshift=-1ex},
			xmin=0,
			xmax=1,
			xtick distance=0.5,
			ymin=0,
			ymax=100,
			width=\plotwidth cm,
			height=\plotheight cm,
			yticklabels=\empty,
			xticklabels=\empty
			]
			\addonepersistencyplot{./data/RandomTransitivelyClosed/cuts+edgeJoin/alphas/0.25/stats_40.csv}{blue};
			\addonepersistencyplot{./data/RandomTransitivelyClosed/cuts+edgeJoin/alphas/0.50/stats_40.csv}{\darkgreen};
			\addonepersistencyplot{./data/RandomTransitivelyClosed/cuts+edgeJoin/alphas/0.75/stats_40.csv}{red};
			
			\nextgroupplot[
			title={Cuts+Thm.~\ref{proposition:subset-U-maximizer} \\ (Prop.~\ref{proposition:tractable-U})},
			title style={align=center, font=\small, yshift=-1ex},
			xmin=0,
			xmax=1,
			xtick distance=0.5,
			ymin=0,
			ymax=100,
			width=\plotwidth cm,
			height=\plotheight cm,
			yticklabels=\empty,
			xticklabels=\empty
			]
			\addonepersistencyplot{./data/RandomTransitivelyClosed/cuts+preorderJoin/alphas/0.25/stats_40.csv}{blue};
			\addonepersistencyplot{./data/RandomTransitivelyClosed/cuts+preorderJoin/alphas/0.50/stats_40.csv}{\darkgreen};
			\addonepersistencyplot{./data/RandomTransitivelyClosed/cuts+preorderJoin/alphas/0.75/stats_40.csv}{red};

			\nextgroupplot[
			title={Cor.~\ref{corollary:boecker-condition-strong} ($b = 0$)},
			title style={align=center, font=\small, yshift=-1ex},
			xmin=0,
			xmax=1,
			xtick distance=0.5,
			ymin=0,
			ymax=100,
			width=\plotwidth cm,
			height=\plotheight cm,
			xticklabels=\empty,
			yticklabel=\empty
			]
			\addzeropersistencyplot{./data/RandomTransitivelyClosed/bbkStrongCut/alphas/0.25/stats_40.csv}{blue};
			\addzeropersistencyplot{./data/RandomTransitivelyClosed/bbkStrongCut/alphas/0.50/stats_40.csv}{\darkgreen};
			\addzeropersistencyplot{./data/RandomTransitivelyClosed/bbkStrongCut/alphas/0.75/stats_40.csv}{red};
			
			\nextgroupplot[
			title={Cuts+Cor.~\ref{corollary:boecker-condition-strong} \\ ($b = 1$)},
			title style={align=center, font=\small, yshift=-1ex},
			xmin=0,
			xmax=1,
			xtick distance=0.5,
			ymin=0,
			ymax=100,
			width=\plotwidth cm,
			height=\plotheight cm,
			yticklabels=\empty,
			xticklabels=\empty
			]
			\addonepersistencyplot{./data/RandomTransitivelyClosed/cuts+bbkStrongJoin/alphas/0.25/stats_40.csv}{blue};
			\addonepersistencyplot{./data/RandomTransitivelyClosed/cuts+bbkStrongJoin/alphas/0.50/stats_40.csv}{\darkgreen};
			\addonepersistencyplot{./data/RandomTransitivelyClosed/cuts+bbkStrongJoin/alphas/0.75/stats_40.csv}{red};

			\nextgroupplot[
			ylabel={Runtime [s]},
			ymode=log,
			xlabel={$\alpha$},
			xmin=0,
			xmax=1,
			ymin=1e-4,
			ymax=1e3,
			width=\plotwidth cm,
			height=\plotheight cm,
			xtick distance=0.50,
			legend pos=north west,
			legend style={font=\tiny}
			]
			\addruntimeplot{./data/RandomTransitivelyClosed/edgeCut/alphas/0.25/stats_40.csv}{blue}
			\addruntimeplot{./data/RandomTransitivelyClosed/edgeCut/alphas/0.50/stats_40.csv}{\darkgreen}
			\addruntimeplot{./data/RandomTransitivelyClosed/edgeCut/alphas/0.75/stats_40.csv}{red}
			
			\nextgroupplot[
			ymode=log,
			xlabel={$\alpha$},
			xmin=0,
			xmax=1,
			ymin=1e-4,
			ymax=1e3,
			width=\plotwidth cm,
			height=\plotheight cm,
			xtick distance=0.50,
			yticklabels=\empty,
			legend pos=north west,
			legend style={font=\tiny}
			]
			\addruntimeplot{./data/RandomTransitivelyClosed/directedCut/alphas/0.25/stats_40.csv}{blue}
			\addlegendentry{$p_E = 0.25$}
			\addruntimeplot{./data/RandomTransitivelyClosed/directedCut/alphas/0.50/stats_40.csv}{\darkgreen}
			\addlegendentry{$p_E = 0.50$}
			\addruntimeplot{./data/RandomTransitivelyClosed/directedCut/alphas/0.75/stats_40.csv}{red}
			\addlegendentry{$p_E = 0.75$}

			\nextgroupplot[
			ymode=log,
			xlabel={$\alpha$},
			xmin=0,
			xmax=1,
			ymin=1e-4,
			ymax=1e3,
			width=\plotwidth cm,
			height=\plotheight cm,
			xtick distance=0.50,
			yticklabel=\empty
			]
			\addruntimeplot{./data/RandomTransitivelyClosed/cuts+edgeJoin/alphas/0.25/stats_40.csv}{blue}
			\addruntimeplot{./data/RandomTransitivelyClosed/cuts+edgeJoin/alphas/0.50/stats_40.csv}{\darkgreen}
			\addruntimeplot{./data/RandomTransitivelyClosed/cuts+edgeJoin/alphas/0.75/stats_40.csv}{red}

			\nextgroupplot[
			ymode=log,
			xlabel={$\alpha$},
			xmin=0,
			xmax=1,
			ymin=1e-4,
			ymax=1e3,
			width=\plotwidth cm,
			height=\plotheight cm,
			xtick distance=0.50,
			yticklabel=\empty
			]
			\addruntimeplot{./data/RandomTransitivelyClosed/cuts+preorderJoin/alphas/0.25/stats_40.csv}{blue}
			\addruntimeplot{./data/RandomTransitivelyClosed/cuts+preorderJoin/alphas/0.50/stats_40.csv}{\darkgreen}
			\addruntimeplot{./data/RandomTransitivelyClosed/cuts+preorderJoin/alphas/0.75/stats_40.csv}{red}

			\nextgroupplot[
			ymode=log,
			xlabel={$\alpha$},
			xmin=0,
			xmax=1,
			ymin=1e-4,
			ymax=1e3,
			width=\plotwidth cm,
			height=\plotheight cm,
			xtick distance=0.50,
			yticklabel=\empty
			]
			\addruntimeplot{./data/RandomTransitivelyClosed/bbkStrongCut/alphas/0.25/stats_40.csv}{blue}
			\addruntimeplot{./data/RandomTransitivelyClosed/bbkStrongCut/alphas/0.50/stats_40.csv}{\darkgreen}
			\addruntimeplot{./data/RandomTransitivelyClosed/bbkStrongCut/alphas/0.75/stats_40.csv}{red}

			\nextgroupplot[
			ymode=log,
			xlabel={$\alpha$},
			xmin=0,
			xmax=1,
			ymin=1e-4,
			ymax=1e3,
			width=\plotwidth cm,
			height=\plotheight cm,
			xtick distance=0.50,
			yticklabel=\empty
			]
			\addruntimeplot{./data/RandomTransitivelyClosed/cuts+bbkStrongJoin/alphas/0.25/stats_40.csv}{blue}
			\addruntimeplot{./data/RandomTransitivelyClosed/cuts+bbkStrongJoin/alphas/0.50/stats_40.csv}{\darkgreen}
			\addruntimeplot{./data/RandomTransitivelyClosed/cuts+bbkStrongJoin/alphas/0.75/stats_40.csv}{red}
			
		\end{groupplot}
	\end{tikzpicture}

%% file: figure-synthetic-alphas.tex
\tikzmath{\plotheight=3.2;}
\tikzmath{\plotwidth=3.2;}

\setlength{\hspacing}{0.15cm}
\setlength{\vspacing}{0.50cm}

\providecommand{\addvariableplot}{}	
\renewcommand{\addvariableplot}[2]{\addplot+[
	only marks,
	mark=*,
	mark options={scale=0.25},
	color=#2
	] table[
	col sep=comma,
	x expr=\thisrow{alpha}, 
	y expr=\thisrow{medianEliminatedVariables}*100,
	y error minus expr=(\thisrow{medianEliminatedVariables} - \thisrow{q25EliminatedVariables})*100,
	y error plus expr=(\thisrow{q75EliminatedVariables} - \thisrow{medianEliminatedVariables})*100
	] {#1};}

\providecommand{\addruntimeplot}{}	
\renewcommand{\addruntimeplot}[2]{\addplot+[
	only marks,
	mark=*,
	mark options={scale=0.25},
	color=#2
	] table[
	col sep=comma,
	x expr=\thisrow{alpha}, 
	y expr=\thisrow{medianDuration} / 1e9,
	y error minus expr=(\thisrow{medianDuration} - \thisrow{q25Duration}) / 1e9,
	y error plus expr=(\thisrow{q75Duration} - \thisrow{medianDuration}) / 1e9
	] {#1};}

	\tikzsetnextfilename{fig-synthetic-alphas}
	\begin{tikzpicture}		
		\begin{groupplot}[group style={group size=3 by 2, horizontal sep=\hspacing, vertical sep=\vspacing}]

			\nextgroupplot[
			title={$p_E = 0.25$},
			title style={align=center, font=\small, yshift=-1ex},
			xmin=0,
			xmax=1,
			xtick distance=0.5,
			ymin=0,
			ymax=100,
			width=\plotwidth cm,
			height=\plotheight cm,
			ylabel={$\frac{\lvert \dom{\hat{x}}\rvert}{\lvert P_V\rvert}\left[\%\right]$},
			legend to name=TCLAlphasLegend,
			legend style={legend columns=4, font=\scriptsize},
			ylabel style={yshift=-6pt},
			legend pos=south west,
			xticklabels=\empty
			]
			\addvariableplot{./data/RandomTransitivelyClosed/allConditions/alphas/0.25/stats_20.csv}{blue}
			\addlegendentry{$\lvert V\rvert = 20$}
			\addvariableplot{./data/RandomTransitivelyClosed/allConditions/alphas/0.25/stats_30.csv}{\darkgreen}
			\addlegendentry{$\lvert V\rvert = 30$}
			\addvariableplot{./data/RandomTransitivelyClosed/allConditions/alphas/0.25/stats_40.csv}{red}
			\addlegendentry{$\lvert V\rvert = 40$}
			\addvariableplot{./data/RandomTransitivelyClosed/allConditions/alphas/0.25/stats_50.csv}{magenta}
			\addlegendentry{$\lvert V\rvert = 50$}

			\nextgroupplot[
			title style={align=center, font=\small, yshift=-1ex},
			xmin=0,
			xmax=1,
			ymin=0,
			ymax=100,
			width=\plotwidth cm,
			height=\plotheight cm,
			title={$p_E = 0.50$},
			yticklabels=\empty,
			xtick distance=0.5,
			xticklabels=\empty
			]
			\addvariableplot{./data/RandomTransitivelyClosed/allConditions/alphas/0.50/stats_20.csv}{blue}
			\addvariableplot{./data/RandomTransitivelyClosed/allConditions/alphas/0.50/stats_30.csv}{\darkgreen}
			\addvariableplot{./data/RandomTransitivelyClosed/allConditions/alphas/0.50/stats_40.csv}{red}
			\addvariableplot{./data/RandomTransitivelyClosed/allConditions/alphas/0.50/stats_50.csv}{magenta}
			
			\nextgroupplot[
			title={$p_E = 0.75$},
			title style={align=center, font=\small, yshift=-1ex},
			xmin=0,
			xmax=1,
			ymin=0,
			ymax=100,
			width=\plotwidth cm,
			height=\plotheight cm,
			yticklabels=\empty,
			xtick distance=0.5,
			xticklabels=\empty
			]
			\addvariableplot{./data/RandomTransitivelyClosed/allConditions/alphas/0.75/stats_20.csv}{blue}
			\addvariableplot{./data/RandomTransitivelyClosed/allConditions/alphas/0.75/stats_30.csv}{\darkgreen}
			\addvariableplot{./data/RandomTransitivelyClosed/allConditions/alphas/0.75/stats_40.csv}{red}
			\addvariableplot{./data/RandomTransitivelyClosed/allConditions/alphas/0.75/stats_50.csv}{magenta}

			\nextgroupplot[
			ylabel={Runtimes [s]},
			title style={align=center},
			xmin=0,
			xmax=1,
			ymin=1e-2,
			ymax=1e3,
			ymode=log,
			width=\plotwidth cm,
			height=\plotheight cm,
			ylabel style={yshift=-6pt},
			xlabel=$\alpha$,
			xtick distance=0.5,
			xlabel style={yshift=4pt}
			]
			
			\addruntimeplot{./data/RandomTransitivelyClosed/allConditions/alphas/0.25/stats_20.csv}{blue}
			\addruntimeplot{./data/RandomTransitivelyClosed/allConditions/alphas/0.25/stats_30.csv}{\darkgreen}
			\addruntimeplot{./data/RandomTransitivelyClosed/allConditions/alphas/0.25/stats_40.csv}{red}
			\addruntimeplot{./data/RandomTransitivelyClosed/allConditions/alphas/0.25/stats_50.csv}{magenta}
			
			\nextgroupplot[
			title style={align=center},
			xmin=0,
			xmax=1,
			ymode=log,
			ymin=1e-2,
			ymax=1e3,
			width=\plotwidth cm,
			height=\plotheight cm,
			yticklabels=\empty,
			xlabel=$\alpha$,
			xtick distance=0.5,
			xlabel style={yshift=4pt}
			]
			
			\addruntimeplot{./data/RandomTransitivelyClosed/allConditions/alphas/0.50/stats_20.csv}{blue}
			\addruntimeplot{./data/RandomTransitivelyClosed/allConditions/alphas/0.50/stats_30.csv}{\darkgreen}
			\addruntimeplot{./data/RandomTransitivelyClosed/allConditions/alphas/0.50/stats_40.csv}{red}
			\addruntimeplot{./data/RandomTransitivelyClosed/allConditions/alphas/0.50/stats_50.csv}{magenta}

			\nextgroupplot[
			ymode=log,
			title style={align=center},
			xmin=0,
			xmax=1,
			ymin=1e-2,
			ymax=1e3,
			width=\plotwidth cm,
			height=\plotheight cm,
			yticklabels=\empty,
			xlabel=$\alpha$,
			xtick distance=0.5,
			xlabel style={yshift=4pt}
			]
			
			\addruntimeplot{./data/RandomTransitivelyClosed/allConditions/alphas/0.75/stats_20.csv}{blue}
			\addruntimeplot{./data/RandomTransitivelyClosed/allConditions/alphas/0.75/stats_30.csv}{\darkgreen}
			\addruntimeplot{./data/RandomTransitivelyClosed/allConditions/alphas/0.75/stats_40.csv}{red}
			\addruntimeplot{./data/RandomTransitivelyClosed/allConditions/alphas/0.75/stats_50.csv}{magenta}

		\end{groupplot}
		
		\path (group c1r2.south west) -- node[below=0.75cm]{\ref{TCLAlphasLegend}} (group c3r2.south east);
	\end{tikzpicture}

%% file: figure-synthetic-ns.tex
\tikzmath{\plotheight=3.2;}
\tikzmath{\plotwidth=3.2;}

\setlength{\hspacing}{0.25cm}
\setlength{\vspacing}{0.50cm}

\providecommand{\addvariableplot}{}	
\renewcommand{\addvariableplot}[2]{\addplot+[
	only marks,
	mark=*,
	mark options={scale=0.25},
	color=#2
	] table[
	col sep=comma,
	x expr=\thisrow{numberOfPoints}, 
	y expr=\thisrow{medianEliminatedVariables}*100,
	y error minus expr=(\thisrow{medianEliminatedVariables} - \thisrow{q25EliminatedVariables})*100,
	y error plus expr=(\thisrow{q75EliminatedVariables} - \thisrow{medianEliminatedVariables})*100
	] {#1};}

\providecommand{\addruntimeplot}{}
\renewcommand{\addruntimeplot}[3][]{\addplot+[
	only marks,
	mark=*,
	mark options={scale=0.25},
	color=#3,
	#1
	] table[
	col sep=comma,
	x expr=\thisrow{numberOfPoints}, 
	y expr=\thisrow{medianDuration} / 1e9,
	y error minus expr=(\thisrow{medianDuration} - \thisrow{q25Duration}) / 1e9,
	y error plus expr=(\thisrow{q75Duration} - \thisrow{medianDuration}) / 1e9
	] {#2};}

	\tikzsetnextfilename{fig-synthetic-ns}
	\begin{tikzpicture}		
		\begin{groupplot}[group style={group size=3 by 2, horizontal sep=\hspacing, vertical sep=\vspacing}]

			\nextgroupplot[
			title={$p_E = 0.25$},
			title style={align=center, font=\small, yshift=-1ex},
			xmin=0,
			xmax=110,
			xtick distance=50,
			ymin=0,
			ymax=100,
			width=\plotwidth cm,
			height=\plotheight cm,
			ylabel={$\frac{\lvert \dom{\hat{x}}\rvert}{\lvert P_V\rvert}\left[\%\right]$},
			legend to name=TCLNumberOfVerticesLegend,
			legend style={legend columns=4, font=\scriptsize},
			ylabel style={yshift=-6pt}
			]
			\addvariableplot{./data/RandomTransitivelyClosed/allConditions/ns/0.25/stats_0.25.csv}{blue}
			\addlegendentry{$\alpha = 0.25$}
			\addvariableplot{./data/RandomTransitivelyClosed/allConditions/ns/0.25/stats_0.65.csv}{\darkgreen}
			\addlegendentry{$\alpha = 0.65$}
			\addvariableplot{./data/RandomTransitivelyClosed/allConditions/ns/0.25/stats_0.70.csv}{red}
			\addlegendentry{$\alpha = 0.70$}
			\addvariableplot{./data/RandomTransitivelyClosed/allConditions/ns/0.25/stats_0.75.csv}{magenta}
			\addlegendentry{$\alpha = 0.75$}

			\nextgroupplot[
			title style={align=center, font=\small, yshift=-1ex},
			xmin=0,
			xmax=110,
			ymin=0,
			ymax=100,
			width=\plotwidth cm,
			height=\plotheight cm,
			title={$p_E = 0.50$},
			yticklabels=\empty,
			xtick distance=50
			]
			\addvariableplot{./data/RandomTransitivelyClosed/allConditions/ns/0.50/stats_0.25.csv}{blue}
			\addvariableplot{./data/RandomTransitivelyClosed/allConditions/ns/0.50/stats_0.65.csv}{\darkgreen}
			\addvariableplot{./data/RandomTransitivelyClosed/allConditions/ns/0.50/stats_0.70.csv}{red}
			\addvariableplot{./data/RandomTransitivelyClosed/allConditions/ns/0.50/stats_0.75.csv}{magenta}
			
			\nextgroupplot[
			title={$p_E = 0.75$},
			title style={align=center, font=\small, yshift=-1ex},
			xmin=0,
			xmax=110,
			ymin=0,
			ymax=100,
			width=\plotwidth cm,
			height=\plotheight cm,
			yticklabels=\empty,
			xtick distance=50
			]
			\addvariableplot{./data/RandomTransitivelyClosed/allConditions/ns/0.75/stats_0.25.csv}{blue}
			\addvariableplot{./data/RandomTransitivelyClosed/allConditions/ns/0.75/stats_0.65.csv}{\darkgreen}
			\addvariableplot{./data/RandomTransitivelyClosed/allConditions/ns/0.75/stats_0.70.csv}{red}
			\addvariableplot{./data/RandomTransitivelyClosed/allConditions/ns/0.75/stats_0.75.csv}{magenta}

			\nextgroupplot[
			xmode=log,
			ymode=log,
			ylabel={Runtimes [s]},
			title style={align=center},
			xmin=0,
			xmax=110,
			ymin=1e-5,
			ymax=1000,
			width=\plotwidth cm,
			height=\plotheight cm,
			ylabel style={yshift=-6pt},
			xlabel=$\vert V\vert$,
			xlabel style={yshift=4pt},
			legend pos=south east,
			legend style={
				font=\scriptsize,
				draw=none,  
				fill opacity=0.9
			},
			legend image post style={xscale=0.25}
			]
			
			\addruntimeplot[forget plot]{./data/RandomTransitivelyClosed/allConditions/ns/0.25/stats_0.25.csv}{blue}
			\addruntimeplot[forget plot]{./data/RandomTransitivelyClosed/allConditions/ns/0.25/stats_0.65.csv}{\darkgreen}
			\addruntimeplot[forget plot]{./data/RandomTransitivelyClosed/allConditions/ns/0.25/stats_0.70.csv}{red}
			\addruntimeplot[forget plot]{./data/RandomTransitivelyClosed/allConditions/ns/0.25/stats_0.75.csv}{magenta}
			\addplot[
			domain=1:100,      
			samples=2,     
			color=magenta,
			thick,
			]
			{2.2356e-07 * x^4.6356};
			\addlegendentry{$O(\lvert V\rvert^{4.6})$}

			\nextgroupplot[
			xmode=log,
			ymode=log,
			title style={align=center},
			xmin=0,
			xmax=110,
			ymin=1e-5,
			ymax=1000,
			width=\plotwidth cm,
			height=\plotheight cm,
			yticklabels=\empty,
			xlabel=$\vert V\vert$,
			xlabel style={yshift=4pt},
			legend pos=south east,
			legend style={
				font=\scriptsize,
				draw=none,  
				fill opacity=0.9
			},
			legend image post style={xscale=0.25}
			]
			
			\addruntimeplot[forget plot]{./data/RandomTransitivelyClosed/allConditions/ns/0.50/stats_0.25.csv}{blue}
			\addruntimeplot[forget plot]{./data/RandomTransitivelyClosed/allConditions/ns/0.50/stats_0.65.csv}{\darkgreen}
			\addruntimeplot[forget plot]{./data/RandomTransitivelyClosed/allConditions/ns/0.50/stats_0.70.csv}{red}
			\addruntimeplot[forget plot]{./data/RandomTransitivelyClosed/allConditions/ns/0.50/stats_0.75.csv}{magenta}
			\addplot[
			domain=1:100,      
			samples=2,     
			color=magenta,
			thick,
			]
			{3.7423e-07 * x^4.4911};
			\addlegendentry{$O(\lvert V\rvert^{4.5})$}

			\nextgroupplot[
			xmode=log,
			ymode=log,
			title style={align=center},
			xmin=0,
			xmax=110,
			ymin=1e-5,
			ymax=1000,
			width=\plotwidth cm,
			height=\plotheight cm,
			yticklabels=\empty,
			xlabel=$\vert V\vert$,
			xlabel style={yshift=4pt},
			legend pos=south east,
			legend style={
				font=\scriptsize,
				draw=none,  
				fill opacity=0.9
			},
			legend image post style={xscale=0.25}
			]
			
			\addruntimeplot[forget plot]{./data/RandomTransitivelyClosed/allConditions/ns/0.75/stats_0.25.csv}{blue}
			\addruntimeplot[forget plot]{./data/RandomTransitivelyClosed/allConditions/ns/0.75/stats_0.65.csv}{\darkgreen}
			\addruntimeplot[forget plot]{./data/RandomTransitivelyClosed/allConditions/ns/0.75/stats_0.70.csv}{red}
			\addruntimeplot[forget plot]{./data/RandomTransitivelyClosed/allConditions/ns/0.75/stats_0.75.csv}{magenta}
			\addplot[
			domain=1:100,      
			samples=2,     
			color=magenta,
			thick,
			]
			{1.2243e-07 * x^4.9308};
			\addlegendentry{$O(\lvert V\rvert^{4.9})$}

		\end{groupplot}
		
		\path (group c1r2.south west) -- node[below=0.75cm]{\ref{TCLNumberOfVerticesLegend}} (group c3r2.south east);
	\end{tikzpicture}

%% file: figure-twitter.tex
\tikzmath{\plotheight=3.2;}
\tikzmath{\plotwidth=3.2;}

\setlength{\hspacing}{0.10cm}
\setlength{\vspacing}{0.50cm}

\providecommand{\addzeropersistencyplot}{}	
\renewcommand{\addzeropersistencyplot}[1]{\addplot+[
	only marks,
	mark=*,
	mark options={scale=0.25},
	color=blue
	] table[
	col sep=comma,
	x expr=\thisrow{initialNumberOfVertices}, 
	y expr=\thisrow{persistentZeroEdges} / \thisrow{initialNumberOfEdges}*100
	] {#1};}
	
\providecommand{\addpersistencyplot}{}	
\renewcommand{\addpersistencyplot}[1]{\addplot+[
	only marks,
	mark=*,
	mark options={scale=0.25},
	color=blue
	] table[
	col sep=comma,
	x expr=\thisrow{initialNumberOfVertices}, 
	y expr=(\thisrow{persistentZeroEdges} + \thisrow{persistentOneEdges}) / \thisrow{initialNumberOfEdges}*100
	] {#1};}

\providecommand{\addonepersistencyplot}{}	
\renewcommand{\addonepersistencyplot}[1]{\addplot+[
	only marks,
	mark=*,
	mark options={scale=0.25},
	color=blue
	] table[
	col sep=comma,
	x expr=\thisrow{initialNumberOfVertices}, 
	y expr=\thisrow{persistentOneEdges} / \thisrow{initialNumberOfEdges}*100
	] {#1};}

\providecommand{\addruntimeplot}{}	
\renewcommand{\addruntimeplot}[2][]{\addplot+[
	only marks,
	mark=*,
	mark options={scale=0.25, color=black},
	#1
	] table[
	col sep=comma,
	x expr=\thisrow{initialNumberOfVertices}, 
	y expr=\thisrow{runtimeNanoseconds} / 1e9
	] {#2};}

	\tikzsetnextfilename{fig-twitter}
	\begin{tikzpicture}		
		\begin{groupplot}[group style={group size= 14 by 2, horizontal sep=\hspacing, vertical sep=\vspacing}]

			\nextgroupplot[
			title={\Cref{proposition:edge-cut-condition}},
			ylabel={$\frac{\lvert \dom{\hat{x}}\rvert}{\lvert P_V\rvert} \left[\%\right]$},
			title style={align=center, font=\small, yshift=-1ex},
			xmin=0,
			xmax=250,
			xtick distance=100,
			ymin=0,
			ymax=100,
			width=\plotwidth cm,
			height=\plotheight cm
			]
			\addzeropersistencyplot{./data/twitter/edgeCut/persistencies.csv};
			
			\nextgroupplot[
			width=0.525*\plotwidth cm,
			height=\plotheight cm,
			enlarge x limits=0.5,
			ymin=0,
			ymax=100,
			xtick=data,
			hide axis,
			xshift=-\hspacing,
			ylabel shift=-1ex
			]
			
			\addplot[
			boxplot prepared from table={
						table=./data/twitter/edgeCut/boxplot-0.csv,
						lower whisker=lw,
						upper whisker=uw,
						lower quartile=lq,
						upper quartile=uq,
						median=med
				},
				boxplot prepared,
				boxplot/draw direction=y
			] coordinates {};
			
			\nextgroupplot[
			title={Cor.~\ref{corollary:directed-cut-condition}},
			title style={align=center, font=\small, yshift=-1ex},
			xmin=0,
			xmax=250,
			xtick distance=100,
			ymin=0,
			ymax=100,
			width=\plotwidth cm,
			height=\plotheight cm,
			yticklabels=\empty
			]
			\addzeropersistencyplot{./data/twitter/directedCut/persistencies.csv};
			
			\nextgroupplot[
			width=0.525*\plotwidth cm,
			height=\plotheight cm,
			enlarge x limits=0.5,
			ymin=0,
			ymax=100,
			xtick=data,
			hide axis,
			xshift=-\hspacing
			]
			\addplot[
			boxplot prepared from table={
				table=./data/twitter/directedCut/boxplot-0.csv,
				lower whisker=lw,
				upper whisker=uw,
				lower quartile=lq,
				upper quartile=uq,
				median=med
			},
			boxplot prepared,
			boxplot/draw direction=y
			] coordinates {};
			
			\nextgroupplot[
			title={Cuts+\Cref{proposition:edge-join-condition}},
			title style={align=center, font=\small, yshift=-1ex},
			xmin=0,
			xmax=250,
			xtick distance=100,
			ymin=0,
			ymax=100,
			width=\plotwidth cm,
			height=\plotheight cm,
			yticklabels=\empty
			]
			\addonepersistencyplot{./data/twitter/cuts+edgeJoin/persistencies.csv};
			
			\nextgroupplot[
			width=0.525*\plotwidth cm,
			height=\plotheight cm,
			enlarge x limits=0.5,
			ymin=0,
			ymax=100,
			xtick=data,
			hide axis,
			xshift=-\hspacing
			]
			\addplot[
			boxplot prepared from table={
				table=./data/twitter/cuts+edgeJoin/boxplot-1.csv,
				lower whisker=lw,
				upper whisker=uw,
				lower quartile=lq,
				upper quartile=uq,
				median=med
			},
			boxplot prepared,
			boxplot/draw direction=y
			] coordinates {};

			\nextgroupplot[
			title={Cuts+\Cref{proposition:subset-U-maximizer} \\ (\Cref{proposition:tractable-U})},
			title style={align=center, font=\small, yshift=-1ex},
			xmin=0,
			xmax=250,
			xtick distance=100,
			ymin=0,
			ymax=100,
			width=\plotwidth cm,
			height=\plotheight cm,
			yticklabels=\empty
			]
			\addonepersistencyplot{./data/twitter/cuts+preorderJoinExact/persistencies.csv};
			
			\nextgroupplot[
			width=0.525*\plotwidth cm,
			height=\plotheight cm,
			enlarge x limits=0.5,
			ymin=0,
			ymax=100,
			xtick=data,
			hide axis,
			xshift=-\hspacing
			]
			\addplot[
			boxplot prepared from table={
				table=./data/twitter/cuts+preorderJoinExact/boxplot-1.csv,
				lower whisker=lw,
				upper whisker=uw,
				lower quartile=lq,
				upper quartile=uq,
				median=med
			},
			boxplot prepared,
			boxplot/draw direction=y
			] coordinates {};

			\nextgroupplot[
			title={Cor.~\ref{corollary:boecker-condition-strong} \\ ($b = 0$)},
			title style={align=center, font=\small, yshift=-1ex},
			xmin=0,
			xmax=250,
			xtick distance=100,
			ymin=0,
			ymax=100,
			width=\plotwidth cm,
			height=\plotheight cm,
			yticklabels=\empty,
			]
			\addzeropersistencyplot{./data/twitter/bbkStrongCutConditions/persistencies.csv};
			
			\nextgroupplot[
			width=0.525*\plotwidth cm,
			height=\plotheight cm,
			enlarge x limits=0.5,
			ymin=0,
			ymax=100,
			xtick=data,
			hide axis,
			xshift=-\hspacing
			]
			\addplot[
			boxplot prepared from table={
				table=./data/twitter/bbkStrongCutConditions/boxplot-0.csv,
				lower whisker=lw,
				upper whisker=uw,
				lower quartile=lq,
				upper quartile=uq,
				median=med
			},
			boxplot prepared,
			boxplot/draw direction=y
			] coordinates {};

			\nextgroupplot[
			title={Cuts+Cor.~\ref{corollary:boecker-condition-strong} \\ ($b = 1$) },
			title style={align=center, font=\small, yshift=-1ex},
			xmin=0,
			xmax=250,
			xtick distance=100,
			ymin=0,
			ymax=100,
			width=\plotwidth cm,
			height=\plotheight cm,
			yticklabels=\empty
			]
			\addonepersistencyplot{./data/twitter/cuts+bbkStrongJoinConditions/persistencies.csv};
			
			\nextgroupplot[
			width=0.525*\plotwidth cm,
			height=\plotheight cm,
			enlarge x limits=0.5,
			ymin=0,
			ymax=100,
			xtick=data,
			hide axis,
			xshift=-\hspacing
			]
			\addplot[
			boxplot prepared from table={
				table=./data/twitter/cuts+bbkStrongJoinConditions/boxplot-1.csv,
				lower whisker=lw,
				upper whisker=uw,
				lower quartile=lq,
				upper quartile=uq,
				median=med
			},
			boxplot prepared,
			boxplot/draw direction=y
			] coordinates {};

			\nextgroupplot[
			title={Joint},
			title style={align=center, font=\small, yshift=-1ex},
			xmin=0,
			xmax=250,
			xtick distance=100,
			ymin=0,
			ymax=100,
			width=\plotwidth cm,
			height=\plotheight cm,
			yticklabels=\empty
			]
			\addpersistencyplot{./data/twitter/allConditions/persistencies.csv};

			\nextgroupplot[
			width=0.525*\plotwidth cm,
			height=\plotheight cm,
			enlarge x limits=0.5,
			ymin=0,
			ymax=100,
			xtick=data,
			hide axis,
			xshift=-\hspacing
			]
			\addplot[
			boxplot prepared from table={
				table=./data/twitter/allConditions/boxplot-all.csv,
				lower whisker=lw,
				upper whisker=uw,
				lower quartile=lq,
				upper quartile=uq,
				median=med
			},
			boxplot prepared,
			boxplot/draw direction=y
			] coordinates {};

			\nextgroupplot[
			ylabel={Runtime [s]},
			ymode=log,
			xmode=log,
			xlabel={$\vert V\vert$},
			xmin=0,
			xmax=250,
			ymin=1e-5,
			ymax=1e5,
			width=\plotwidth cm,
			height=\plotheight cm,
			legend pos=south east,
			legend style={
				font=\scriptsize,
				draw=none,  
				fill opacity=0.9
				},
			legend image post style={xscale=0.25},
			ylabel shift=-1ex
			]
			\addruntimeplot[forget plot]{./data/twitter/edgeCut/persistencies.csv}
			\addplot[
			domain=1:250,      
			samples=2,     
			color=black,
			thick,
			]
			{6.3308e-09 * x^4.4844};
			\addlegendentry{$O(\lvert V\rvert^{4.5})$}

			\nextgroupplot[
			width=0.5*\plotwidth cm,
			height=\plotheight cm,
			hide axis
			]

			\nextgroupplot[
			ymode=log,
			xmode=log,
			xlabel={$\vert V\vert$},
			xmin=0,
			xmax=250,
			ymin=1e-5,
			ymax=1e5,
			width=\plotwidth cm,
			height=\plotheight cm,
			yticklabels=\empty,
			legend pos=south east,
			legend style={
				font=\scriptsize,
				draw=none,  
				fill opacity=0.9
			},
			legend image post style={xscale=0.25}
			]
			\addruntimeplot[forget plot]{./data/twitter/directedCut/persistencies.csv}
			\addplot[
			domain=1:250,      
			samples=2,     
			color=black,
			thick,
			]
			{6.1000e-08 * x^2.4523};
			\addlegendentry{$O(\lvert V\rvert^{2.5})$}

			\nextgroupplot[
			width=0.5*\plotwidth cm,
			height=\plotheight cm,
			hide axis
			]
						
			\nextgroupplot[
			ymode=log,
			xmode=log,
			xlabel={$\vert V\vert$},
			xmin=0,
			xmax=250,
			ymin=1e-5,
			ymax=1e5,
			width=\plotwidth cm,
			height=\plotheight cm,
			yticklabel=\empty,
			legend pos=south east,
			legend style={
				font=\scriptsize,
				draw=none,  
				fill opacity=0.9
			},
			legend image post style={xscale=0.25}
			]
			\addruntimeplot[forget plot]{./data/twitter/cuts+edgeJoin/persistencies.csv}
			\addplot[
			domain=1:250,      
			samples=2,     
			color=black,
			thick,
			]
			{6.3792e-09 * x^5.0883};
			\addlegendentry{$O(\lvert V\rvert^{5.1})$}

			\nextgroupplot[
			width=0.5*\plotwidth cm,
			height=\plotheight cm,
			hide axis
			]
			
			\nextgroupplot[
			ymode=log,
			xmode=log,
			xlabel={$\vert V\vert$},
			xmin=0,
			xmax=250,
			ymin=1e-5,
			ymax=1e5,
			width=\plotwidth cm,
			height=\plotheight cm,
			yticklabel=\empty,
			legend pos=south east,
			legend style={
				font=\scriptsize,
				draw=none,  
				fill opacity=0.9
			},
			legend image post style={xscale=0.25}
			]
			\addruntimeplot[forget plot]{./data/twitter/cuts+preorderJoinExact/persistencies.csv}
			\addplot[
			domain=1:250,      
			samples=2,     
			color=black,
			thick,
			]
			{2.4255e-08 * x^4.4736};
			\addlegendentry{$O(\lvert V\rvert^{4.5})$}
			
			\nextgroupplot[
			width=0.5*\plotwidth cm,
			height=\plotheight cm,
			hide axis
			]

			\nextgroupplot[
			ymode=log,
			xmode=log,
			xlabel={$\vert V\vert$},
			xmin=0,
			xmax=250,
			ymin=1e-5,
			ymax=1e5,
			width=\plotwidth cm,
			height=\plotheight cm,
			yticklabels=\empty,
			legend pos=south east,
			legend style={
				font=\scriptsize,
				draw=none,  
				fill opacity=0.9
			},
			legend image post style={xscale=0.25}
			]
			\addruntimeplot[forget plot]{./data/twitter/bbkStrongCutConditions/persistencies.csv}
			\addplot[
			domain=1:250,      
			samples=2,     
			color=black,
			thick,
			]
			{5.5356e-08 * x^3.7423};
			\addlegendentry{$O(\lvert V\rvert^{3.7})$}

			\nextgroupplot[
			width=0.5*\plotwidth cm,
			height=\plotheight cm,
			hide axis
			]

			\nextgroupplot[
			ymode=log,
			xmode=log,
			xlabel={$\vert V\vert$},
			xmin=0,
			xmax=250,
			ymin=1e-5,
			ymax=1e5,
			width=\plotwidth cm,
			height=\plotheight cm,
			yticklabel=\empty,
			legend pos=south east,
			legend style={
				font=\scriptsize,
				draw=none,  
				fill opacity=0.9
			},
			legend image post style={xscale=0.25}
			]
			\addruntimeplot[forget plot]{./data/twitter/cuts+bbkStrongJoinConditions/persistencies.csv}
			\addplot[
			domain=1:250,      
			samples=2,     
			color=black,
			thick,
			]
			{2.9723e-08 * x^4.2394};
			\addlegendentry{$O(\lvert V\rvert^{4.2})$}

			\nextgroupplot[
			width=0.5*\plotwidth cm,
			height=\plotheight cm,
			hide axis
			]
			
			\nextgroupplot[
			ymode=log,
			xmode=log,
			xlabel={$\vert V\vert$},
			xmin=0,
			xmax=250,
			ymin=1e-5,
			ymax=1e5,
			width=\plotwidth cm,
			height=\plotheight cm,
			yticklabel=\empty,
			legend pos=south east,
			legend style={
				font=\scriptsize,
				draw=none,  
				fill opacity=0.9
			},
			legend image post style={xscale=0.25}
			]
			\addruntimeplot[forget plot]{./data/twitter/allConditions/persistencies.csv}
				\addplot[
			domain=1:250,      
			samples=2,     
			color=black,
			thick,
			]
			{1.7194e-08 * x^5.1053};
			\addlegendentry{$O(\lvert V\rvert^{5.1})$}

			\nextgroupplot[
			width=0.5*\plotwidth cm,
			height=\plotheight cm,
			hide axis
			]	
			
		\end{groupplot}
	\end{tikzpicture}

%% file: section-conclusion.tex
\vspace{-1ex}
\section{Conclusion}\label{section:conclusion}

We introduce new exact and heuristic algorithms that take as input an instance of the NP-hard preordering problem, terminate in polynomial time, and output a partial solution that is extendable to an optimal solution.
These algorithms decide partial optimality conditions that we prove, building on prior work on partially optimal preordering, clustering and ordering.
In experiments with synthetic data, we observe how the effectiveness of the algorithms transitions from fixing all variables, for easy instances, to not fixing any variables, for hard instances.
In experiments with published social networks, we observe that algorithms based on new conditions we propose are more effective in fixing variables to zero than algorithms that decide conditions we transfer from prior work, at the cost of longer but still polynomial runtime.

%% file: appendix-equivalence-classes.tex
\section{Merging of Equivalence Classes}
\label{appendix:merging-of-equivalence-classes}

If there exists a subset $U \subseteq V$ such that $P_U \subseteq \hat{x}^{-1}(1)$, then the elements of $U$ can be joined:

\begin{proposition}
	Let $V\neq \emptyset$ finite and $c\in \mathbb{R}^{P_V}$. 
	Moreover, let $\hat{x}\in \hat{X}_V$ and $U \subseteq V$ such that $P_U \subseteq \hat{x}^{-1}(1)$. 
	Let $V' \coloneqq \left(V\setminus U\right) \cup \{U\}$ and $c'\in \mathbb{R}^{P_{V'}}$ such that
	\begin{align}
		\forall pq\in P_{V\setminus U} &\colon \quad c'_{pq} = c_{pq}\enspace,\\
		\forall p \in V\setminus U &\colon \quad c'_{pU} = \sum_{q\in U}c_{pq}\enspace,\\
		\forall q \in V\setminus U &\colon \quad c'_{Uq} = \sum_{p\in U}c_{pq}\enspace.
	\end{align}
	Furthermore, define $\hat{x}'\in \hat{X}_{V'}$ such that 
	\begin{align}
		\forall pq\in P_{V\setminus U} \colon \quad \hat{x}'_{pq} &= \hat{x}'_{pq} \enspace,\\
		\forall q \in V \setminus U \colon \quad \hat{x}'_{Uq} &= \begin{cases}
			1 & \textnormal{if $\forall p\in U \colon \hat{x}_{pq} = 1$} \\
			0 & \textnormal{if $\forall p\in U \colon \hat{x}_{pq} = 0$} \\
			* & \textnormal{otherwise}
		\end{cases} \enspace,\\
		\forall p \in V \setminus U\colon \quad \hat{x}'_{pU} &= \begin{cases}
			1 & \textnormal{if $\forall q\in U \colon \hat{x}_{pq} = 1$} \\
			0 & \textnormal{if $\forall q\in U \colon \hat{x}_{pq} = 0$} \\
			* & \textnormal{otherwise}
		\end{cases}\enspace.
	\end{align}
	Then 
	\begin{equation}
		\max_{x\in X_V[\hat{x}]} \varphi_c(x) = \sum_{e\in P_{U}}c_{e} + \max_{x'\in X_{V'}[\hat{x}']} \varphi_{c'}(x')\enspace.
	\end{equation}
\end{proposition}
\begin{proof}
	Let $x\in X_V[\hat{x}]$. We define a bijection $\varphi\colon X_V[\hat{x}] \to X_{V'}[\hat{x}']$ by the following identities:
	\begin{align}
		\forall pq\in P_{V\setminus U} &\colon\quad \varphi(x)_{pq} = x_{pq} \\
		\forall p\in V\setminus U &\colon\quad \varphi(x)_{pU} = x_{pq} && \textnormal{for any $q\in U$} \enspace, \\
		\forall q\in V\setminus U&\colon\quad \varphi(x)_{Uq} = x_{pq} && \textnormal{for any $p\in U$}\enspace.		
	\end{align}
	(1) $\varphi$ is well-defined, i.e., it follows that $x_{pq} = x_{pq'}$ for every $\{q, q'\}\in \tbinom U2$ and $p\in V\setminus U$, and that $x_{pq} = x_{p'q}$ for every $\{p, p'\}\in \tbinom U2$ and $q\in V\setminus U$, by transitivity and the fact that $x_{pq} = 1$ for all $pq\in P_U$.
	
	\noindent
	(2) Let $x' = \varphi(x)$. Now:
	\begin{align}
		\sum_{pq\in P_V}c_{pq} x_{pq}
		&=\sum_{pq\in P_U}c_{pq}x_{pq} + \sum_{p\in U}\sum_{q\in V\setminus U}c_{pq}x_{pq} 
		+ \sum_{p\in V\setminus U}\sum_{q\in U}c_{pq}x_{pq} + \sum_{pq\in P_{V\setminus U}}c_{pq}x_{pq} \\
		&= \sum_{pq\in P_U}c_{pq} + \sum_{q\in V\setminus U} \left(\sum_{p\in U}c_{pq}\right)x_{Uq} 
		+ \sum_{p\in V\setminus U}\left(\sum_{q\in U}c_{pq}\right)x_{pU} + \sum_{pq\in P_{V\setminus U}}c_{pq}x_{pq}\\
		&=\sum_{pq\in P_U}c_{pq} + \sum_{q\in V\setminus U} c'_{Uq} x_{Uq} 
		+ \sum_{p\in V\setminus U}c'_{pU} x_{pU} + \sum_{pq\in P_{V\setminus U}}c'_{pq}x_{pq} \\
		&= \sum_{pq\in P_U}c_{pq} + \sum_{pq\in P_{\left(V\setminus U\right)\cup \{U\}}}c'_{pq}x'_{pq} 
		= \sum_{pq\in P_U}c_{pq} + \sum_{pq\in P_{V'}}c'_{pq}x'_{pq}\enspace.
	\end{align}
	This concludes the proof.
\end{proof}%

%% file: appendix-energy-minimization.tex
\section{Applying~\Cref{proposition:edge-join-condition}}

\subsection{Energy Minimization Problem}\label{appendix:energy-minimization}

We cast the minimization of the rhs.~of~\eqref{eq:improvingness-join-generalization-map-1} as an energy minimization problem.
Without loss of generality, let $U'' = V\setminus \left(U \cup U'\right)$ for $U$ and $U'$ as in~\Cref{proposition:edge-join-condition}.
We consider labels $\mathcal{L}\coloneqq \{U, U', U''\}$, 
functions $H_p\colon \mathcal{L} \to \mathbb{R}$ for every $p\in V$
and functions $H_{pq} \colon\mathcal{L}\times \mathcal{L} \to \mathbb{R}$ for every $pq\in P_V$. 
Then, minimizing~\eqref{eq:improvingness-join-generalization-map-1} is equivalent to solving
\begin{equation}
	\label{eq:energy-minimization-optimization-problem}
	\min_{x\in \mathcal{L}^V} \sum_{pq\in P_V} H_{pq}(x_p, x_q) + \sum_{p\in V}H_p(x_p)\enspace,
\end{equation}
where $H_p(x_p) = 0$ for all $p\in V\setminus \{i, j\}$ and $x_p \in \mathcal{L}$,
$H_{pq}(x_p, x_q) = 0$ for all $pq\in P_V$ and $(x_p, x_q) \notin (U\times U')\cup (U'\times U) \cup (U''\times U)\cup (U'\times U'')$
and 
\begin{align}
	\label{eq:energy-minimization-costs-1}
	&H_{i}(U) = H_j(U') = 0\;,\; \\
	\label{eq:energy-minimization-costs-2}
	&H_i(U') = H_i(U'') = H_j(U) = H_j(U'') = \infty \\
	\label{eq:energy-minimization-costs-3}
	&\forall pq\in P_V \colon H_{pq}(U', U) = H_{pq}(U'', U) = H_{pq}(U', U'') 
	=\begin{cases}
		c_{pq}^+ & \textnormal{if $pq\notin \dom{\hat{x}}$}\\
		0 & \textnormal{if $\hat{x}_{pq}$}\\
		\infty & \textnormal{if $\hat{x}_{pq} = 1$}
	\end{cases} \\
	\label{eq:energy-minimization-costs-4}
	&\forall pq\in P_V \colon H_{pq}(U, U') 
	= \begin{cases}
		c_{pq}^- & \textnormal{if $\hat{x}_{pi} \neq 0 \neq \hat{x}_{jq}\land \hat{x}_{pq} \neq 1$}\\
		0 & \textnormal{if $\hat{x}_{pi} = 0 \lor \hat{x}_{jq} = 0$}\\
		\infty & \textnormal{if $\hat{x}_{pi} \neq 0 \neq \hat{x}_{jq}\land \hat{x}_{pq} = 0$}
	\end{cases}
\end{align}

%% file: appendix-alpha-beta-swap-moves.tex
\subsection{$\alpha\beta$-swaps}\label{appendix:alpha-beta-swaps}

Let $V\neq \emptyset$, $\mathcal{L}\neq\emptyset$, $H_{pq}\colon \mathcal{L}\times \mathcal{L}\to \mathbb{R}$ for every $pq\in P_V$ and $H_p\colon \mathcal{L}\to \mathbb{R}$ for every $p\in V$. 
For every $x\in \mathcal{L}^V$ we define the energy:
\begin{equation}
	\psi(x) = \sum_{pq\in P_V}H_{pq}(x_p, x_q) + \sum_{p \in V}H_p(x_p)\enspace.
\end{equation}
For simplicity, we assume that $H_{pq}(\gamma, \gamma) = 0$ for every $pq\in P_V$ and $\gamma\in \mathcal{L}$. 
This assumption is fulfilled for our value function~\eqref{eq:energy-minimization-costs-3}--\eqref{eq:energy-minimization-costs-4}.
For any $\{\alpha, \beta\} \in \tbinom{\mathcal{L}}{2}$, 
any
$x\in \mathcal{L}^V$ and any $U \subseteq x^{-1}(\alpha)\cup x^{-1}(\beta) \eqqcolon V_{\alpha\beta}$ we define
\begin{align}
	\operatorname{\alpha\beta-swap}_U(x)_p = \begin{cases}
		x_{p} & \textnormal{if $p\notin V_{\alpha\beta}$} \\
		\alpha & \textnormal{if $p\in U$}\\
		\beta & \textnormal{if $p\in V_{\alpha\beta}\setminus U$}
	\end{cases}\enspace.
\end{align}
Now:
\begin{align}
	&\quad \psi(\operatorname{\alpha\beta-swap}_U(x)) \\
	=&\sum_{p\in U}\Bigl( 
	\smashoperator[r]{\sum_{q\in V_{\alpha\beta}\setminus U}}H_{pq}(\alpha, \beta) 
	+ 
	\sum_{\mathclap{q\in V\setminus V_{\alpha\beta}}}H_{pq}(\alpha, x_q) + H_p(\alpha)
	\Bigr) \\
	&+ \sum_{p\in V_{\alpha\beta}\setminus U}
	\Bigl(
	\smashoperator[r]{\sum_{q\in U}}H_{pq}(\beta, \alpha) 
	+ 
	\sum_{\mathclap{q\in V\setminus V_{\alpha\beta}}}H_{pq}(\beta, x_q) + H_p(\beta)
	\Bigr) \\
	&+ \sum_{p\in V\setminus V_{\alpha\beta}}\Bigl(
	\sum_{q\in U}H_{pq}(x_p, \alpha)
	+ \sum_{q\in V_{\alpha\beta}\setminus U }H_{pq}(x_p, \beta)
	+ \sum_{\substack{q\in V\setminus V_{\alpha\beta} \\ q\neq p}}H_{pq}(x_p, x_q) + H_p(x_p)\Bigr)\\
	= &\sum_{p\in U}
	\Bigl(\smashoperator[r]{\sum_{q\in V\setminus V_{\alpha\beta}}} \left(H_{pq}(\alpha, x_q) + H_{qp}(x_q, \alpha)\right) + H_p(\alpha)\Bigr) \\
	&+ 
	\sum_{p\in V_{\alpha\beta}\setminus U}
	\Bigl(\smashoperator[r]{\sum_{q\in V\setminus V_{\alpha\beta}}}\left(H_{pq}(\beta, x_q) + H_{qp}(x_q, \beta)\right) + H_p(\beta)\Bigr)\\
	&+\sum_{p\in U}\sum_{q\in V_{\alpha\beta}\setminus U}\left(H_{pq}(\alpha, \beta)  + H_{qp}(\beta, \alpha)\right) 
	+ \sum_{p\in V\setminus V_{\alpha\beta}}\Bigl(\sum_{\substack{V\setminus V_{\alpha\beta} \\ q \neq p}}H_{pq}(x_p, x_q) + H_p(x_p) \Bigr)
\end{align}
Let 
\begin{align}
C \coloneqq \sum_{pq\in P_{V\setminus V_{\alpha\beta}}}H_{pq}(x_p, x_q) +  \sum_{p\in V\setminus V_{\alpha\beta}}H_p(x_p)
\end{align}
and $y\in \{0, 1\}^{V_{\alpha\beta}}$ such that $y_p = 1$ if and only if $p\in U$. 
Moreover, let $G \coloneqq (V', E)$ with $V' \coloneqq V_{\alpha\beta}\cup \{\alpha, \beta\}$ and  
\begin{equation}
	E\coloneqq \{\alpha p\mid p\in V_{\alpha\beta}\} \cup \{p\beta\mid p\in V_{\alpha\beta}\} \cup P_{V_{\alpha\beta}}\enspace,
\end{equation}
and the capacities $w\in \mathbb{R}^E$ such that
\begin{align}
	&\forall p\in V_{\alpha\beta}\colon w_{\alpha p} = \sum_{q\in V\setminus V_{\alpha\beta}}\Bigl(H_{pq}(\beta, x_q) + H_{qp}(x_q, \beta)\Bigr) + H_p(\beta) \enspace,\\
	&\forall p\in V_{\alpha\beta}\colon w_{p \beta} = \sum_{q\in V\setminus V_{\alpha\beta}}\Bigl(H_{pq}(\alpha, x_q) + H_{qp}(x_q, \alpha)\Bigr) + H_p(\alpha)\enspace,\\
	&\forall pq\in P_{V_{\alpha\beta}} \colon w_{pq} = H_{pq}(\alpha, \beta) + H_{qp}(\beta, \alpha)\enspace.
\end{align}
Then:
\begin{align}
	\label{eq:alpha-beta-cut-problem}
	&\psi(\operatorname{\alpha\beta-swap}_U(x)) = C + \sum_{pq\in P_{V_{\alpha\beta}}}w_{pq}y_p (1-y_q) 
	+ \sum_{p\in V_{\alpha\beta}}w_{\alpha p}(1-y_p) + \sum_{p \in V_{\alpha\beta}}w_{p\beta}y_p\enspace.
\end{align}

Thus, for a fixed $x\in \mathcal{L}^V$ the problem of finding the optimal $\alpha\beta$-swap reduces to a min-$\alpha\beta$-cut problem wrt.~$G$ and $w$. 
With our values \eqref{eq:energy-minimization-costs-1},~\eqref{eq:energy-minimization-costs-2},~\eqref{eq:energy-minimization-costs-3} and~\eqref{eq:energy-minimization-costs-4}, we have $w_{pq} \geq 0$ for all $pq\in P_{V_{\alpha\beta}}$, and $w_{\alpha p}, w_{p\beta} \geq 0$ for all $p\in V_{\alpha\beta}$. 
Thus, minimizing~\eqref{eq:alpha-beta-cut-problem} wrt.~$y\in \{0, 1\}^{V_{\alpha\beta}}$ can be done efficiently by means of a max-flow algorithm. 

%% file: appendix-google+.tex
\section{Experiments on Google+ Dataset}\label{appendix:google+}
\counterwithin{figure}{section}

For Google+,~\Cref{figure:gplus-250} shows the percentage of variables fixed and the runtime, both as a function of $\lvert V\rvert$, for applications of cut conditions~\Cref{proposition:edge-cut-condition},~\Cref{corollary:directed-cut-condition} and \Cref{corollary:boecker-condition-strong} ($b = 0$), and join conditions \Cref{proposition:edge-join-condition},~\Cref{proposition:subset-U-maximizer} (together with~\Cref{proposition:tractable-U}) and~\Cref{corollary:boecker-condition-strong} ($b = 1$), separately and jointly. Also here, all join conditions are applied on top of partial optimality from the cut conditions. We report these for all instances with $\vert V\rvert \leq 250$.

We make the following observations: \Cref{proposition:edge-cut-condition} consistently fixes more variables to zero than \Cref{corollary:directed-cut-condition}, but takes two orders more time to apply.
Both \Cref{proposition:edge-cut-condition,corollary:directed-cut-condition} are generally more effective in fixing variables to zero than~\Cref{corollary:boecker-condition-strong} ($b = 0$).
For these instances, fewer variables are fixed to one than for instances of the synthetic dataset, but each join condition fixes some variables. 

\input{figure-google+-250}

%% file: figure-google+-250.tex
\tikzmath{\plotheight=3.2;}
\tikzmath{\plotwidth=3.2;}

\setlength{\hspacing}{0.10cm}
\setlength{\vspacing}{0.50cm}

\providecommand{\addzeropersistencyplot}{}	
\renewcommand{\addzeropersistencyplot}[1]{\addplot+[
	only marks,
	mark=*,
	mark options={scale=0.25},
	color=blue
	] table[
	col sep=comma,
	x expr=\thisrow{initialNumberOfVertices}, 
	y expr=\thisrow{persistentZeroEdges} / \thisrow{initialNumberOfEdges}*100
	] {#1};}

\providecommand{\addpersistencyplot}{}	
\renewcommand{\addpersistencyplot}[1]{\addplot+[
	only marks,
	mark=*,
	mark options={scale=0.25},
	color=blue
	] table[
	col sep=comma,
	x expr=\thisrow{initialNumberOfVertices}, 
	y expr=(\thisrow{persistentZeroEdges} + \thisrow{persistentOneEdges}) / \thisrow{initialNumberOfEdges}*100
	] {#1};}

\providecommand{\addonepersistencyplot}{}	
\renewcommand{\addonepersistencyplot}[1]{\addplot+[
	only marks,
	mark=*,
	mark options={scale=0.25},
	color=blue
	] table[
	col sep=comma,
	x expr=\thisrow{initialNumberOfVertices}, 
	y expr=\thisrow{persistentOneEdges} / \thisrow{initialNumberOfEdges}*100
	] {#1};}

\providecommand{\addruntimeplot}{}	
\renewcommand{\addruntimeplot}[2][]{\addplot+[
	only marks,
	mark=*,
	mark options={scale=0.25, color=black},
	#1
	] table[
	col sep=comma,
	x expr=\thisrow{initialNumberOfVertices}, 
	y expr=\thisrow{runtimeNanoseconds} / 1e9
	] {#2};}

\begin{figure}[h]
	\centering 
\tikzsetnextfilename{fig-gplus-250}
\begin{tikzpicture}		
	\begin{groupplot}[group style={group size= 14 by 2, horizontal sep=\hspacing, vertical sep=\vspacing}]

		\nextgroupplot[
		title={\Cref{proposition:edge-cut-condition}},
		ylabel={$\frac{\lvert \dom{\hat{x}}\rvert}{\lvert P_V\rvert} \left[\%\right]$},
		title style={align=center, font=\small, yshift=-1ex},
		xmin=0,
		xmax=250,
		xtick distance=100,
		ymin=0,
		ymax=100,
		width=\plotwidth cm,
		height=\plotheight cm,
		ylabel shift=-1ex
		]
		\addzeropersistencyplot{./data/gplus-250/edgeCut/persistencies.csv};
		
		\nextgroupplot[
		width=0.525*\plotwidth cm,
		height=\plotheight cm,
		enlarge x limits=0.5,
		ymin=0,
		ymax=100,
		xtick=data,
		hide axis,
		xshift=-\hspacing
		]
		
		\addplot[
		boxplot prepared from table={
			table=./data/gplus-250/edgeCut/boxplot-0.csv,
			lower whisker=lw,
			upper whisker=uw,
			lower quartile=lq,
			upper quartile=uq,
			median=med
		},
		boxplot prepared,
		boxplot/draw direction=y
		] coordinates {};
		
		\nextgroupplot[
		title={\Cref{corollary:directed-cut-condition}},
		title style={align=center, font=\small, yshift=-1ex},
		xmin=0,
		xmax=250,
		xtick distance=100,
		ymin=0,
		ymax=100,
		width=\plotwidth cm,
		height=\plotheight cm,
		yticklabels=\empty
		]
		\addzeropersistencyplot{./data/gplus-250/directedCut/persistencies.csv};
		
		\nextgroupplot[
		width=0.525*\plotwidth cm,
		height=\plotheight cm,
		enlarge x limits=0.5,
		ymin=0,
		ymax=100,
		xtick=data,
		hide axis,
		xshift=-\hspacing
		]
		\addplot[
		boxplot prepared from table={
			table=./data/gplus-250/directedCut/boxplot-0.csv,
			lower whisker=lw,
			upper whisker=uw,
			lower quartile=lq,
			upper quartile=uq,
			median=med
		},
		boxplot prepared,
		boxplot/draw direction=y
		] coordinates {};
		
		\nextgroupplot[
		title={Cuts+\Cref{proposition:edge-join-condition}},
		title style={align=center, font=\small, yshift=-1ex},
		xmin=0,
		xmax=250,
		xtick distance=100,
		ymin=0,
		ymax=100,
		width=\plotwidth cm,
		height=\plotheight cm,
		yticklabels=\empty
		]
		\addonepersistencyplot{./data/gplus-250/cuts+edgeJoin/persistencies.csv};
		
		\nextgroupplot[
		width=0.525*\plotwidth cm,
		height=\plotheight cm,
		enlarge x limits=0.5,
		ymin=0,
		ymax=100,
		xtick=data,
		hide axis,
		xshift=-\hspacing
		]
		\addplot[
		boxplot prepared from table={
			table=./data/gplus-250/cuts+edgeJoin/boxplot-1.csv,
			lower whisker=lw,
			upper whisker=uw,
			lower quartile=lq,
			upper quartile=uq,
			median=med
		},
		boxplot prepared,
		boxplot/draw direction=y
		] coordinates {};

		\nextgroupplot[
		title={Cuts+\Cref{proposition:subset-U-maximizer} \\ (Prop.~\ref{proposition:tractable-U})},
		title style={align=center, font=\small, yshift=-1ex},
		xmin=0,
		xmax=250,
		xtick distance=100,
		ymin=0,
		ymax=100,
		width=\plotwidth cm,
		height=\plotheight cm,
		yticklabels=\empty
		]
		\addonepersistencyplot{./data/gplus-250/cuts+preorderJoinExact/persistencies.csv};
		
		\nextgroupplot[
		width=0.525*\plotwidth cm,
		height=\plotheight cm,
		enlarge x limits=0.5,
		ymin=0,
		ymax=100,
		xtick=data,
		hide axis,
		xshift=-\hspacing
		]
		\addplot[
		boxplot prepared from table={
			table=./data/gplus-250/cuts+preorderJoinExact/boxplot-1.csv,
			lower whisker=lw,
			upper whisker=uw,
			lower quartile=lq,
			upper quartile=uq,
			median=med
		},
		boxplot prepared,
		boxplot/draw direction=y
		] coordinates {};

		\nextgroupplot[
		title={\Cref{corollary:boecker-condition-strong} \\ ($b = 0$)},
		title style={align=center, font=\small, yshift=-1ex},
		xmin=0,
		xmax=250,
		xtick distance=100,
		ymin=0,
		ymax=100,
		width=\plotwidth cm,
		height=\plotheight cm,
		yticklabels=\empty,
		]
		\addzeropersistencyplot{./data/gplus-250/bbkStrongCutConditions/persistencies.csv};
		
		\nextgroupplot[
		width=0.525*\plotwidth cm,
		height=\plotheight cm,
		enlarge x limits=0.5,
		ymin=0,
		ymax=100,
		xtick=data,
		hide axis,
		xshift=-\hspacing
		]
		\addplot[
		boxplot prepared from table={
			table=./data/gplus-250/bbkStrongCutConditions/boxplot-0.csv,
			lower whisker=lw,
			upper whisker=uw,
			lower quartile=lq,
			upper quartile=uq,
			median=med
		},
		boxplot prepared,
		boxplot/draw direction=y
		] coordinates {};

		\nextgroupplot[
		title={Cuts+\Cref{corollary:boecker-condition-strong} \\ ($b = 1$) },
		title style={align=center, font=\small, yshift=-1ex},
		xmin=0,
		xmax=250,
		xtick distance=100,
		ymin=0,
		ymax=100,
		width=\plotwidth cm,
		height=\plotheight cm,
		yticklabels=\empty
		]
		\addonepersistencyplot{./data/gplus-250/cuts+bbkStrongJoinConditions/persistencies.csv};
		
		\nextgroupplot[
		width=0.525*\plotwidth cm,
		height=\plotheight cm,
		enlarge x limits=0.5,
		ymin=0,
		ymax=100,
		xtick=data,
		hide axis,
		xshift=-\hspacing
		]
		\addplot[
		boxplot prepared from table={
			table=./data/gplus-250/cuts+bbkStrongJoinConditions/boxplot-1.csv,
			lower whisker=lw,
			upper whisker=uw,
			lower quartile=lq,
			upper quartile=uq,
			median=med
		},
		boxplot prepared,
		boxplot/draw direction=y
		] coordinates {};

		\nextgroupplot[
		title={Joint},
		title style={align=center, font=\small, yshift=-1ex},
		xmin=0,
		xmax=250,
		xtick distance=100,
		ymin=0,
		ymax=100,
		width=\plotwidth cm,
		height=\plotheight cm,
		yticklabels=\empty
		]
		\addpersistencyplot{./data/gplus-250/allConditions/persistencies.csv};

		\nextgroupplot[
		width=0.525*\plotwidth cm,
		height=\plotheight cm,
		enlarge x limits=0.5,
		ymin=0,
		ymax=100,
		xtick=data,
		hide axis,
		xshift=-\hspacing
		]
		\addplot[
		boxplot prepared from table={
			table=./data/gplus-250/allConditions/boxplot-all.csv,
			lower whisker=lw,
			upper whisker=uw,
			lower quartile=lq,
			upper quartile=uq,
			median=med
		},
		boxplot prepared,
		boxplot/draw direction=y
		] coordinates {};

		\nextgroupplot[
		ylabel={Runtime [s]},
		ymode=log,
		xmode=log,
		xlabel={$\vert V\vert$},
		xmin=0,
		xmax=250,
		ymin=1e-5,
		ymax=1e5,
		width=\plotwidth cm,
		height=\plotheight cm,
		legend pos=south east,
		legend style={
			font=\scriptsize,
			draw=none,  
			fill opacity=0.9
		},
		legend image post style={xscale=0.25},
		ylabel shift=-1ex
		]
		\addruntimeplot[forget plot]{./data/gplus-250/edgeCut/persistencies.csv}
		\addplot[
		domain=1:250,      
		samples=2,     
		color=black,
		thick,
		]
		{1.2884e-09 * x^4.5988};
		\addlegendentry{$O(\lvert V\rvert^{4.6})$}

		\nextgroupplot[
		width=0.5*\plotwidth cm,
		height=\plotheight cm,
		hide axis
		]

		\nextgroupplot[
		ymode=log,
		xmode=log,
		xlabel={$\vert V\vert$},
		xmin=0,
		xmax=250,
		ymin=1e-5,
		ymax=1e5,
		width=\plotwidth cm,
		height=\plotheight cm,
		yticklabels=\empty,
		legend pos=south east,
		legend style={
			font=\scriptsize,
			draw=none,  
			fill opacity=0.9
		},
		legend image post style={xscale=0.25}
		]
		\addruntimeplot[forget plot]{./data/gplus-250/directedCut/persistencies.csv}
		\addplot[
		domain=1:250,      
		samples=2,     
		color=black,
		thick,
		]
		{1.5735e-07 * x^2.1658};
		\addlegendentry{$O(\lvert V\rvert^{2.2})$}

		\nextgroupplot[
		width=0.5*\plotwidth cm,
		height=\plotheight cm,
		hide axis
		]

		\nextgroupplot[
		ymode=log,
		xmode=log,
		xlabel={$\vert V\vert$},
		xmin=0,
		xmax=250,
		ymin=1e-5,
		ymax=1e5,
		width=\plotwidth cm,
		height=\plotheight cm,
		yticklabel=\empty,
		legend pos=south east,
		legend style={
			font=\scriptsize,
			draw=none,  
			fill opacity=0.9
		},
		legend image post style={xscale=0.25}
		]
		\addruntimeplot[forget plot]{./data/gplus-250/cuts+edgeJoin/persistencies.csv}
		\addplot[
		domain=1:250,      
		samples=2,     
		color=black,
		thick,
		]
		{9.1433e-09 * x^4.7814};
		\addlegendentry{$O(\lvert V\rvert^{4.8})$}

		\nextgroupplot[
		width=0.5*\plotwidth cm,
		height=\plotheight cm,
		hide axis
		]
		
		\nextgroupplot[
		ymode=log,
		xmode=log,
		xlabel={$\vert V\vert$},
		xmin=0,
		xmax=250,
		ymin=1e-5,
		ymax=1e5,
		width=\plotwidth cm,
		height=\plotheight cm,
		yticklabel=\empty,
		legend pos=south east,
		legend style={
			font=\scriptsize,
			draw=none,  
			fill opacity=0.9
		},
		legend image post style={xscale=0.25}
		]
		\addruntimeplot[forget plot]{./data/gplus-250/cuts+preorderJoinExact/persistencies.csv}
		\addplot[
		domain=1:250,      
		samples=2,     
		color=black,
		thick,
		]
		{9.0347e-08 * x^4.2460};
		\addlegendentry{$O(\lvert V\rvert^{4.2})$}
		
		\nextgroupplot[
		width=0.5*\plotwidth cm,
		height=\plotheight cm,
		hide axis
		]

		\nextgroupplot[
		ymode=log,
		xmode=log,
		xlabel={$\vert V\vert$},
		xmin=0,
		xmax=250,
		ymin=1e-5,
		ymax=1e5,
		width=\plotwidth cm,
		height=\plotheight cm,
		yticklabels=\empty,
		legend pos=south east,
		legend style={
			font=\scriptsize,
			draw=none,  
			fill opacity=0.9
		},
		legend image post style={xscale=0.25}
		]
		\addruntimeplot[forget plot]{./data/gplus-250/bbkStrongCutConditions/persistencies.csv}
		\addplot[
		domain=1:250,      
		samples=2,     
		color=black,
		thick,
		]
		{7.3766e-06 * x^2.5930};
		\addlegendentry{$O(\lvert V\rvert^{2.6})$}

		\nextgroupplot[
		width=0.5*\plotwidth cm,
		height=\plotheight cm,
		hide axis
		]

		\nextgroupplot[
		ymode=log,
		xmode=log,
		xlabel={$\vert V\vert$},
		xmin=0,
		xmax=250,
		ymin=1e-5,
		ymax=1e5,
		width=\plotwidth cm,
		height=\plotheight cm,
		yticklabel=\empty,
		legend pos=south east,
		legend style={
			font=\scriptsize,
			draw=none,  
			fill opacity=0.9
		},
		legend image post style={xscale=0.25}
		]
		\addruntimeplot[forget plot]{./data/gplus-250/cuts+bbkStrongJoinConditions/persistencies.csv}
		\addplot[
		domain=1:250,      
		samples=2,     
		color=black,
		thick,
		]
		{2.4451e-07 * x^3.5667};
		\addlegendentry{$O(\lvert V\rvert^{3.6})$}

		\nextgroupplot[
		width=0.5*\plotwidth cm,
		height=\plotheight cm,
		hide axis
		]
		
		\nextgroupplot[
		ymode=log,
		xmode=log,
		xlabel={$\vert V\vert$},
		xmin=0,
		xmax=250,
		ymin=1e-5,
		ymax=1e5,
		width=\plotwidth cm,
		height=\plotheight cm,
		yticklabel=\empty,
		legend pos=south east,
		legend style={
			font=\scriptsize,
			draw=none,  
			fill opacity=0.9
		},
		legend image post style={xscale=0.25}
		]
		\addruntimeplot[forget plot]{./data/gplus-250/allConditions/persistencies.csv}
		\addplot[
		domain=1:250, 
		samples=2, 
		color=black,
		thick,
		]
		{1.7194e-08 * x^5.1053};
		\addlegendentry{$O(\lvert V\rvert^{5.1})$}

		\nextgroupplot[
		width=0.5*\plotwidth cm,
		height=\plotheight cm,
		hide axis
		]	
		
	\end{groupplot}
\end{tikzpicture}
\caption{
	Shown above are the percentage of fixed variables (Row~1) as well as the runtimes (Row~2) 
	for applying \Cref{proposition:edge-cut-condition},~\Cref{corollary:directed-cut-condition},~\Cref{proposition:edge-join-condition}, ~\Cref{proposition:subset-U-maximizer} (together with~\Cref{proposition:tractable-U}) and~\Cref{corollary:boecker-condition-strong} for $b\in \{0, 1\}$ individually as well as all conditions jointly (Column~7)
	to instances of the Google+ dataset for $\lvert V\rvert \leq 250$.
}
\label{figure:gplus-250}
\end{figure}

%% file: appendix-boecker-condition.tex
\section{Experiments for Applying~\Cref{corollary:boecker-condition-strong}}
\label{appendix:boecker}

In~\Cref{figure:joint-vs-boecker-synthetic-alphas,figure:joint-vs-boecker-synthetic-ns,figure:joint-vs-boecker-twitter} we report the joint application of condition~\Cref{corollary:boecker-condition-strong} for $b \in \{0, 1\}$ due to~\cite{boecker2009}. On the synthetic dataset the results are comparable to those reported in~\Cref{figure:synthetic-alphas,figure:synthetic-ns}. However, there is a slight improvement of the fraction of variables fixed, esp. for $\alpha \in \{0.65, 0.70, 0.75\}$ and as a function of $\lvert V\rvert$ when applying our additional conditions. On the Twitter dataset the application of our additional conditions fix significantly more variables.

\begin{figure}
	\begin{minipage}[t]{0.475\textwidth}
		\centering
		\input{figure-synthetic-alphas-boecker}
		\caption{
			Shown above are the percentage of fixed variables (Row~1) and runtimes (Row~2) 
			for applying~\Cref{corollary:boecker-condition-strong} ($b \in \{0, 1\}$)
			to instances of the synthetic dataset with respect to $\lvert V\rvert = 40$ and $p_E \in \{0.25, 0.50, 0.75\}$.}
		\label{figure:joint-vs-boecker-synthetic-alphas}
	\end{minipage}%
	\hfill
	\begin{minipage}[t]{0.475\textwidth}
		\centering
		\input{figure-synthetic-ns-boecker}
		\caption{
			Shown above are the percentage of fixed variables (Row~1) and runtimes (Row~2) 
			for applying \Cref{corollary:boecker-condition-strong}
			to instances of the synthetic dataset with $\alpha\in \{0.25, 0.65, 0.70, 0.75\}$ and $p_E \in \{0.25, 0.50, 0.75\}$.}
		\label{figure:joint-vs-boecker-synthetic-ns}
	\end{minipage}%
\end{figure}
\begin{figure}
	\centering
	\input{figure-twitter-joint-vs-boecker}
	\caption{Above, we report the percentage of fixed variables (Row~1) as well as the runtime (Row~2) for applying condition~\Cref{corollary:boecker-condition-strong} ($b\in \{0, 1\}$) and all conditions jointly to instances of the Twitter dataset.}
	\label{figure:joint-vs-boecker-twitter}
\end{figure}

%% file: figure-synthetic-alphas-boecker.tex
\tikzmath{\plotheight=3.2;}
\tikzmath{\plotwidth=3.2;}

\setlength{\hspacing}{0.25cm}
\setlength{\vspacing}{0.50cm}

\providecommand{\addvariableplot}{}	
\renewcommand{\addvariableplot}[2]{\addplot+[
	only marks,
	mark=*,
	mark options={scale=0.25},
	color=#2
	] table[
	col sep=comma,
	x expr=\thisrow{alpha}, 
	y expr=\thisrow{medianEliminatedVariables}*100,
	y error minus expr=(\thisrow{medianEliminatedVariables} - \thisrow{q25EliminatedVariables})*100,
	y error plus expr=(\thisrow{q75EliminatedVariables} - \thisrow{medianEliminatedVariables})*100
	] {#1};}

\providecommand{\addruntimeplot}{}	
\renewcommand{\addruntimeplot}[2]{\addplot+[
	only marks,
	mark=*,
	mark options={scale=0.25},
	color=#2
	] table[
	col sep=comma,
	x expr=\thisrow{alpha}, 
	y expr=\thisrow{medianDuration} / 1e9,
	y error minus expr=(\thisrow{medianDuration} - \thisrow{q25Duration}) / 1e9,
	y error plus expr=(\thisrow{q75Duration} - \thisrow{medianDuration}) / 1e9
	] {#1};}
	
	\tikzsetnextfilename{figure-boecker-synthetic-alphas}
	\begin{tikzpicture}
		\begin{groupplot}[group style={group size=3 by 2, horizontal sep=\hspacing, vertical sep=\vspacing}]

			\nextgroupplot[
			title={$p_E = 0.25$},
			title style={align=center, font=\small, yshift=-1ex},
			xmin=0,
			xmax=1,
			xtick distance=0.5,
			ymin=0,
			ymax=100,
			width=\plotwidth cm,
			height=\plotheight cm,
			ylabel={$\frac{\lvert \dom{\hat{x}}\rvert}{\lvert P_V\rvert}\left[\%\right]$},
			legend to name=TCLAlphasLegend,
			legend style={legend columns=4, font=\scriptsize},
			ylabel style={yshift=-6pt},
			legend pos=south west,
			xticklabels=\empty
			]
			\addvariableplot{./data/RandomTransitivelyClosed/bbkConditions/alphas/0.25/stats_20.csv}{blue}
			\addlegendentry{$\lvert V\rvert = 20$}
			\addvariableplot{./data/RandomTransitivelyClosed/bbkConditions/alphas/0.25/stats_30.csv}{\darkgreen}
			\addlegendentry{$\lvert V\rvert = 30$}
			\addvariableplot{./data/RandomTransitivelyClosed/bbkConditions/alphas/0.25/stats_40.csv}{red}
			\addlegendentry{$\lvert V\rvert = 40$}
			\addvariableplot{./data/RandomTransitivelyClosed/bbkConditions/alphas/0.25/stats_50.csv}{magenta}
			\addlegendentry{$\lvert V\rvert = 50$}

			\nextgroupplot[
			title style={align=center, font=\small, yshift=-1ex},
			xmin=0,
			xmax=1,
			ymin=0,
			ymax=100,
			width=\plotwidth cm,
			height=\plotheight cm,
			title={$p_E = 0.50$},
			yticklabels=\empty,
			xtick distance=0.5,
			xticklabels=\empty
			]
			\addvariableplot{./data/RandomTransitivelyClosed/bbkConditions/alphas/0.50/stats_20.csv}{blue}
			\addvariableplot{./data/RandomTransitivelyClosed/bbkConditions/alphas/0.50/stats_30.csv}{\darkgreen}
			\addvariableplot{./data/RandomTransitivelyClosed/bbkConditions/alphas/0.50/stats_40.csv}{red}
			\addvariableplot{./data/RandomTransitivelyClosed/bbkConditions/alphas/0.50/stats_50.csv}{magenta}
			
			\nextgroupplot[
			title={$p_E = 0.75$},
			title style={align=center, font=\small, yshift=-1ex},
			xmin=0,
			xmax=1,
			ymin=0,
			ymax=100,
			width=\plotwidth cm,
			height=\plotheight cm,
			yticklabels=\empty,
			xtick distance=0.5,
			xticklabels=\empty
			]
			\addvariableplot{./data/RandomTransitivelyClosed/bbkConditions/alphas/0.75/stats_20.csv}{blue}
			\addvariableplot{./data/RandomTransitivelyClosed/bbkConditions/alphas/0.75/stats_30.csv}{\darkgreen}
			\addvariableplot{./data/RandomTransitivelyClosed/bbkConditions/alphas/0.75/stats_40.csv}{red}
			\addvariableplot{./data/RandomTransitivelyClosed/bbkConditions/alphas/0.75/stats_50.csv}{magenta}

			\nextgroupplot[
			ylabel={Runtimes [s]},
			title style={align=center},
			xmin=0,
			xmax=1,
			ymin=1e-3,
			ymax=1e2,
			ymode=log,
			width=\plotwidth cm,
			height=\plotheight cm,
			ylabel style={yshift=-6pt},
			xlabel=$\alpha$,
			xtick distance=0.5,
			xlabel style={yshift=4pt}
			]
			
			\addruntimeplot{./data/RandomTransitivelyClosed/bbkConditions/alphas/0.25/stats_20.csv}{blue}
			\addruntimeplot{./data/RandomTransitivelyClosed/bbkConditions/alphas/0.25/stats_30.csv}{\darkgreen}
			\addruntimeplot{./data/RandomTransitivelyClosed/bbkConditions/alphas/0.25/stats_40.csv}{red}
			\addruntimeplot{./data/RandomTransitivelyClosed/bbkConditions/alphas/0.25/stats_50.csv}{magenta}
			
			\nextgroupplot[
			title style={align=center},
			xmin=0,
			xmax=1,
			ymode=log,
			ymin=1e-3,
			ymax=1e2,
			width=\plotwidth cm,
			height=\plotheight cm,
			yticklabels=\empty,
			xlabel=$\alpha$,
			xtick distance=0.5,
			xlabel style={yshift=4pt}
			]
			
			\addruntimeplot{./data/RandomTransitivelyClosed/bbkConditions/alphas/0.50/stats_20.csv}{blue}
			\addruntimeplot{./data/RandomTransitivelyClosed/bbkConditions/alphas/0.50/stats_30.csv}{\darkgreen}
			\addruntimeplot{./data/RandomTransitivelyClosed/bbkConditions/alphas/0.50/stats_40.csv}{red}
			\addruntimeplot{./data/RandomTransitivelyClosed/bbkConditions/alphas/0.50/stats_50.csv}{magenta}

			\nextgroupplot[
			ymode=log,
			title style={align=center},
			xmin=0,
			xmax=1,
			ymin=1e-3,
			ymax=1e2,
			width=\plotwidth cm,
			height=\plotheight cm,
			yticklabels=\empty,
			xlabel=$\alpha$,
			xtick distance=0.5,
			xlabel style={yshift=4pt}
			]
			
			\addruntimeplot{./data/RandomTransitivelyClosed/bbkConditions/alphas/0.75/stats_20.csv}{blue}
			\addruntimeplot{./data/RandomTransitivelyClosed/bbkConditions/alphas/0.75/stats_30.csv}{\darkgreen}
			\addruntimeplot{./data/RandomTransitivelyClosed/bbkConditions/alphas/0.75/stats_40.csv}{red}
			\addruntimeplot{./data/RandomTransitivelyClosed/bbkConditions/alphas/0.75/stats_50.csv}{magenta}

		\end{groupplot}
		
		\path (group c1r2.south west) -- node[below=0.75cm]{\ref{TCLAlphasLegend}} (group c3r2.south east);
	\end{tikzpicture}

%% file: figure-synthetic-ns-boecker.tex
\tikzmath{\plotheight=3.2;}
\tikzmath{\plotwidth=3.2;}

\setlength{\hspacing}{0.20cm}
\setlength{\vspacing}{0.50cm}

\providecommand{\addvariableplot}{}	
\renewcommand{\addvariableplot}[2]{\addplot+[
	only marks,
	mark=*,
	mark options={scale=0.25},
	color=#2
	] table[
	col sep=comma,
	x expr=\thisrow{numberOfPoints}, 
	y expr=\thisrow{medianEliminatedVariables}*100,
	y error minus expr=(\thisrow{medianEliminatedVariables} - \thisrow{q25EliminatedVariables})*100,
	y error plus expr=(\thisrow{q75EliminatedVariables} - \thisrow{medianEliminatedVariables})*100
	] {#1};}

\providecommand{\addruntimeplot}{}	
\renewcommand{\addruntimeplot}[3][]{\addplot+[
	only marks,
	mark=*,
	mark options={scale=0.25},
	color=#3,
	#1
	] table[
	col sep=comma,
	x expr=\thisrow{numberOfPoints}, 
	y expr=\thisrow{medianDuration} / 1e9,
	y error minus expr=(\thisrow{medianDuration} - \thisrow{q25Duration}) / 1e9,
	y error plus expr=(\thisrow{q75Duration} - \thisrow{medianDuration}) / 1e9
	] {#2};}
	
	\tikzsetnextfilename{figure-boecker-synthetic-ns}
	\begin{tikzpicture}		
		\begin{groupplot}[group style={group size=3 by 2, horizontal sep=\hspacing, vertical sep=\vspacing}]

			\nextgroupplot[
			title={$p_E = 0.25$},
			title style={align=center, font=\small, yshift=-1ex},
			xmin=0,
			xmax=110,
			xtick distance=50,
			ymin=0,
			ymax=100,
			width=\plotwidth cm,
			height=\plotheight cm,
			ylabel={$\frac{\lvert \dom{\hat{x}}\rvert}{\lvert P_V\rvert}\left[\%\right]$},
			legend to name=TCLNumberOfVerticesLegend,
			legend style={legend columns=4, font=\scriptsize},
			ylabel style={yshift=-6pt}
			]
			\addvariableplot{./data/RandomTransitivelyClosed/bbkConditions/ns/0.25/stats_0.25.csv}{blue}
			\addlegendentry{$\alpha = 0.25$}
			\addvariableplot{./data/RandomTransitivelyClosed/bbkConditions/ns/0.25/stats_0.65.csv}{\darkgreen}
			\addlegendentry{$\alpha = 0.65$}
			\addvariableplot{./data/RandomTransitivelyClosed/bbkConditions/ns/0.25/stats_0.70.csv}{red}
			\addlegendentry{$\alpha = 0.70$}
			\addvariableplot{./data/RandomTransitivelyClosed/bbkConditions/ns/0.25/stats_0.75.csv}{magenta}
			\addlegendentry{$\alpha = 0.75$}

			\nextgroupplot[
			title style={align=center, font=\small, yshift=-1ex},
			xmin=0,
			xmax=110,
			ymin=0,
			ymax=100,
			width=\plotwidth cm,
			height=\plotheight cm,
			title={$p_E = 0.50$},
			yticklabels=\empty,
			xtick distance=50
			]
			\addvariableplot{./data/RandomTransitivelyClosed/bbkConditions/ns/0.50/stats_0.25.csv}{blue}
			\addvariableplot{./data/RandomTransitivelyClosed/bbkConditions/ns/0.50/stats_0.65.csv}{\darkgreen}
			\addvariableplot{./data/RandomTransitivelyClosed/bbkConditions/ns/0.50/stats_0.70.csv}{red}
			\addvariableplot{./data/RandomTransitivelyClosed/bbkConditions/ns/0.50/stats_0.75.csv}{magenta}
			
			\nextgroupplot[
			title={$p_E = 0.75$},
			title style={align=center, font=\small, yshift=-1ex},
			xmin=0,
			xmax=110,
			ymin=0,
			ymax=100,
			width=\plotwidth cm,
			height=\plotheight cm,
			yticklabels=\empty,
			xtick distance=50
			]
			\addvariableplot{./data/RandomTransitivelyClosed/bbkConditions/ns/0.75/stats_0.25.csv}{blue}
			\addvariableplot{./data/RandomTransitivelyClosed/bbkConditions/ns/0.75/stats_0.65.csv}{\darkgreen}
			\addvariableplot{./data/RandomTransitivelyClosed/bbkConditions/ns/0.75/stats_0.70.csv}{red}
			\addvariableplot{./data/RandomTransitivelyClosed/bbkConditions/ns/0.75/stats_0.75.csv}{magenta}

			\nextgroupplot[
			xmode=log,
			ymode=log,
			ylabel={Runtimes [s]},
			title style={align=center},
			xmin=0,
			xmax=110,
			ymin=1e-5,
			ymax=1000,
			width=\plotwidth cm,
			height=\plotheight cm,
			ylabel style={yshift=-6pt},
			xlabel=$\vert V\vert$,
			xlabel style={yshift=4pt},
			legend pos=south east,
			legend style={
				font=\scriptsize,
				draw=none,  
				fill opacity=0.9
			},
			legend image post style={xscale=0.25}
			]
			
			\addruntimeplot[forget plot]{./data/RandomTransitivelyClosed/bbkConditions/ns/0.25/stats_0.25.csv}{blue}
			\addruntimeplot[forget plot]{./data/RandomTransitivelyClosed/bbkConditions/ns/0.25/stats_0.65.csv}{\darkgreen}
			\addruntimeplot[forget plot]{./data/RandomTransitivelyClosed/bbkConditions/ns/0.25/stats_0.70.csv}{red}
			\addruntimeplot[forget plot]{./data/RandomTransitivelyClosed/bbkConditions/ns/0.25/stats_0.75.csv}{magenta}
			\addplot[
			domain=1:100,      
			samples=2,     
			color=magenta,
			thick,
			]
			{3.2345e-08 * x^4.9876};
			\addlegendentry{$O(\lvert V\rvert^{5.0})$}

			\nextgroupplot[
			xmode=log,
			ymode=log,
			title style={align=center},
			xmin=0,
			xmax=110,
			ymin=1e-5,
			ymax=1000,
			width=\plotwidth cm,
			height=\plotheight cm,
			yticklabels=\empty,
			xlabel=$\vert V\vert$,
			xlabel style={yshift=4pt},
			legend pos=south east,
			legend style={
				font=\scriptsize,
				draw=none,  
				fill opacity=0.9
			},
			legend image post style={xscale=0.25}
			]
			
			\addruntimeplot[forget plot]{./data/RandomTransitivelyClosed/bbkConditions/ns/0.50/stats_0.25.csv}{blue}
			\addruntimeplot[forget plot]{./data/RandomTransitivelyClosed/bbkConditions/ns/0.50/stats_0.65.csv}{\darkgreen}
			\addruntimeplot[forget plot]{./data/RandomTransitivelyClosed/bbkConditions/ns/0.50/stats_0.70.csv}{red}
			\addruntimeplot[forget plot]{./data/RandomTransitivelyClosed/bbkConditions/ns/0.50/stats_0.75.csv}{magenta}
			\addplot[
			domain=1:100,      
			samples=2,     
			color=magenta,
			thick,
			]
			{2.8290e-08 * x^5.0002};
			\addlegendentry{$O(\lvert V\rvert^{5.0})$}

			\nextgroupplot[
			xmode=log,
			ymode=log,
			title style={align=center},
			xmin=0,
			xmax=110,
			ymin=1e-5,
			ymax=1000,
			width=\plotwidth cm,
			height=\plotheight cm,
			yticklabels=\empty,
			xlabel=$\vert V\vert$,
			xlabel style={yshift=4pt},
			legend pos=south east,
			legend style={
				font=\scriptsize,
				draw=none,  
				fill opacity=0.9
			},
			legend image post style={xscale=0.25}
			]
			
			\addruntimeplot[forget plot]{./data/RandomTransitivelyClosed/bbkConditions/ns/0.75/stats_0.25.csv}{blue}
			\addruntimeplot[forget plot]{./data/RandomTransitivelyClosed/bbkConditions/ns/0.75/stats_0.65.csv}{\darkgreen}
			\addruntimeplot[forget plot]{./data/RandomTransitivelyClosed/bbkConditions/ns/0.75/stats_0.70.csv}{red}
			\addruntimeplot[forget plot]{./data/RandomTransitivelyClosed/bbkConditions/ns/0.75/stats_0.75.csv}{magenta}
			\addplot[
			domain=1:100,      
			samples=2,     
			color=blue,
			thick,
			]
			{2.2788e-08 * x^4.9777};
			\addlegendentry{$O(\lvert V\rvert^{5.0})$}
			
		\end{groupplot}
		
		\path (group c1r2.south west) -- node[below=0.75cm]{\ref{TCLNumberOfVerticesLegend}} (group c3r2.south east);
	\end{tikzpicture}

%% file: figure-twitter-joint-vs-boecker.tex
\tikzmath{\plotheight=3.45;}
\tikzmath{\plotwidth=3.45;}

\setlength{\hspacing}{0.25cm}
\setlength{\vspacing}{0.50cm}

\providecommand{\addzeropersistencyplot}{}	
\renewcommand{\addzeropersistencyplot}[1]{\addplot+[
	only marks,
	mark=*,
	mark options={scale=0.25},
	color=blue
	] table[
	col sep=comma,
	x expr=\thisrow{initialNumberOfVertices}, 
	y expr=\thisrow{persistentZeroEdges} / \thisrow{initialNumberOfEdges}*100
	] {#1};}

\providecommand{\addpersistencyplot}{}	
\renewcommand{\addpersistencyplot}[1]{\addplot+[
	only marks,
	mark=*,
	mark options={scale=0.25},
	color=blue
	] table[
	col sep=comma,
	x expr=\thisrow{initialNumberOfVertices}, 
	y expr=(\thisrow{persistentZeroEdges} + \thisrow{persistentOneEdges}) / \thisrow{initialNumberOfEdges}*100
	] {#1};}

\providecommand{\addonepersistencyplot}{}	
\renewcommand{\addonepersistencyplot}[1]{\addplot+[
	only marks,
	mark=*,
	mark options={scale=0.25},
	color=blue
	] table[
	col sep=comma,
	x expr=\thisrow{initialNumberOfVertices}, 
	y expr=\thisrow{persistentOneEdges} / \thisrow{initialNumberOfEdges}*100
	] {#1};}

\providecommand{\addruntimeplot}{}	
\renewcommand{\addruntimeplot}[2][]{\addplot+[
	only marks,
	mark=*,
	mark options={scale=0.25, color=black},
	#1
	] table[
	col sep=comma,
	x expr=\thisrow{initialNumberOfVertices}, 
	y expr=\thisrow{runtimeNanoseconds} / 1e9
	] {#2};}
	
	\tikzsetnextfilename{figure-boecker-twitter}
	\begin{tikzpicture}		
		\begin{groupplot}[group style={group size= 4 by 2, horizontal sep=\hspacing, vertical sep=\vspacing}]

			\nextgroupplot[
			title={Cor.~\ref{corollary:boecker-condition-strong}},
			ylabel={$\frac{\lvert \dom{\hat{x}}\rvert}{\lvert P_V\rvert} \left[\%\right]$},
			title style={align=center, font=\small, yshift=-1ex},
			xmin=0,
			xmax=250,
			xtick distance=100,
			ymin=0,
			ymax=100,
			width=\plotwidth cm,
			height=\plotheight cm
			]
			\addpersistencyplot{./data/twitter/bbkConditions/persistencies.csv};
			
			\nextgroupplot[
			width=0.5*\plotwidth cm,
			height=\plotheight cm,
			enlarge x limits=0.5,
			ymin=0,
			ymax=100,
			xtick=data,
			hide axis,
			xshift=-\hspacing
			]
			
			\addplot[
			boxplot prepared from table={
				table=./data/twitter/bbkConditions/boxplot-all.csv,
				lower whisker=lw,
				upper whisker=uw,
				lower quartile=lq,
				upper quartile=uq,
				median=med
			},
			boxplot prepared,
			boxplot/draw direction=y
			] coordinates {};

			\nextgroupplot[
			title={Joint},
			title style={align=center, font=\small, yshift=-1ex},
			xmin=0,
			xmax=250,
			xtick distance=100,
			ymin=0,
			ymax=100,
			width=\plotwidth cm,
			height=\plotheight cm,
			yticklabels=\empty
			]
			\addpersistencyplot{./data/twitter/allConditions/persistencies.csv};

			\nextgroupplot[
			width=0.5*\plotwidth cm,
			height=\plotheight cm,
			enlarge x limits=0.5,
			ymin=0,
			ymax=100,
			xtick=data,
			hide axis,
			xshift=-\hspacing
			]
			\addplot[
			boxplot prepared from table={
				table=./data/twitter/allConditions/boxplot-all.csv,
				lower whisker=lw,
				upper whisker=uw,
				lower quartile=lq,
				upper quartile=uq,
				median=med
			},
			boxplot prepared,
			boxplot/draw direction=y
			] coordinates {};

			\nextgroupplot[
			ylabel={Runtime [s]},
			ymode=log,
			xmode=log,
			xlabel={$\vert V\vert$},
			xmin=0,
			xmax=250,
			ymin=1e-5,
			ymax=1e5,
			width=\plotwidth cm,
			height=\plotheight cm,
			legend pos=south east,
			legend style={
				font=\scriptsize,
				draw=none,  
				fill opacity=0.9
			},
			legend image post style={xscale=0.25}
			]
			\addruntimeplot[forget plot]{./data/twitter/bbkConditions/persistencies.csv}
			\addplot[
			domain=1:250,      
			samples=2,     
			color=black,
			thick,
			]
			{6.9621e-08 * x^3.8245};
			\addlegendentry{$O(\lvert V\rvert^{3.8})$}

			\nextgroupplot[
			width=0.5*\plotwidth cm,
			height=\plotheight cm,
			hide axis
			]

			\nextgroupplot[
			ymode=log,
			xmode=log,
			xlabel={$\vert V\vert$},
			xmin=0,
			xmax=250,
			ymin=1e-5,
			ymax=1e5,
			width=\plotwidth cm,
			height=\plotheight cm,
			yticklabel=\empty,
			legend pos=south east,
			legend style={
				font=\scriptsize,
				draw=none,  
				fill opacity=0.9
			},
			legend image post style={xscale=0.25}
			]
			\addruntimeplot[forget plot]{./data/twitter/allConditions/persistencies.csv}
			\addplot[
			domain=1:250,      
			samples=2,     
			color=black,
			thick,
			]
			{1.7194e-08 * x^5.1053};
			\addlegendentry{$O(\lvert V\rvert^{5.1})$}

			\nextgroupplot[
			width=0.5*\plotwidth cm,
			height=\plotheight cm,
			hide axis
			]	
			
		\end{groupplot}
	\end{tikzpicture}

%% file: manuscript.bib
@article{groetschel-1984,
	author = {Martin Grötschel and Michael Jünger and Gerhard Reinelt},
	title = {A Cutting Plane Algorithm for the Linear Ordering Problem},
	year = {1984},
	journal = {Operations Research},
	volume = {32},
	number = {6},
	pages = {1195--1220},
	doi = {10.1287/opre.32.6.1195},
}

@book{marti-2011,
	author = {Mart\'{\i}, Rafael and Reinelt, Gerhard},
	title = {The Linear Ordering Problem. Exact and Heuristic Methods in Combinatorial Optimization},
	year = {2011},
	publisher = {Springer},
	doi={10.1007/978-3-642-16729-4}
}

@article{alush2012,
	author = {Alush, Amir and Goldberger, Jacob},
	doi = {10.1109/TPAMI.2011.280},
	journal = {Transactions on Pattern Analysis and Machine Intelligence},
	number = {10},
	pages = {1966--1977},
	title = {Ensemble Segmentation Using Efficient Integer Linear Programming},
	volume = {34},
	year = {2012}}

@inproceedings{lange2019,
	author = {Lange, Jan{-}Hendrik and Andres, Bjoern and Swoboda, Paul},
	booktitle = {Computer Vision and Pattern Recognition (CVPR)},
	doi = {10.1109/CVPR.2019.00625},
	title = {Combinatorial Persistency Criteria for Multicut and Max-Cut},
	year = {2019}
}

@inproceedings{lange2018,
	author = {Lange, Jan-Hendrik and Karrenbauer, Andreas and Andres, Bjoern},
	booktitle = {ICML},
	title = {Partial Optimality and Fast Lower Bounds for Weighted Correlation Clustering},
	url = {https://proceedings.mlr.press/v80/lange18a.html},
	year = {2018},
}

@phdthesis{shekhovtsov-2013,
	address = {Prague},
	author = {Shekhovtsov, Alexander},
	school = {Center for Machine Perception, Czech Technical University},
	supervisor = {Hlav{\'a}{\v c}, V{\'a}clav},
	title = {Exact and Partial Energy Minimization in Computer Vision},
	year = {2013}
}

@inproceedings{shekhovtsov-2014,
	author = {Shekhovtsov, Alexander},
	booktitle = {CVPR},
	doi = {10.1109/CVPR.2014.152},
	title = {Maximum Persistency in Energy Minimization},
	year = {2014},
}

@inproceedings{shekhovtsov-2015,
	author = {Shekhovtsov, Alexander and Swoboda, Paul and Savchynskyy, Bogdan},
	booktitle = {CVPR},
	doi = {10.1109/CVPR.2015.7298650},
	title = {Maximum Persistency via Iterative Relaxed Inference with Graphical Models},
	year = {2015},
}

@article{goldberg-1988,
	author = {Goldberg, Andrew V. and Tarjan, Robert E.},
	doi = {10.1145/48014.61051},
	journal = {Journal of the ACM},
	number = {4},
	pages = {921--940},
	title = {A New Approach to the Maximum-Flow Problem},
	volume = {35},
	year = {1988}}

@InProceedings{stein2023,
	title = 	 {Partial Optimality in Cubic Correlation Clustering},
	author =       {Stein, David and Di Gregorio, Silvia and Andres, Bjoern},
	booktitle = {ICML},
	year = 	 {2023},
	url={https://proceedings.mlr.press/v202/stein23a.html}
}

@InProceedings{stein2024,
	title = 	 {Partial Optimality in the Linear Ordering Problem},
	author =       {Stein, David and Andres, Bjoern},
	booktitle = {ICML},
	year={2024},
	url = 	 {https://proceedings.mlr.press/v235/stein24a.html}
}

@misc{gurobi,
	author = {{Gurobi Optimization, LLC}},
	title = {{Gurobi Optimizer Reference Manual}},
	year = 2023,
	url = "https://www.gurobi.com"
}

@inproceedings{leskovec2012,
	author = {Leskovec, Jure and Mcauley, Julian},
	booktitle = {NeurIPS},
	title = {Learning to Discover Social Circles in Ego Networks},
	url = {https://proceedings.neurips.cc/paper_files/paper/2012/file/7a614fd06c325499f1680b9896beedeb-Paper.pdf},
	year = {2012}
}

@InProceedings{hornakova2017,
	title = 	 {Analysis and Optimization of Graph Decompositions by Lifted Multicuts},
	author =       {Andrea Hor{\v{n}}{\'a}kov{\'a} and Jan-Hendrik Lange and Bjoern Andres},
	booktitle = 	 {ICML},
	year = 	 {2017},
	url = 	 {https://proceedings.mlr.press/v70/hornakova17a.html}
}

@article{boykov2001,
	author={Boykov, Y. and Veksler, O. and Zabih, R.},
	journal={IEEE Transactions on Pattern Analysis and Machine Intelligence}, 
	title={Fast approximate energy minimization via graph cuts}, 
	year={2001},
	volume={23},
	number={11},
	pages={1222-1239},
	doi={10.1109/34.969114}
}

@article{boecker2009,
	author={B{\"o}cker, Sebastian
	and Briesemeister, Sebastian
	and Klau, Gunnar W.},
	title={On optimal comparability editing with applications to molecular diagnostics},
	journal={BMC Bioinformatics},
	year={2009},
	volume={10},
	number={1},
	pages={S61},
	issn={1471--2105},
	doi={10.1186/1471-2105-10-S1-S61}
	}

@article{irmai2025,
	author={Jannik Irmai and Maximilian Moeller and Bjoern Andres},
	title={Algorithms for the preordering problem and their application to the task of jointly clustering and ordering the accounts of a social network},
	journal={Transactions on Machine Learning Research},
	year={2026},
	url={https://openreview.net/forum?id=cBsUnv7Cb3},
}

@article{wakabayashi1998,
	author={Yoshiko Wakabayashi},
	title={The Complexity of Computing Medians of Relations}, 
	volume={3},  
	doi={10.11606/resimeusp.v3i3.74876}, 
	number={3},
	journal={Resenhas do Instituto de Matemática e Estatística da Universidade de São Paulo}, 
	year={1998}, 
	pages={323--349}
}

@article{weller2012,
	title = {On making directed graphs transitive},
	journal = {Journal of Computer and System Sciences},
	volume = {78},
	number = {2},
	pages = {559--574},
	year = {2012},
	doi = {https://doi.org/10.1016/j.jcss.2011.07.001},
	author = {Mathias Weller and Christian Komusiewicz and Rolf Niedermeier and Johannes Uhlmann}
}

@article{jacob2008,
	author = {Jacob, Juby and Jentsch, Marcel and Kostka, Dennis and Bentink, Stefan and Spang, Rainer},
	title = {Detecting hierarchical structure in molecular characteristics of disease using transitive approximations of directed graphs},
	year = {2008},
	address = {USA},
	volume = {24},
	number = {7},
	doi = {10.1093/bioinformatics/btn056},
	journal = {Bioinformatics},
	pages = {995--1001},
	numpages = {7}
}

@article{bansal-2004,
    author = {Bansal, Nikhil and Blum, Avrim and Chawla, Shuchi},
    title = {Correlation Clustering},
    journal = {Machine Learning},
    year = {2004},
    volume = {56},
    number = {1},
    pages = {89--113},
    doi = {10.1023/B:MACH.0000033116.57574.95},
}

@article{muller1996partial,
  title={On the partial order polytope of a digraph},
  author={M{\"u}ller, Rudolf},
  journal={Mathematical Programming},
  volume={73},
  number={1},
  pages={31--49},
  year={1996},
  publisher={Springer},
  doi={10.1007/BF02592097}
}

@phdthesis{gurgel1992,
  title={Poliedros de grafos transitivos},
  author={Gurgel, Maria Angela},
  year={1992},
  school={Universidade de S{\~a}o Paulo},
  note={doi: 10.11606/T.45.1992.tde-20210728-191547}
}

@article{gurgel1997adjacency,
  title={Adjacency of vertices of the complete pre-order polytope},
  author={Gurgel, Maria Angela and Wakabayashi, Yoshiko},
  journal={Discrete Mathematics},
  volume={175},
  number={1--3},
  pages={163--172},
  year={1997},
  doi={10.1016/S0012-365X(96)00143-4},
}

@misc{stein-code-2026,
	author={Stein, David},
	title={Partial Optimality in the Preordering Problem -- Code},
	year={2026},
	url={https://github.com/dsteindd/partial-optimality-in-the-preordering-problem}
}
